\documentclass[11pt]{article}
\usepackage[margin=1in]{geometry}
\usepackage{datetime}
\usepackage{enumitem}
\usepackage{pifont}
\usepackage{wrapfig}
\usepackage{subcaption}
\usepackage[T1]{fontenc}
\usepackage[utf8]{inputenc}
\usepackage{pgfplots}
\usepackage{amsmath, amssymb}
\usepackage{amsthm}
\usepackage{color}
\usepackage{cases}
\usepackage{xparse}
\usepackage{xargs}
\usepackage{appendix}
\usepackage[linesnumbered,lined,boxed,noend]{algorithm2e}
\usepackage{algpseudocode}
\usepackage{verbatim, xspace}
\usepackage{tikz}
\usepackage{mathtools}
\usepackage{framed}
\usepackage{tikz}

\usepackage{titlesec}

\usetikzlibrary{decorations.pathreplacing,angles,quotes}
\usetikzlibrary{arrows}

\usepackage{hyperref}
\usepackage{pgfplots}

\usepackage{xcolor}

\newcommand{\citet}[1]{\cite{#1}}

\usepackage[colorinlistoftodos,prependcaption,textsize=tiny]{todonotes}
\newcommandx{\unsure}[2][1=]{\todo[linecolor=green,backgroundcolor=green!25,bordercolor=green,#1]{\normalsize #2}}
\newcommandx{\improvement}[2][1=]{\todo[inline,linecolor=blue,backgroundcolor=blue!05,bordercolor=blue,#1]{\normalsize #2}}
\newcommandx{\info}[2][1=]{\todo[linecolor=yellow,backgroundcolor=yellow!25,bordercolor=yellow,#1]{#2}}
\newcommandx{\floatmodel}[2][1=]{\todo[inline,linecolor=red,backgroundcolor=yellow!25,bordercolor=yellow,#1]{#2}}
\newcommandx{\thiswillnotshow}[2][1=]{\todo[disable,#1]{#2}}
\newcommandx{\celine}[2][1=]{\todo[inline,linecolor=green,backgroundcolor=green!25,bordercolor=green,caption={\normalsize \textbf{Celine}},#1]{\normalsize #2}}
\newcommandx{\karol}[2][1=]{\todo[inline,linecolor=blue,backgroundcolor=blue!25,bordercolor=blue,caption={\normalsize \textbf{Karol}},#1]{\normalsize #2}}
\newcommandx{\jesper}[2][1=]{\todo[inline,linecolor=red,backgroundcolor=red!25,bordercolor=red,caption={\normalsize \textbf{Jesper}},#1]{\normalsize #2}}
\newcommandx{\kuba}[2][1=]{\todo[inline,linecolor=gray,backgroundcolor=red!25,bordercolor=red,caption={\normalsize \textbf{Kuba}},#1]{\normalsize #2}}
\newtheorem{theorem}{Theorem}
\newtheorem{definition}[theorem]{Definition}
\newtheorem{lemma}[theorem]{Lemma}
\newtheorem{corollary}[theorem]{Corollary}
\newtheorem{claim}[theorem]{Claim}

\newtheorem{observation}[theorem]{Observation}

\makeatletter
\newtheorem*{rep@theorem}{\rep@title}
\newcommand{\newreptheorem}[2]{%
	\newenvironment{rep#1}[1]{%
		\def\rep@title{#2 \ref{##1}}%
		\begin{rep@theorem}[restated]}%
		{\end{rep@theorem}}}
\newreptheorem{theorem}{Theorem}
\newreptheorem{lemma}{Lemma}	
\makeatother

\numberwithin{theorem}{section}

\newcommand{\eps}{\varepsilon}

\newcommand{\Oh}{\mathcal{O}}
\newcommand{\Os}{\Oh^{\star}}
\newcommand{\Otilde}{\widetilde{\Oh}}
\newcommand{\Ot}{\Otilde}

\newcommand{\nat}{\mathbb{N}}
\newcommand{\N}{\mathbb{N}}

\newcommand{\real}{\mathbb{R}_{+}}

\newcommand{\poly}{\mathrm{poly}}
\newcommand{\polylog}{\,\textup{polylog}}

\newcommand{\Ex}[1]{\mathbb{E}\left[ #1 \right]}

\newcommand{\Uni}{\mathrm{Uni}}
\newcommand{\Bin}{\mathrm{Bin}}

\usepackage{footmisc}

\newcommand{\prun}{\mathtt{crit}} 
\newcommand{\setrng}[1]{\{1,\ldots,#1\}} 
\newcommand{\coverprod}{\circledast}
\newcommand{\DP}{\mathtt{DP}}

\renewcommand{\leq}{\leqslant}
\renewcommand{\geq}{\geqslant}
\renewcommand{\le}{\leqslant}
\renewcommand{\ge}{\geqslant}

\usepackage{framed}
\usepackage[framemethod=TikZ]{mdframed}

\newenvironment{case}
{\mdfsetup{%
		nobreak=true,
		middlelinecolor=gray,
		middlelinewidth=1pt,
		backgroundcolor=gray!10,
		innertopmargin=7pt,
		roundcorner=5pt}
	\begin{mdframed}}
	{\end{mdframed}}

\newcommand{\cS}{\ensuremath{\mathcal{S}}}
\newcommand{\cW}{\ensuremath{\mathcal{W}}}

\newcommand{\bbN}{\ensuremath{\mathbb{N}}}

\newcommand{\dc}{{\downarrow}}
\newcommand{\uc}{{\uparrow}}

\newcounter{openquestion}
\newenvironment{openquestion}
{\vspace{-0.4em}\begin{center}\begin{minipage}{\textwidth}\begin{framed}\refstepcounter{openquestion}\vspace{-0.3em}\textbf{Question~\theopenquestion:}}	{\vspace{-0.4em}\end{framed}\end{minipage}\end{center}\vspace{-0.4em}}

\title{A Faster Exponential Time Algorithm for Bin Packing With a Constant Number of Bins via Additive Combinatorics \footnote{An extended abstract of this manuscript was presented and included in the proceedings of the 2021 ACM-SIAM Symposium on Discrete Algorithms.}}
\date{}

\author{
    Jesper Nederlof\footnote{Utrecht University, The
    Netherlands, \texttt{j.nederlof@uu.nl}. Supported by
    the project CRACKNP that has received funding from the European
    Research Council (ERC) under the European Union’s Horizon 2020 research and
    innovation programme (grant agreement No 853234).}
    \and
    Jakub Pawlewicz\footnote{Institute of Informatics, University of
    Warsaw, Poland, \texttt{pan@mimuw.edu.pl}. }
    \and
    Céline M. F. Swennenhuis\footnote{Eindhoven University of Technology, The
    Netherlands, \texttt{c.m.f.swennenhuis@tue.nl}. Supported by the Netherlands
    Organization for Scientific Research under project no. 613.009.031b.}
    \and
    Karol W\k{e}grzycki\footnote{Saarland University and Max Planck Institute for Informatics,
        Saarbr\"ucken, Germany, \texttt{wegrzycki@cs.uni-saarland.de}.  
    This work is part of the project TIPEA that has
    received funding from the European Research Council (ERC) under the European Unions Horizon
    2020 research and innovation programme (grant agreement No. 850979).
    Author was also supported Foundation for Polish Science (FNP), by the grants
    2016/21/N/ST6/01468 and 2018/28/T/ST6/00084 of the Polish National Science
    Center and project TOTAL that has received funding from the European
    Research Council (ERC) under the European Union’s Horizon 2020 research and
    innovation programme (grant agreement No 677651).}
}

\begin{document}
\maketitle
\begin{abstract}

    In the Bin Packing problem one is given $n$ items with weights
    $w(1),\ldots,w(n)$ and $m$ bins with capacities $c_1,\ldots,c_m$. The goal is
    to find a partition of the items into sets $S_1,\ldots,S_m$ such that
    $w(S_j) \leq c_j$ for every bin $j$, where $w(X)$ denotes $\sum_{i \in
    X}w(i)$.

    Bj{\"{o}}rklund, Husfeldt and Koivisto (SICOMP 2009) presented an $\Os(2^n)$ time algorithm for
    Bin Packing (the $\Os(\cdot)$ notation omits factors polynomial in the input size). In this paper, we show that for every $m \in \nat$ there exists
    a constant $\sigma_m >0$ such that an instance of Bin Packing with $m$ bins can be solved
    in $\Oh(2^{(1-\sigma_m)n})$ randomized time. 
    Before our work, such improved algorithms were not known even for $m$ equals
    $4$.
	
	A key step in our approach is the following new result in Littlewood-Offord
    theory on the additive combinatorics of subset sums: For every $\delta >0$
    there exists an $\eps >0$ such that if $|\{ X\subseteq \{1,\ldots,n \} :
    w(X)=v \}| \geq 2^{(1-\eps)n}$ for some $v$ then $|\{ w(X): X \subseteq \{1,\ldots,n\} \}|\leq 2^{\delta n}$.

\end{abstract}

\thispagestyle{empty}

 \begin{picture}(0,0)
 \put(462,-145)
 {\hbox{\includegraphics[width=40px]{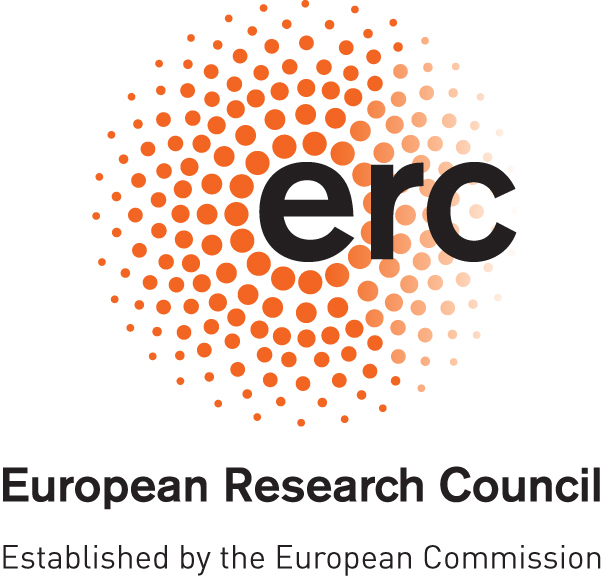}}}
 \put(452,-205)
 {\hbox{\includegraphics[width=60px]{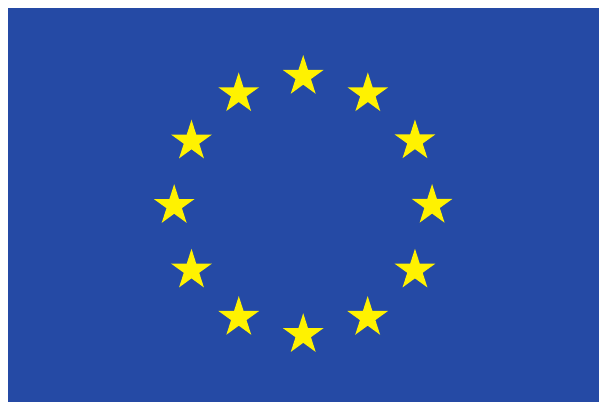}}}
 \end{picture}

\clearpage
\setcounter{page}{1}

\section{Introduction}\label{sec:intro}
A central aim in contemporary algorithm design is to minimize the worst-case complexity of an algorithm for a given (supposedly) hard computational problem in a fine-grained sense. 
The underlying goal is to reveal the optimal running time witnessed by (1) an algorithm with worst-case complexity $T(n)$ on instances with parameter $n$, and (2) a lower bound that excludes improvements to $T(n)^{1-\eps}$ time for some constant $\eps >0$.
For some problems, it is an especially intriguing question whether natural running times of the basic algorithms solving them are optimal.
One of the most important instances of such a question for an NP-complete
problem is about improvements over a relatively direct dynamic programming algorithm for Set Cover: 
\begin{openquestion}\label{q1}
	Can Set Cover with $n$ elements be solved in $\Os((2-\eps)^n)$ time, for some $\eps >0$?
\end{openquestion}
Here and throughout the paper, we use the $\Os$ notation to hide factors polynomial in the input size.\footnote{In principle, it is natural to assume the Set Cover instance has $n$ elements and $\poly(n)$ sets, but an algorithm by Bj\"orklund et al.~\cite{BjorklundHK09} solves Set
Cover instances in $2^nn^{\Oh(1)}$ time irrespective of the number of sets.}
Unfortunately, Question~\ref{q1} seems to have a fate similar to the Strong
Exponential Time Hypothesis (which is about a similar improvement for the CNF-SAT problem): While there is an increasing interest and dependence on its validity (see e.g.~\cite{subset-sum-lower, krauthgamer_et_al:LIPIcs:2019:10284}), we seem to be far from resolving it.

Therefore, it is natural to study Question~\ref{q1} for special cases of Set Cover. And indeed, improved algorithms of the type asked in Question~\ref{q1} were already presented for instances with small sets~\cite{Koivisto09},  (more generally) large solutions~\cite{Nederlof16}, and for several other cases (see e.g.~\cite{GolovnevKM16}).

However, some of the most fundamental NP-complete problems that are special cases of Set Cover such as Graph Coloring and Directed Hamiltonicity\footnote{Krauthgamer and Trabelsi~\cite{krauthgamer_et_al:LIPIcs:2019:10284} rewrite a Directed Hamiltonicity instance efficiently as a Set Cover instance.} still defy considerable research efforts to obtain the type of improved algorithms asked for in Question~\ref{q1} (see e.g.~\cite{BjorklundKK17,FominK13}).

\paragraph{Bin Packing}
We study one such fundamental NP-complete problem, the \emph{Bin Packing problem}:
Given item weights $w(1),\ldots,w(n) \in \nat$ and capacities
$c_1,\ldots,c_m \in \nat$, can we partition items $\{1,\ldots,n\}$ into $m$ sets $S_1, \ldots,
S_m \subseteq \{1,\ldots,n\}$ such that $w(S_j)\leq c_j$ for each $j \in
\{1,\ldots,m\}$? Here $w(X)$ denotes $\sum_{i \in X}w(i)$.
Due to its elegant formulation and clear practical applicability, Bin Packing is
a central problem in computer science. For example, it models the
most basic non-trivial scheduling problem with multiple machines. While Bin
Packing has been extensively studied from an approximation and online algorithms perspective~\cite{CoffmanJr.2013}, much less research has been devoted to \emph{exact algorithms} for Bin Packing.

The currently fastest algorithm for Bin Packing is a consequence\footnote{Assuming the capacity of each bin equals $c$, create a Set Cover instance with all item sets of weight at most $c$.\label{footnote}}  of the aforementioned algorithm for Set Cover from~\cite{BjorklundHK09}, and it runs in $\Os(2^n)$ time. With Question~\ref{q1} on the horizon, we ask whether this can be improved:

\begin{openquestion}\label{q2}
	Can Bin Packing with $n$ items be solved in $\Oh((2-\eps)^n)$ time, for some $\eps >0$?
\end{openquestion}

The only improvement over the $\Os(2^n)$ time algorithm for Bin Packing
is due to Lente et al.~\cite{LenteLST13}, who
gave an $\Os(m^{n/2})$ time algorithm. Note that this is only an improvement for $m=2,3$
bins and Question~\ref{q2} remained illusive for $m = 4$ already.
In stark contrast, our main result is an improvement over the $\Os(2^n)$ time algorithm for \emph{every} constant number of bins:

\begin{theorem}[Main Theorem]
	\label{mainthm}
	For every $m \in \nat$ there is a constant $\sigma_m >0$ such that 
    every Bin Packing instance with $m$ bins can be solved in
    $\Oh(2^{(1-\sigma_m)n})$ time with high probability.
\end{theorem}

While our algorithm does not resolve Question~\ref{q2}, we believe it makes
substantial progress on it because (1) Set Cover with a
constant-sized solution is as least as hard as a general Set Cover, and (2) the other
extreme, Set Cover with a linear number sets in the solution (and hence Bin Packing with a linear number of bins with equal capacity\textsuperscript{\ref{footnote}}), can be solved in
$\Oh((2-\eps)^n)$ time (see~\cite{Nederlof16}).

\subsection{Our Approach for Proving Theorem~\ref{mainthm}}

As our starting point, we extend the methods from \cite{BjorklundHKK09, Nederlof16} to show
that instances of Bin Packing with the following restrictions admit an $\Os(2^{(1-\sigma_m)n})$ time randomized algorithm for some
$\sigma_m > 0$:

\begin{description}
    \item[Restriction~(1)] the instance has a \emph{low concentration} in the sense that
        $\beta(w) \leq 2^{(1-\eps)n}$ for some $\eps >0$, where $\beta(w)$ is the
        \emph{maximum frequency} $\max_{v}|\{ X\subseteq \{1,\ldots,n \} : w(X)=v \}|$, and
    \item[Restriction~(2)] the instance is \emph{tight} in the sense that
        $\sum_{j=1}^m c_j =w(\setrng{n})$.
\end{description}

\newcommand{\re}[1]{Restriction~(#1)}

Fix a set of bins $L \subseteq \{1,\ldots,m\}$ and recall  $(S_1,\ldots,S_m)$ denotes a solution.
The crux of \re{1} and \re{2} is that together they imply that the number of
distinct sets $S^L:= \bigcup_{j \in L}S_j$
is at most $2^{(1-\eps)n}$ since $w(S^L)=\sum_{j \in L} c_j$. We explain in \S~\ref{subsubsec:zetmob} how this allows a faster algorithm via the methods
of~\cite{BjorklundHKK09, Nederlof16}. In the nutshell, these sets $S^L$
correspond to the \emph{candidates} for the solutions that we need to check and
bounding this number automatically corresponds to the running time of the
algorithm.

However, extending this algorithm to an improved algorithm that solves \emph{all} instances with a constant number of bins requires both new combinatorial (for relaxing \re{1}) and new algorithmic (for relaxing \re{2}) insights that are our main contributions. Therefore we first discuss these insights.
\titleformat{\subsubsection}
{\normalfont\bfseries}{\thesubsubsection}{1em}{}
\subsubsection{Combinatorial Ideas: Lifting \re{1} via Littlewood--Offord Theory.}\label{subsubsec:liftr1}
Our main combinatorial contribution is a new structural insight on instances that do not satisfy \re{1}, i.e. vectors $w$ with $|\{ X\subseteq \{1,\ldots,n \} : w(X)=v \}| \geq 2^{(1-\eps)n}$ for some $v$ and $\eps >0$. 

The challenge of determining the structure of such vectors $w$ is well-known in additive
combinatorics as the \emph{Littlewood--Offord Problem}.
Its rich theory has found applications ranging from pure mathematics (such as
estimating the singularity of random Bernoulli matrices~\cite{annals} or zeroes
of random polynomials~\cite{littlewood1943number}), to database
security~\cite{Griggs1998DatabaseSA}, and to computational complexity theory~\cite{DiakonikolasS13,KaneW16, MekaNV16}. See also the designated chapter in the standard textbook on additive combinatorics~\cite{0020358}.
However, whereas most works (with notable exceptions being
e.g.~\cite{halasz1977estimates,RUDELSON2008600}) assumed inversely
\emph{polynomially} small concentration, e.g. $\beta(w) \geq 2^n / \poly(n)$,
\re{1} is about inversely \emph{exponentially} small concentration.

Recent work studied such exponentially small concentration with applications to improved exponential time algorithms for the Subset Sum problem~\cite{stacs2015, BansalGN018}.
Specifically, they studied trade-off between the parameters $\beta(w)$ and
$|w(2^{\setrng{n}})|:= |\{w(X) : X \subseteq \setrng{n}\}|$.
 Two extremal cases are:
\begin{align*}
    \text{If } w_a :=& (0,0,\ldots,0) &\text{ then }&&  |w_a(2^{\setrng{n}})| = 1 \text{ and } \beta(w_a) = 2^n\\
    \text{If } w_b :=& (1,2,\ldots,2^{n-1}) &\text{ then }&& |w_b(2^{\setrng{n}})| = 2^n \text{ and } \beta(w_b) = 1
\end{align*}
One may suspect that all vectors $w \in \mathbb{Z}^n$ are a combination of these
two extremes and therefore that a smooth trade-off between $\beta(w)$ and
$w(2^{\setrng{n}})$ can be proved.
In the case $w \in \mathbb{F}^{n\times n}_2$ (where
$|w(2^{\setrng{n}})|\beta(w)=2^n$), this suspicion can be confirmed.\footnote{Here $w(i)$ should be interpreted as the the $i$'th row of $w$, so it is a $n$-dimensional binary vector for every $i$. Then by the rank-nullity theorem $|w(2^{\setrng{n}})|=2^{\mathrm{rk}_2(w)}$ and
	$\beta(w)=2^{n-\mathrm{rk}_2(w)}$, where $\mathrm{rk}_2$ is the rank
	over $\mathbb{F}^n_2$.}
Observe that a similar trade-off for $w \in \mathbb{Z}^n$ would allow us to lift
\re{1} by a simple $\Os(|w(2^{\setrng{n}})|^m)$ time algorithm for Bin Packing (Lemma~\ref{lem:CaseA}).

Unfortunately, this intuition is not true and the case $w \in \mathbb{Z}^n$ is
far more subtle. For instance, Wiman~\cite{wiman2017improved} showed in his
remarkable bachelor thesis that, surprisingly, vectors satisfying simultaneously
both $|w(2^{\setrng{n}})| \geq 2^{(1-\eps)n}$ and $\beta(w)\geq 2^{0.2563n}$ exist for any $\eps >0$.
Our main combinatorial contribution is a proof that instances with the same parameters
but the roles of $\beta(w)$ and $|w(2^{\setrng{n}})|$ swapped do \emph{not}
exist:\footnote{See Section~\ref{sec:prel} for the formal definition of the
    $\Oh(\cdot)$,$\Os(\cdot)$
and $\Oh_{\eps \rightarrow 0}(\cdot)$ notation.}:

\begin{theorem}
	\label{lothm}
    Let $\eps >0$. If $\beta(w) \geq 2^{(1-\eps)n}$, then $|w(2^{\setrng{n}})|\leq
    2^{\delta n}$, where $\delta= \delta(\eps) = \Oh_{\eps \rightarrow
    0}\left(\frac{\log\log(1/\eps)}{\sqrt{\log(1/\eps)}}\right)$.
\end{theorem}

The dependency of $\delta$ on $\eps$ was recently improved to $\delta(\eps) = \Oh(\sqrt{\eps})$ by Jain et al.~\cite{jain2021anticoncentration}. 
The previous best bounds were given by Austrin et al.~\cite{stacs2015} who found
a connection with \emph{Uniquely Decodable Code Pairs (UDCPs)} from information
theory (see Subsection~\ref{subsec:relwork} for details).  
This implies for example that if $\beta(w) \geq 2^{(1-\eps)n}$, then
$|w(2^{\setrng{n}})| \leq 2^{0.4228n+\sqrt{\eps}}$ by a result on UDCPs from~\cite{AustrinKKN18}.
However, the reduction from~\cite{stacs2015} is symmetric with respect to
swapping the roles of $\beta(w)$ and $w(2^{\setrng{n}})$, and thus by the result from~\cite{wiman2017improved}
UDCP techniques alone are not enough to decrease the constant $0.4228$ beyond $0.2563$.

Therefore, we need new ideas to reduce the constant $0.4228$ to an arbitrarily small one. To do so, we first investigate the combinatorial structure of the hyperplane $H := \{x \in \mathbb{Z}^n : \langle w,x \rangle = v \}$, assuming $|H \cap \{0,1\}^n| \geq 2^{(1-\eps)n}$. Afterwards we apply an argument similar to the UDCP connection from~\cite{stacs2015}. 
We formally describe our approach for proving Theorem~\ref{lothm} in
Section~\ref{sec:LO}.

Note that Theorem~\ref{lothm} enables us to lift~\re{1}: We may assume $\beta(w)
\leq 2^{(1-\eps_m)n}$ where $\eps_m >0$ depends on $m$ since otherwise a simple
dynamic programming algorithm that runs in  $\Oh(n \cdot m \cdot |w(2^{\setrng{n}})|^m)$ time will be fast enough for
constant $m$ (see
Lemma~\ref{lem:CaseA}).

\subsubsection{New Algorithmic Ideas: Lifting~\re{2}}\label{subsubsec:liftr2}

As mentioned before, \re{2} is algorithmically useful because of the following reason: We aim
to detect a solution $S_1,\ldots,S_m$ to the Bin Packing instance by listing all
candidates for $S^L:=\bigcup_{j \in L}S_j$ for some $L\subseteq \{1,\ldots,m\}$, and \re{2} implies
that $w\big(S^L\big)=\sum_{j \in L}c_j$. This allows us to narrow down the number of candidates to $2^{(1-\eps)n}$ by \re{1} (we explain in \S\ref{subsubsec:zetmob} why this is useful). 
Note this even narrows down the number of candidates for $S^L$ if all bins have polynomially bounded \emph{slack}, 
i.e.,  $c_j - w(S_j) \leq \poly(n)$ since the number of possibilities of
$w(S^L)$ is only $\poly(n)$ as $m=\Oh(1)$.

But generally this strategy does not work whenever a bin has a large
\emph{slack}, that is when $c_j - w(S_j)$ is large. While reductions in several
similar situations were able to turn inequalities into equalities via general
rounding techniques (such as~\cite{NederlofLZ12,VVWW10}), we need a more sophisticated method in this paper to deal with this issue: The idea of~\cite{NederlofLZ12} is to divide the weights by roughly $c_j - w(S_j)$ and (conservatively) round to an integer. In this case, the bin $j$ has small slack with respect to the rounded weight function.
The major complication however is that for different bins we would then need to
work with differently rounded weight functions, which still does not allow us to
narrow down the number of options for $w(S^L)$ and
hence (via \re{1}) the number of candidates for $S^L$.

Instead, for an integer $k$ we work with a rounded version $w_k$ of weights $w$
where $w_k(i)$ is obtained from $w(i)$ by only keeping the $k$ most significant
bits.  We will show we can choose integer $\prun$ (which we call \emph{critical
pruner}) such that $|w_\prun(2^{\setrng{n}})|\approx 2^{\delta n}$, for some parameter
$\delta$ that depends on $m$. We will deal with the bins in two different ways
depending on whether it has large slack (i.e. its slack is at least approximately $n \cdot
2^{\ell-\prun}$, assuming all weights are $\ell$-bit integers) or not:

\begin{itemize}
	\item \textbf{Large Slack Bins:} Our approach for such bins is loosely inspired
        by rounding approximation algorithms, e.g. the FPTAS for Knapsack (see e.g.~\cite[Section 11.8]{0015106}).
        Observe that if some bin has large slack, we can split it into two parts, and we only need to keep track of the rounded weight of these parts in
        order to verify whether they indeed jointly fit into the bin. Because we
        assumed the \emph{upper bound} $|w_\prun(2^{\setrng{n}})| \le 2^{\delta n}
        \cdot \poly(n)$ we can afford
        to keep track of all combinations of rounded weights as long as $\delta < 1/m$.
        
    \item \textbf{Small Slack Bins:} We deal with all small slack bins jointly by considering a split of the bins $(L,R)$ such that and all bins in $L$ have small slack. Now we use the \emph{lower bound} $|w_\prun(2^{\setrng{n}})| \ge 2^{\delta n} /\poly(n)$ and our additive
        combinatorics result guarantees $\beta(w_\prun) \leq
        2^{(1-\eps(\delta))n}$ for some $\eps(\delta) > 0$.  Now, we use the
        fact that all bins have small slack. Note, that there are only
        $\poly(n)$ candidates for $w_\prun(S^L)$ and therefore
        there are at most $\Os(2^{(1-\eps(\delta))n})$ candidates for
        $S^L$, which can be algorithmically exploited.
\end{itemize}

In this informal discussion, we omitted several nontrivial technical issues.
In particular, to deal with instances with \emph{both} a substantial
number of small slack bins and large slack bins, we need to distinguish several additional cases. Due to the subtle technical issues, we need to deal with each one of
them in slightly different ways. Details are postponed to
Section~\ref{sec:main}.

\subsubsection{Solving Instances that Satisfy~\re{1} and~\re{2}.}\label{subsubsec:zetmob}

We now discuss how the methods from~\cite{BjorklundHKK09, Nederlof16} can be
used to solve all instances that satisfy \re{1} and \re{2} in $\Os(
2^{(1-\sigma_m)n})$ time for some $\sigma_m > 0$.
An important subroutine from~\cite{BjorklundHKK09} is an algorithm that, given a
set family $\cW \subseteq 2^{\setrng{n}}$ and set of bins $L$, computes for all $W \in
\cW$ whether the items in $W$ can be divided among the bins in $L$. That is,  it
computes whether $W$ can be a candidate for $S^L=\bigcup_{j \in L}S_j$.
The running time of this algorithm is $\Oh(|\dc\cW|n)$, where $\dc\cW := \{X \subseteq W : W \in \cW\}$
is defined as the \emph{down-closure} of $\cW$. The analogous \emph{up-closure}
of all supersets of elements from $\cW$ is denoted with $\uc \cW$.
 Let us fix a solution $(S_1,\ldots,S_m)$. We consider two cases based on how
 `balanced'\footnote{The actual definition of $\alpha$-balancedness
 (Definition~\ref{def:abalsol}) will be independent of the ordering of the
 bins.} a solution is, with respect to a small parameter $0 \le \alpha \le 1/2$: 

\begin{figure}[t!]
	\centering
	\begin{subfigure}[b]{0.45\textwidth}
		\includegraphics[width=\linewidth]{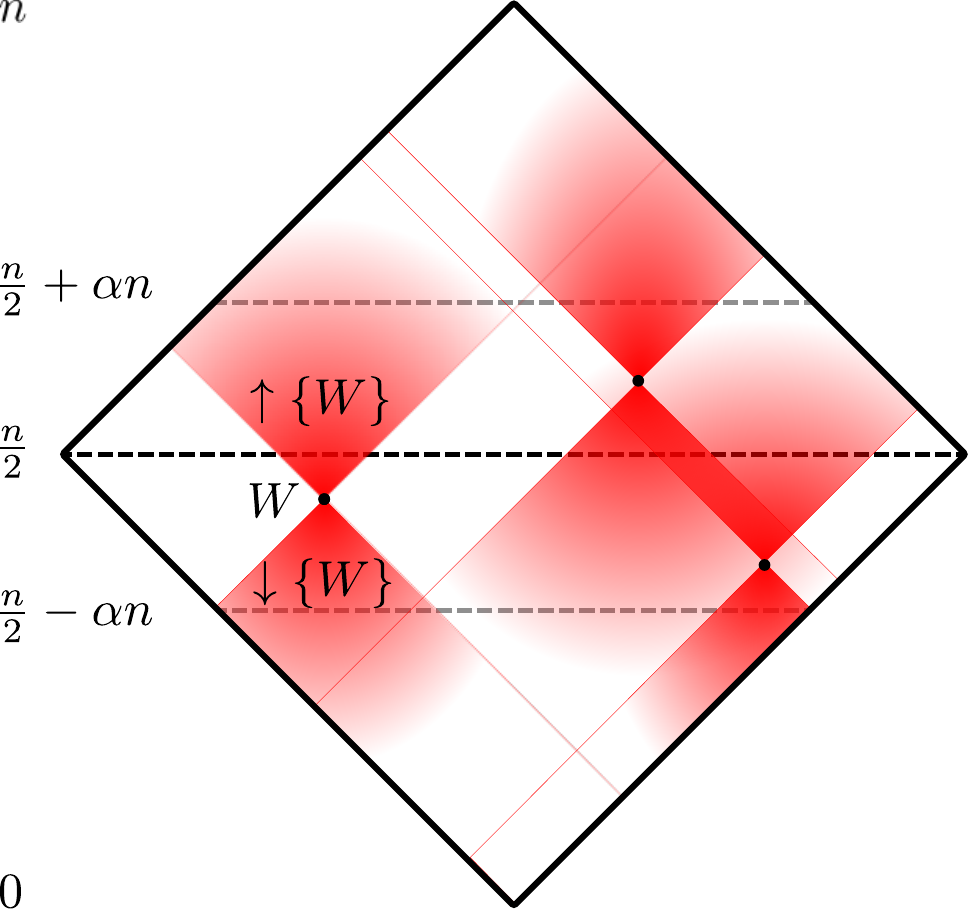}
	\end{subfigure}
    \hspace*{1cm}
	\begin{subfigure}[b]{0.45\textwidth}
		\includegraphics[width=\linewidth]{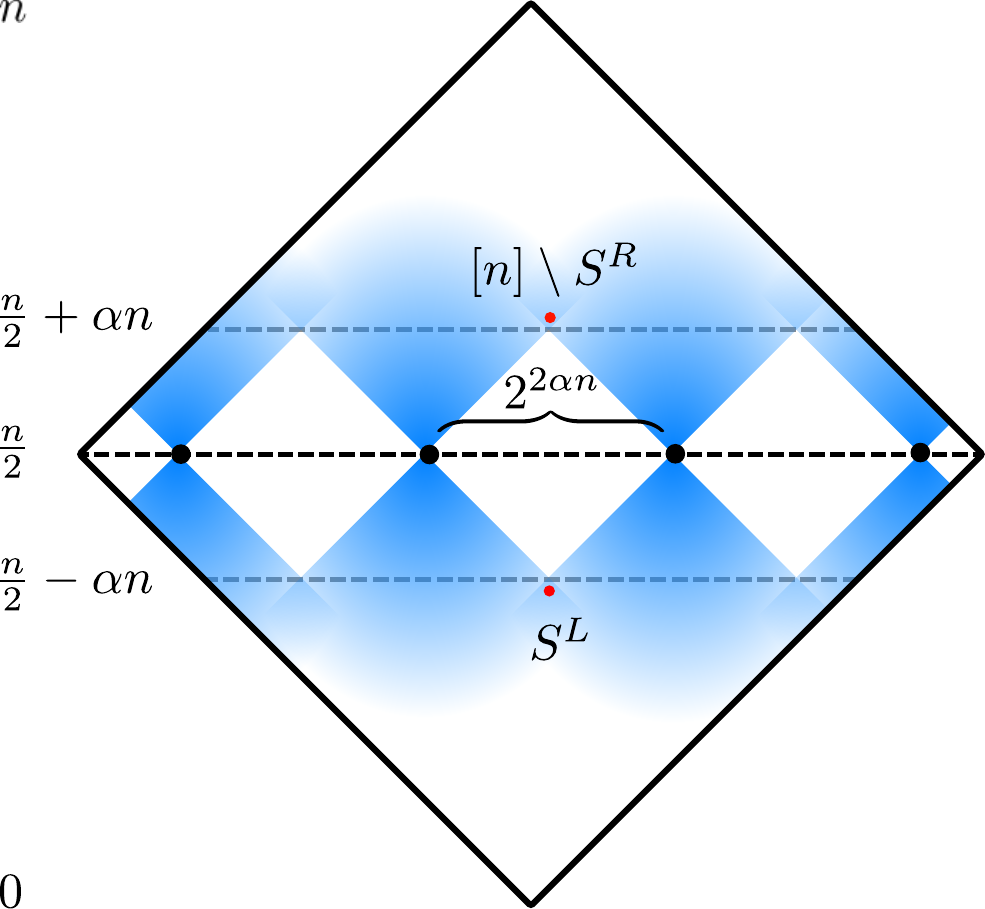}
	\end{subfigure}
	\caption{Schematic view of the algorithm from~\S\ref{subsubsec:zetmob} . A point in the square represents a set in
        $2^{\setrng{n}}$. The vertical axis corresponds to the cardinality of this set
        (e.g., longest horizontal line represents all sets in $\binom{\setrng{n}}{n/2}$).
		The left figure illustrates the analysis for
		the case when there exists an $\alpha$-balanced
        solution $W \subseteq \setrng{n}$ with $|n/2 - |W|| \le \alpha n$. We iterate through all
        $W$ in time proportional to area the of the colored region. The right figure
        illustrates the case of an $\alpha$-unbalanced solution.
		A division
        of the solution $(L,R)$ is witnessed by the roughly $2^{2\alpha n}$ sets $W$ in
    $\binom{\setrng{n}}{n/2}$ satisfying $S^L \subseteq W \subseteq \setrng{n} \setminus S^R$.}
	\label{fig:alg1}
\end{figure}

\subparagraph*{Case 1:} There exists an integer $b \in \{1,\ldots,m\}$ such that
$|n/2 - \sum_{j=1}^b |S_j|| \le \alpha n$.  In this case, observe
$\bigcup_{j=1}^bS_j$ is an element of 

\begin{displaymath} 
    \mathcal{W} \coloneqq \left\{Y \subseteq \{1,\ldots,n\} \text{ such that }
        w(Y) = \sum_{j=1}^b c_j \text{ and } \left| \frac{n}{2} - |Y|\right| \le \alpha n\right\}.
\end{displaymath}

Moreover, by~\re{1} we have $|\mathcal{W}| \le \beta(w) \le 2^{(1-\eps)n}$.
This means that we can enumerate $\cW$ in essentially $2^{(1-\eps)n}$ time,
because we will present an $\Oh((|\dc \cW| + |\uc\cW|)n)$ time algorithm that
for each $W  \in \cW$ computes whether $W$ can divided among bins $1,\ldots,b$
and $\{1,\ldots,n\} \setminus W$ among bins $b+1,\ldots,m$ (this algorithm is
based on techniques from~\cite{BjorklundHKK09}). This will detect a solution if
it exists. We bound the running time using the property $n/2 - \alpha n \le |W|
\le n/2 + \alpha n$. In this case we will show $|\dc \cW|+|\uc\cW| \leq
2^{(1-\eps')n}$ and hence the algorithm is fast enough (see left
Figure~\ref{fig:alg1} for an illustration).

\subparagraph*{Case 2:} For every $b \in \{1,\ldots,m\}$ we have $|n/2 - \sum_{j=1}^b |S_j|| > \alpha n]$.
Here we can use a method from~\cite{Nederlof16}: We let $\mathcal{W}$ consist of
$2^{(1-2\alpha)n}$ independently sampled subsets of $\{1,\ldots,n\}$ of cardinality $n/2$.
We answer \emph{yes} if there exist $W \in \cW$, disjoint sets
$S^L = S'_1,\ldots,S'_{b-1} \subseteq W$ and $S^R =S'_{b+1},\ldots,S'_{m}
\subseteq \{1,\ldots,n\} \setminus W$ such that $w(S'_j)=c_j$ for all $j \in
\{1,\ldots,m\}\setminus\{b\}$. This
condition can also be computed in $\Oh((|\dc \cW| +|\uc\cW|)n)$ time by the methods of~\cite{BjorklundHKK09}.
The crux is that both conditions together imply our instance is a
\emph{yes}-instance, since the remaining elements have total weight $c_b$ by
\re{2}. Moreover, by the balancedness assumption at least $2^{2\alpha
n}$ sets $W \subseteq \{1,\ldots,n\}$ with the above conditions exist.  Therefore the
random sampling will include such a $W$ with good probability (see right
Figure~\ref{fig:alg1} for an illustration).

\subsection{Related Work}
\label{subsec:relwork}

\paragraph{Littlewood--Offord, UDCP's, and Exponential Time Algorithms.}
Two sets $A,B \subseteq \{0,1\}^n$ form a Uniquely Decodable Code Pair (UDCP) if $|A+B|=|A|\cdot|B|$, where $A+B:= \{a+b: a\in A,b \in B\}$ (and addition is in $\mathbb{Z}^n$). The maximal sizes of UDCP's have been very well studied in information theory. See e.g.~\cite[Section 3.5.1]{schlegelgrant} for a (not so recent) overview. Two record upper bounds are $|A|\cdot |B| \leq 2^{1.5n}$ (from~\cite{Tilborg78}) and $|A| \leq 2^{(0.4228+\sqrt{\eps})n}$ whenever $|B| \leq 2^{(1-\eps)n}$ (from~\cite{AustrinKKN18}).
The study of UDCP's is relevant for this paper by the following connection shown
in~\cite{stacs2016}: For any vector $w \in \mathbb{Z}^n$, there is a UDCP $A,B
\subseteq \{0,1\}^n$ such that $|A|=|w(2^{\setrng{n}})|$ and $|B|=\beta(w)$.

A study of the trade-off between the parameters $|w(2^{\setrng{n}})|$ and $\beta(w)$
was already fruitful for obtaining improved exponential time algorithms in two
earlier papers in the context of the Subset Sum problem. In this problem one is
given $w \in \mathbb{Z}^n$ and a target integer $t$ and one needs to find a
subset $X \subseteq \setrng{n}$ such that $w(X)=t$.
First, the aforementioned paper~\cite{stacs2016} combined their connection to
UDCP's with the bound from~\cite{Tilborg78} to show that instances of Subset Sum
satisfying $|w(2^{\setrng{n}})| \geq 2^{0.997n}$ can be solved in $\Oh(2^{0.49991n})$
time, thereby improving the best $\Os(2^{n/2})$ worst case running time
from~\cite{HorowitzS74} for these instances.  Second, a slight variant of the
trade-off was used in~\cite{BansalGN018} to give a $\Oh(2^{0.86n})$ time
algorithm that uses polynomial space (assuming read-only access to 
the exponential number of random bits).

In a recent work by Jain et al.~\cite{jain2021anticoncentration}, the dependency of $\delta$ on $\eps$ in Theorem~\ref{lothm} was improved to $\delta(\eps) = \Oh(\sqrt{\eps})$. As a corollary, $\sigma_m$ in Theorem~\ref{mainthm} can be bounded with $\sigma_m = \Omega(m^{-12})$.

\paragraph{Exact Algorithms for Set Cover.}
Question~\ref{q1} was for the first time explicitly posed in~\cite{CyganDLMNOPSW16}, who showed that a negative answer to (a variant of) the question implies hardness in a fine-grained sense for the Subset Sum, Steiner Tree, and Connected Vertex Cover problems. A main motivation in~\cite{CyganDLMNOPSW16} for posing the question was a curious reduction showing that there is no improved algorithm for counting the number of Set Cover solutions modulo 2 unless improved algorithms for CNF-Sat exist (i.e. the Strong Exponential Time Hypothesis fails).
Later the assumption that no improved algorithm exists was dubbed as `Set Cover Conjecture' (see e.g.~\cite[Conjecture 14.36]{CyganFKLMPPS15}).
Since then, the conjecture has been used in several works, e.g. in~\cite{Abboud19,krauthgamer_et_al:LIPIcs:2019:10284}.

On the positive side, (especially for this work) important algorithmic tools were developed in~\cite{BjorklundHK09}: Fast zeta and M\"obius
transformations were introduced in the area of exponential time algorithms to
show that Set Cover can be solved in $2^{n}\cdot \poly(n)$ time even when the number of
sets in the input is exponential in $n$. One major consequence was a $2^n \cdot
\poly(n)$ time algorithm for computing whether an input graph on $n$ vertices
has a proper coloring with $k$ colors. While for $k \le 6$ faster algorithms exist~\cite{zamir:LIPIcs.ICALP.2021.113} this is still the fastest algorithm for $k > 6$.

Improved algorithms for solving Set Cover instances of sets with bounded cardinality were given in~\cite{Koivisto09}.
Later, this was generalized to improved algorithms for Set Cover instances where the optimum is linear in the universe size~\cite{Nederlof16}. Other instances that allow improved algorithms were also presented in e.g.~\cite{GolovnevKM14}.

\paragraph{Exact Algorithms for Bin Packing.}
In a textbook on exact exponential time algorithms, it was shown that Bin Packing
can be solved in time $\Oh(n2^{n} \cdot \max_{i}\{w(i)\})$ time~\cite[Section
4.2.3]{FominK10}.  A faster algorithm $\Os(2^n)$ time algorithm was given
in~\cite{BjorklundHK09}.  Even faster algorithms were given for $m=2,3$ in~\cite{LenteLST13}.

In~\cite{GoemansR14} it was shown that Bin Packing can be solved in polynomial
time if there are only a constant number of distinct weights of items.  Jansen et
al.~\cite{JansenKMS13} study Bin Packing with a constant number of bins and
bounded items weights was studied. They presented a dynamic programming algorithm
(similar to the one proposed by us in Lemma~\ref{lem:CaseA}) and show that it
runs in time $n^{\Oh(m)}$ if the item weights are polynomial in $n$.  This running time
cannot be improved to $n^{o(m / \log m)}$, unless the Exponential Time
Hypothesis fails~\cite{JansenKMS13}.

\paragraph{Heuristics for Bin Packing.}

The applications and combinatorial properties of Bin Packing have been studied since the
1930's~\cite{kantorovich}. To the best of our knowledge, the first attempt to
exactly solve Bin Packing with the assistance of the modern computer was developed
in the fifties by Eisemann~\cite{eisemann57}, with the motivation to trim losses in cutting
rolls of paper. Starting from the seventies, the research on exact algorithms
for Bin Packing focused on the branch-and-bound technique proposed by
Eilon and Christofides~\cite{eilon1971loading}. These heuristics work great in
practice. Nevertheless, there are no theoretical guarantees on their worst-case performance.

For a modern survey and experimental evaluations of the available software see
\cite{knapsack-book-martello,delorme2016bin}. 

\paragraph{Approximation Algorithms for Bin Packing.}

Bin Packing is one of the problems that initiated the study of approximation
algorithms. The earliest one is the \emph{First Fit} algorithm analysed by
Johnson~\cite{first-fit} that requires at most $1.7 \cdot \text{OPT} + 1$ bins. The
major breakthrough was done by
Karmarkar-Karp~\cite{karmarkar} who provided a polynomial time algorithm
that requires at most $\text{OPT} + \Oh(\log^2(\text{OPT}))$ bins. Recently, a
big leap forward was done by Rothvo{\ss}~\cite{rothvoss13} who
gave a polynomial time algorithm that requires only $\text{OPT} +
\Oh(\log(\text{OPT})\log\log(\text{OPT}))$ bins and Hoberg and
Rothvo{\ss}~\cite{hoberg17} who improved this even further to
$\text{OPT} + \Oh(\log(\text{OPT}))$ bins.

\subsection{Organization}
This paper is organized as follows:
In Section~\ref{sec:prel} we present some preliminaries and introduce some notations.
In Section~\ref{sec:main} we present the algorithm and proof of our main theorem, assuming Theorem~\ref{lothm}.
The latter theorem is proved in the next two Sections~\ref{sec:LO} and~\ref{sec:technical-section}. 
In Appendix~\ref{sec:pruner} we include the proofs of technical Lemmas
from Section~\ref{sec:main}. In Appendix~\ref{sec:inequalities} we include the proofs of
useful inequalities regarding binary entropy.

\section{Preliminaries}\label{sec:prel}

Throughout the paper, we use the $\Os$ notation to hide polynomial factors in the
input size. The notation $\Ot(T)$ means $\Oh(T \cdot \polylog(T))$. 
The number of bins is assumed to be constant, i.e. $m =\Oh(1)$. 
We say a function $f(\eps) = \Oh_{\eps
\rightarrow 0}(g(\eps))$ if there exists a positive number $C$ and sufficiently
small $\eps_0 > 0$, such that $|f(\eps)| \le C \cdot g(\eps)$ for all $\eps <
\eps_0$. We use $\Omega_{\eps \rightarrow 0}$ similarly to express lower
bounds. Finally $\poly(n)$ is a shorthand notation for $n^{\Oh(1)}$.
All the logarithms are base $2$ unless stated otherwise.

In this paper, we assume that basic
arithmetic operations take constant time. 
We use a result of Frank and Tardos~\cite{DBLP:journals/combinatorica/FrankT87},
in a similar way to \cite{etscheid2017polynomial}, to assume that $\max_i
\{\log{w(i)}\} \le \poly(n)$.
%

If $a,b \in \mathbb{R}$ and $b \geq 0$ we let $[a \pm b]$ denote the interval $[a-b,a+b]$.
If $A$ and $B$ are sets, we denote by $B^A$ the set of vectors indexed by $A$ with values from $B$,
and we will interchangeably address these vectors as functions from $A$ to $B$.
If $f \in B^A$ and $b \in B$ we denote $f^{-1}(b)  \coloneqq \{ a \in A: f(a) = b\}$ for its inverse evaluated at $b$.
For example, when $x \in \{0,1\}^{\{1,\ldots,n\}}$ then $x^{-1}(0) = \{ i \in
\{1,\ldots,n\} \; : \; x(i) = 0\}$,
If $x,y \in \mathbb{R}^A$ we denote $\langle x,y \rangle  \coloneqq \sum_{a \in A} x_a\cdot y_a$ for their inner product.

To quickly refer to the properties of a solution of a Bin Packing instance we use the following notations:
The function $w$ indicates the weights of the input. It is extended to sets $X
\subseteq \{1,\ldots,n\}$ by defining $w(X) \coloneqq \sum_{ i \in X}w(i)$ and to set families
$\mathcal{F} \subseteq 2^{\setrng{n}}$ by defining $w(\mathcal{F}) \coloneqq \{w(X) : X \in \mathcal{F}\}$.
We say a set $X\subseteq \setrng{n}$ of items can be \emph{divided} over bins $L \subseteq
\setrng{m}$ if there is a partition $X_1,\dots,X_{|L|}$ of $X$, such that for all $j \in
\{1,\dots,|L|\}$, the set $X_j$ can be placed in bin $j$, i.e., $w(X_j)\le c_j$. 

We abstract a simple probabilistic argument that we use several times in our algorithms. It follows easily by observing that a random element from $X$ is in $Y$ with probability $|X| / |Y|$.
\begin{observation}\label{lem:probabilities}
    Let $U$ be any universe set and let $Y \subseteq U$ be
    an arbitrary nonempty subset of $U$. Let $Z$ be a set obtained by sampling (with
    replacement) $\lceil |U| / |Y| \rceil$ times uniformly at random
    from $U$ (if at the end of this process an element repeats, we take a single
    occurrence of this element). Then $\Pr[Z\cap Y \not= \emptyset] \geq  1-\frac{1}{e}$. 
\end{observation}

\subsection{Preliminary Tools: Fast Transformations}

Our algorithm will crucially rely on the following algorithmic tools and definitions from~\cite{BjorklundHKK09}.
\begin{definition}[Zeta and M\"obius Transform]
	Let $f: 2^U \rightarrow \bbN$. Then the zeta transform $\zeta f$ and M\"obius transform $\mu f$ are functions from $2^U$ to $\bbN$ such that for every $X \subseteq U$:

    \begin{align*}
        (\zeta f)(X) \coloneqq \sum_{Y \subseteq X} f(Y)  &&\text{ and}&& (\mu f)(X) \coloneqq \sum_{Y \subseteq X} (-1)^{|U \setminus Y|} f(Y)
        .
    \end{align*}
    
\end{definition}

\begin{definition}
	Given $\cS \subseteq 2^U$, the \emph{down-closure} $\dc\cS$ and \emph{up-closure} $\uc\cS$ are defined as follows:
    \begin{align*}
        \dc\cS  \coloneqq \{ X : \exists S \in \cS \text{ such that } X \subseteq S   \} &&\text{and}&&
        \uc\cS  \coloneqq \{ X : \exists S \in \cS \text{ such that } X \supseteq S   \}.
    \end{align*}
\end{definition}

\begin{theorem}[Fast zeta and M\"obius transform~\cite{BjorklundHKK09}]\label{Thm:zetamob}
	Suppose that $f:2^U \rightarrow \bbN$ is such that $f(X)$ can be evaluated in $T$ time for any given $X \subseteq U$, and let $\cS \subseteq 2^U$ be a set family.
	There is an algorithm that can compute for every $X \in \dc\cS$ the values
    $(\zeta f)(X)$ and $(\mu f)(X)$. The algorithm runs in
    $\Oh(|\dc\cS|\cdot|U|\cdot T)$ time.
\end{theorem}

\begin{definition}[Cover and Entry-Wise Product]
	Given $f,g : 2^{U} \rightarrow \bbN$, the cover product $f \coverprod g =h$ and the
    entry-wise product $f \cdot g = h'$ are the functions $h,h': 2^{U} \rightarrow \bbN$ such that
	\begin{align*}
        h(Z)  \coloneqq \sum_{X \cup Y = Z} f(X) g(Y) &&\text{and}&& h'(Z)  \coloneqq f(Z) \cdot g(Z).
	\end{align*}
\end{definition}

\begin{theorem}[\cite{BjorklundHKK07}]\label{thm:coverproductzeta}
	$\mu((\zeta f)\cdot (\zeta g)) = f \coverprod g$.
\end{theorem}

\begin{theorem} \label{thm:ZetaMobBins} Suppose that we have a Bin Packing
    instance with bin capacities $c_1,\dots,c_m$ and item weight function $w$.
    Then for any $B \subseteq \setrng{m}$ and set $\cW \subseteq 2^{\setrng{n}}$, computing
    for all $X \in \dc \cW$ whether $X$ can be divided over the bins in $B$ can
    be done in time $\Oh(|\dc \cW|n)$. Similarly, for any $B\subseteq \setrng{m}$ and set
    $\cW\subseteq 2^{\setrng{n}}$, computing for all $X \in \uc \cW$ whether
    $\setrng{n}\setminus X$ can be divided over the bins in $B$ can be done in
    time $\Oh(|\uc \cW|n)$.  
\end{theorem}
\begin{proof}
    For all $j = 1,\dots,m$ define a function $f_j: 2^{\setrng{n}} \to \{0,1\}$ as  
	\[
	f_j(X) = \begin{cases} 1, &\text{ if } w(X) \le c_j \\ 0, &\text{ otherwise.}\end{cases}
	\] 
	Assume without loss of generality that $B = \{1,\dots,d\}$.	Notice that $X$
    can be divided over the bins in $B$ if and only if $(f_1\coverprod
    f_2\coverprod \cdots \coverprod f_d)(X) >0 $. By Theorem~\ref{thm:coverproductzeta} we have that 
	\[f_1\coverprod f_2\coverprod \cdots \coverprod f_d = \mu((\zeta f_1) \cdot (\zeta f_2 )\cdots (\zeta f_d)).\]
	Then, the right hand side can be computed in $\Oh(|\dc\cW|)$ time using
    subsequently fast $d$ zeta transformation (Theorem~\ref{Thm:zetamob}),
    na\"ive entry-wise product computation, and one fast M\"obius transformation
    (Theorem~\ref{Thm:zetamob}). The proof for the second part of the theorem
    one takes $\cW'  \coloneqq \{ \setrng{n}\setminus  W: W \in \cW\}$ and applies the technique above to $\cW'$. Notice that indeed $\dc\cW' = \uc \cW$.
\end{proof}

Note this can be used to obtain the algorithm already mentioned in
Section~\ref{sec:intro}. We include the proof to introduce the reader
to the state-of-the-art algorithm which will be expanded in the later
sections.
\begin{theorem}[\cite{BjorklundHKK09}]\label{thm:bpbl}
	Bin Packing with capacities $c_1,\ldots,c_m$ can be solved in $\Os(2^n)$ time.
\end{theorem}
\begin{proof}
    For $i=1,\ldots,m$ define the function $f_i: 2^{\setrng{n}} \rightarrow \{0,1\}$ as
	\[
	f_i (X) = \begin{cases}
	1, & \text{if } w(X) \leq c_i\\
	0, & \text{otherwise}.
	\end{cases}
	\]
    Note that $(f_1 \coverprod f_2 \coverprod \ldots \coverprod f_m)(\setrng{n})
    > 0$ if and only if the answer to Bin Packing is positive. By Theorem~\ref{thm:coverproductzeta} we have that
	\[
	f_1 \coverprod f_2 \coverprod \ldots \coverprod f_m = \mu ((\zeta f_1) \cdot (\zeta  f_2) \ldots (\zeta \cdot f_m)),
	\]
	and the right hand side can be computed in $\Os(2^n)$ time using
    subsequently fast $m$ zeta transformations (Theorem~\ref{Thm:zetamob}),
    na\"ive entry-wise product computation, and one fast M\"obius transformation (Theorem~\ref{Thm:zetamob}).
\end{proof}

\subsection{The Entropy Function and Binomial Coefficients}
We heavily use properties of the entropy function, which we will now define.
For a discrete probability space $\mathcal{D}=(\Omega,p)$ where $p : \Omega
\rightarrow (0,1)$, the \emph{entropy} of $\mathcal{D}$ is defined as follows:
\begin{equation}\label{eq:entro}
    h(\mathcal{D}) \coloneqq -\sum_{x \in \Omega}p(x)\log p(x).
\end{equation}

We say $p=(p_1,\ldots,p_k)$ is a probability vector if the $p_i$'s are
non-negative and satisfy $\sum_{i=1}^k p_i=1$. If no underlying probability
space is given, we may interpret $p$ as a probability measure over
$\{1,\ldots,k\}$ and thus~\eqref{eq:entro} gives $h(p)=-\sum_{i=1}^k p_i\log
p_i$. The \emph{support} of the vector $p = (p_1,\ldots,p_k)$ the set of its
non-zero coordinates and the \emph{size} of the support is the number of
non-zero coordinates of $p$. If $p \in (0,1)$, we use the shorthand notation
$h(p) \coloneqq h(p,1-p)$. The multinomial coefficient $\binom{n}{p_1n,p_2n,\ldots,p_kn}$can be
approximated with $h(p)$ as follows:
	
\begin{lemma}[\cite{csiszar2004information}, Lemma 2.2]\label{lem:multvsentropy}
	If $p$ is a probability vector with support of size at most $s$, then
	\[
        \binom{n+s-1}{s-1}^{-1} 2^{h(p)n} \leq \binom{n}{p_1 n,\ldots,p_k n} \leq 2^{h(p)n}.
	\]
\end{lemma}

We will frequently use the special case $\binom{n}{pn} \leq 2^{h(p)n}$ when $p \in (0,1)$.

The following lemma states the intuitive fact that close probability vectors have close entropy.
\begin{lemma}\label{lem:entclose}
	Let $p,q \in \mathbb{R}^k$ be probability vectors such that $|p_i-q_i| \leq
    \eps$ for each $i=1,\ldots,k$. Then $|h(p)-h(q)| \leq \frac{1}{\ln(2)} \cdot
    \eps k \log \tfrac{1}{\eps}$.
\end{lemma}
\begin{proof}
	Recall that $h(p)=\sum_{i=1}^k p_i \log \tfrac{1}{p_i}$.
	Thus the lemma follows by applying the following inequalities to all
    summands of the entropy of $p$ and $q$: When $x,\eps,x+\eps \in [0,1]$, then
    \begin{displaymath}
        x\log\tfrac{1}{x} -\frac{\eps}{\ln(2)} \leq  (x+\eps) \log \tfrac{1}{x+\eps} \leq x \log \tfrac{1}{x} + \eps \log \tfrac{1}{\eps}.
    \end{displaymath}
	The second inequality is direct, and the first inequality can be derived as
    \begin{displaymath}
        (x+\eps) \log \tfrac{1}{x+\eps} = x \log\tfrac{1}{x} + x
        \log\tfrac{x}{x+\eps}+ \eps \log\tfrac{1}{x+\eps} \geq x\log\tfrac{1}{x} - x
        \log(1+\tfrac{\eps}{x}) \geq x \log \tfrac{1}{x} - \frac{1}{\ln(2)}
        \cdot \eps,
    \end{displaymath}
	where in the last inequality we use the standard fact that for every $z \in \real$ it holds that $1+z \leq \exp(z)$.
\end{proof}

\section{Proof of Theorem~\ref{mainthm}}\label{sec:main}

In this section we prove our main theorem which we first restate for convenience:

\begin{theorem}    
 For every $m \in \nat$ there is a constant $\sigma_m >0$ such that every Bin Packing instance with $m$ bins can be solved in
    $\Oh(2^{(1-\sigma_m)n})$ time with high probability.
\end{theorem}

For the proof of Theorem~\ref{mainthm}, we combine four lemmas, each solving
particular types of instances (see Figure~\ref{tikz:MainThmProof} for an overview
of the algorithm). We will refer to these different types of instances by
Case~A, Case~B, Case~C and Case~D.   

This section is organized as follows:
In Subsection~\ref{sec:Definitions} we introduce definitions that will be used throughout this section, such as the key definition of \emph{$\alpha$-balanced solutions}. 
We then prove in Subsections~\ref{sec:CaseA} and \ref{sec:CaseB} that `easy' instances of
Bin Packing, namely those where $w$ generates relatively few distinct sums (Case~A) and those with $\alpha$-unbalanced solutions for some $\alpha>0$ (Case~B), can be solved
fast. We can therefore assume that there are only $\alpha$-balanced solutions
and that $|w(2^{\setrng{n}})| > 2^{\delta n}$ (for some $\delta < \frac{1}{m}$) in the rest of the section.

Subsection~\ref{sec:PrunedSlack} introduces a few more definitions, such as the
``slack of a bin'', which is the vacant capacity of a bin in a solution. This is
also where we define the `$\prun$-pruned item weights' as the bit representation
of the weights, pruned to the $\prun \in \nat$ most significant bits. The parameter
$\prun$ is then chosen such that $|w_\prun(2^{\setrng{n}})| = \Theta(n 2^{\delta n})$, as discussed in \S~\ref{subsubsec:liftr2}.
These definitions will be central in solving the remaining two types of instances. 

In Subsection~\ref{sec:manysmallslack}, we consider instances where at least
roughly half of the items are in a bin with small slack (Case C). This is where we use
the approach discussed in \S\ref{subsubsec:liftr1} and apply Theorem~\ref{lothm}
on the $\prun$-pruned item weights, to conclude that $\beta(w_\prun) \le
2^{(1-\epsilon)n}$ for some $\epsilon >0$. 

Subsection~\ref{sec:fewsmallslack} then solves instances where at least roughly
half of the items are in a bin with large slack (Case D). In the proof, we can split the
large slack bins into two parts, where we use the $\prun$-pruned item weights in
each of these parts to determine whether they fit. Because
$|w_\prun(2^{\setrng{n}})|$ is of order $2^{\delta n} $, there is of order $2^{\delta
m n}$ different tuples of weights. This we can keep track of since we assumed
$\delta < \frac{1}{m}$. Furthermore, we split the large slack bins into
two parts. This operation guarantees (that with a constant probability) we can
correctly guess the partition of the items in small slack bins into two parts. 

Finally, the proof of Theorem~\ref{mainthm} can be found in
Subsection~\ref{subsec:mainproof}, where we combine all these results by
selecting the appropriate values for $\delta$ and $\alpha$ based on the number of
bins (see Figure~\ref{tikz:MainThmProof} for overview of the algorithm).
\begin{figure}[t]
    \centering
     \begin{tikzpicture} [scale=0.6, every node/.style={scale=0.8}] 
		\draw [rounded corners, fill = white!100!black] (-.5,0) rectangle (4.5,1.5)
        node[pos=.5] {$|w(2^{\setrng{n}})| \le 2^{\delta n}$? };
		
		\draw [rounded corners, fill = white!80!black] (-4.5,-3) rectangle (0.5,-1.5) node[pos=.5] {Case A: Lemma~\ref{lem:CaseA}};
		\draw [rounded corners, fill = white!100!black] (3,-3) rectangle (9,-1.5) node[pos=.5] {$\alpha$-unbalanced solution?};
		
		\draw[->,thick]  (1,0) to (-1,-1.3);
        \node at (-1,-0.7) {yes};
		\draw[->,thick]  (3,0) to (5,-1.3);
        \node at (5,-0.7) {no};
		
		\draw [rounded corners, fill = white!80!black] (-.5,-6) rectangle (4.5,-4.5) node[pos=.5] {Case B: Lemma~\ref{lem:CaseB}};
		\draw [rounded corners, fill = white!100!black] (5.5,-6) rectangle (14.5,-4.5) node[pos=.5] {At least $(1/2-\alpha)n$ small slack items?};
		
		\draw[->,thick]  (5,-3) to (3,-4.3);
		\node at (3,-3.7) {yes};
		\draw[->,thick]  (7,-3) to (9,-4.3);
		\node at (9,-3.7) {no};
		
		\draw [rounded corners, fill = white!80!black] (3.5,-9) rectangle (8.5,-7.5) node[pos=.5] {Case C: Lemma~\ref{lem:CaseC}};
		\draw [rounded corners, fill = white!80!black] (11.5,-9) rectangle (16.5,-7.5) node[pos=.5] {Case D: Lemma~\ref{lem:CaseD}};
		
		\draw[->,thick]  (9,-6) to (7,-7.3);
		\node at (7,-6.7) {yes};
		\draw[->,thick]  (11,-6) to (13,-7.3);
		\node at (13,-6.7) {no};

    \end{tikzpicture}
    
    \caption{Overview of use of Lemmas proving Theorem~\ref{mainthm}. }
    \label{tikz:MainThmProof}
\end{figure}

\subsection{Balanced Solutions and Witnesses}\label{sec:Definitions}

Fix an instance of Bin Packing. We begin by formally defining solutions.

\begin{definition}[Solution]
    A partition $S_1,\dots,S_m$ of $\setrng{n}$ is a \emph{solution} of an instance of
    Bin Packing with $n$ items and $m$ bins, if for all $j \in \{1,\dots,m\}$
    the set $S_j$ fits into the bin $j$ (i.e., $w(S_j) \le c_j$).
\end{definition}

The following notion of a \emph{witness} will be crucial in our approach.

\begin{definition} [$(L,R)$-witnesses] 
    Let $L,R \subseteq \setrng{m}$, $L\cap R = \emptyset$. A set $W \subseteq \setrng{n}$ is an
    $(L,R)$-\emph{witness} if there is a solution $S_1,\dots,S_m$ such that
    \begin{align*}
        \bigcup_{j \in L} S_j \subseteq W &&\text{ and } && \bigcup_{j \in R} S_j \subseteq \setrng{n}\setminus W  
      .
    \end{align*}
\end{definition}
We commonly denote $S^L \coloneqq \bigcup_{j \in L} S_j$ and $S^R \coloneqq \bigcup_{j \in R} S_j$.
Observe that 
to verify that a set $W\subseteq \setrng{n}$ is an $(L,R)$-witness, it is
sufficient to find $S^L \subseteq W$ and $S^R \subseteq \setrng{n}\setminus W$
with the following properties:
(i) the items in $S^L$ can be distributed to the bins in $L$, (ii)
items $S^R$ can be distributed to the bins in $R$, and (iii)
items $\setrng{n}\setminus (S^L\cup S^R)$ can be distributed to the bins in $\setrng{m}\setminus (L\cup R)$.
Hence, finding a witness gives us a `certificate' for the existence of a solution. This will be used several times throughout this section. 
 
Our algorithmic approach will heavily depend on whether or not the set of items can be evenly divided, which we formalize as follows:
\begin{definition} [$\alpha$-balanced solution]\label{def:abalsol}
Let $S_1,\dots,S_m$ be a solution of Bin Packing. The solution is
$\alpha$\emph{-balanced} if for all permutations $\pi : \setrng{m} \to
\setrng{m}$ there
exists an index $b \in \setrng{m}$ such that $\sum_{j=1}^b |S_{\pi(j)}| \in [n/2 \pm \alpha n]$. If a solution is not $\alpha$-balanced, it is called $\alpha$\emph{-unbalanced}.
\end{definition}

Hence, a solution is $\alpha$-unbalanced if and only if there exists a
permutation $\pi : \setrng{m} \to \setrng{m}$ and a $b\in \setrng{m}$ such that $\sum_{j = 1}^{b-1} |S_{\pi(j)}| < (1/2 - \alpha)n$ and $ \sum_{j = 1}^{b} |S_{\pi(j)}|> (1/2 + \alpha)n$.

\subsection{Solving Case A: Few Distinct Sums}

 If the instance generates relatively few distinct sums in the sense that
 $|w(2^{\setrng{n}})| \le 2^{\delta n}$ for some small $\delta < 1/m$, we can solve Bin
 Packing sufficiently fast. This is the algorithm we use in Case A.

 \label{sec:CaseA}
\begin{case}
	\textbf{In Case A:} \; (\ding{51})  $|w(2^{\setrng{n}})| \le 2^{\delta n}$
\end{case}

\begin{lemma} \label{lem:CaseA}
	A solution of Bin Packing can be found in $\Oh(n \cdot m \cdot |w(2^{\setrng{n}})|^{m})$ time.
\end{lemma}
\begin{proof}
	As a first step, we compute the set $w(2^{\setrng{n}})$ in $\Oh(n \cdot |w(2^{\setrng{n}})|)$ time with
	Lemma~\ref{lem:computesums}.  Next, we use the following dynamic
	programming algorithm:
	
	For every $i \in \{0,\ldots,n\}$ and $a_1,\ldots,a_m \in w(2^{\setrng{n}})$, define 
	\begin{displaymath}
		\DP_i[a_1,\dots,a_m] \coloneqq \begin{cases}
			\mathtt{true} &\text{if items } \{1,\dots,i\} \text{ can be
				distributed to } \\& \text{bins with capacities } a_1,\ldots,a_m \text{ exactly},\\
			\mathtt{false} &\text{otherwise.}
		\end{cases}
	\end{displaymath}
	
	We initiate a dynamic programming table with $\DP_0[0,\ldots,0] =
	\mathtt{true}$ and the remaining entries of $\DP_0$ are set to
	$\mathtt{false}$, i.e., $\DP_0[a_1,\ldots,a_m] = \mathtt{false}$ if there
	exists $j \in \{1,\ldots,m\}$ with $a_j \neq 0$.  Then the following
	recurrence relation holds for every $i \in \{1,\ldots,n\}$ and
	$a_1,\ldots,a_m \in w(2^{\setrng{n}})$:
	
	\begin{displaymath}
		\DP_i[a_1,\ldots,a_m] = \bigvee_{j \in \{1,\ldots,m\}} \DP_{i-1}[a_1,\ldots, a_j - w(i),\ldots, a_m]
		.
	\end{displaymath}
	
	This concludes the description of the dynamic programming procedure. Observe
	that in the above recursion we only need to keep track of $\mathtt{true}$
	entries of table $\DP_i$ for every $i \in \{1,\ldots,n\}$. Hence, we only
	need to check entries with $a_1,\ldots, a_m \in w(2^{\setrng{n}})$, since
	those are all the possible sums that $w$ generates. Therefore, we have 
	$\DP_i(a_1,\ldots,a_m)=\mathtt{false}$ whenever $a_j \notin
	w(2^{\setrng{n}})$ for some $j \in \{1,\ldots,n\}$.  Since each entry of $\DP_i$ table can be
	computed in $\Oh(m)$ time, the running time follows.

	Finally, observe that when $c_1,\dots,c_m$ are the capacities of the bins of
	the bin packing instance, then there exists $a_1,\ldots,a_m$ such that
	$\DP_n(a_1,...,a_m) = \mathtt{true}$ and $a_j \le c_j$ for all $j \in
	\{1,\ldots,m\}$ if and only if there is a solution to bin packing instance.
	This condition can be verified by scanning $\Oh(w(2^{\setrng{n}}))$ many $\mathtt{true}$-entries of table $\DP_n$ .
\end{proof}

\subsection{Solving Case B: Unbalanced Solutions}
\label{sec:CaseB}

Next, we show that $\alpha$-unbalanced solutions (for some $\alpha>0$) can be detected quickly, i.e. we solve the instances in Case~B.

\begin{case}
	\textbf{In Case B:} \; (\ding{51}) The instance has an $\alpha$-unbalanced solution (for some $\alpha > 0$).
\end{case}
Note that we might also add the assumption that $|w(2^{\setrng{n}})| > 2^{\delta n}$ as Case A solves all other instances. However, we do not need this assumption for Case B, as we have the following result. 
\begin{lemma} \label{lem:CaseB}
	If a Bin Packing instance has an $\alpha$-unbalanced solution, then such a
	solution can be found in  $\Os(2^{(1-f_B(\alpha))n})$ time with probability $\ge
	1/2$ where $f_B(\alpha) = \Omega_{\alpha \to 0} \left (\frac{\alpha^2}{\log^2 (\alpha)} \right )$. 
\end{lemma}

\begin{proof}
	
	The algorithm iterates over all subsets $L, R \subseteq \setrng{m}$ such
	that $L\cap R = \emptyset$, $|L\cup R| = m-1$. Let $b\in \setrng{m}$ be the
	only element not in $L\cup R$. For each such $L$ and $R$, the algorithm
	will search for an $(L,R)$-witness of size $\frac{n}{2}$.
	Concretely, it samples a set $\cW$ of $2^{(1-2\alpha)n}$ random subsets (with replacement, and removing copies afterward) of
	$\setrng{n}$ of size $\frac{n}{2}$, and it computes for every $W\in \cW$ whether it is an $(L,R)$-witness as follows: First, it computes which sets from $\dc \cW \cup \uc\cW$ are potential candidates for $S^L$ and $S^R$. This is done by computing the booleans $l_X$ for every $X \in \dc\cW$ and $r_X$ for every $X \in \uc\cW$, where 
	\begin{align*}
		l_X &:= \begin{cases} \mathtt{true} & \text{ if }X \text{ can be
				distributed to the bins in } L, \\ \mathtt{false} & \text{ otherwise,}\end{cases}\\
		r_X &:= \begin{cases} \mathtt{true} & \text{ if } \setrng{n}\setminus X
			\text{ can be distributed to the bins in } R, \\ \mathtt{false} & \text{ otherwise.} \end{cases} 
	\end{align*}
	This can be done in time $\Oh((|\dc\cW|+|\uc\cW|)n)$ using Theorem~\ref{thm:ZetaMobBins}.
	
	Second, for each $W \in \cW$, we search for sets $X^L \subseteq W$ and $X^R
	\subseteq \setrng{n}\setminus W$ of maximum weight such that they can be distributed to the bins in $L$ and $R$ respectively. To do this, we compute $l^*_X$ for every $X \in \dc\cW$ and $r^*_X$ for every $X \in \uc\cW$, where 
	\[ l^*_X := \max_{Y \subseteq X: l_Y=\mathtt{true}} w(Y), \qquad r^*_X := \max_{Y \supseteq
		X: r_Y=\mathtt{true}} w(\setrng{n}\setminus Y).   \]
	This can be done using dynamic programming with the recurrence relations 
	\[l^*_X = \begin{cases} w(X) &\text{if }l_X = \mathtt{true}, \\ \max_{i \in X}
		l^*_{X\setminus \{i\}} &\text{if }l_X=\mathtt{false}, 
	\end{cases}  \qquad \text{ and } \qquad r^*_X = \begin{cases}
		w(\setrng{n}\setminus X) &\text{if }r_X = \mathtt{true}, \\ \max_{i \not
			\in X} r^*_{X\cup \{i\}} &\text{if }r_X=\mathtt{false}.
	\end{cases} \]
	The running time is only $\Oh((|\dc\cW|+|\uc\cW|)n)$ since the values $l^*_X$ for $X
	\in \dc \cW$ do not depend on entries $l^*_{Y}$ for $Y \notin \dc\cW$, and the values $r^*_X$ for $X \in \uc \cW$ do not depend on entries $r^*_{Y}$ for $Y \notin \uc\cW$. Thus the algorithm only needs to evaluate $|\dc\cW|+|\uc\cW|$ table entries which can be done in time $\Oh(n)$ per entry.
	
	Third, the algorithm checks if there exists a $W \in \cW$ such that
	$w(\{1,\ldots,n\}) - l^*_W -r^*_W \le c_b$ and returns \emph{yes} if this is the
	case. If for all different choices of $L$ and $R$, no $(L,R)$-witness has been
	found, the algorithm returns \emph{no}.
	
	\subsubsection*{Correctness of Algorithm}
	Assume that there is an $\alpha$-unbalanced solution $S_1,\ldots,S_m$. Let $\pi
	: \setrng{m} \to \setrng{m}$ be a permutation of the bins such that $\sum_{j = 1}^{b-1}
	|S_{\pi(j)}| < (1/2 - \alpha)n$ and $ \sum_{j = 1}^{b} |S_{\pi(j)}|> (1/2 +
	\alpha)n$ for some $b\in\setrng{m}$. Thus $|S_{\pi(b)}| \geq 2\alpha n$.
	Take $L = \{\pi(1),\dots,\pi(b-1)\}$ and $R = \{\pi(b+1),\dots,\pi(m)\}$.
	Recall the notation $S^L=\bigcup_{j = 1}^{b-1} S_{\pi(j)}$.
	Since for every $Y \in \binom{S_{\pi(b)}}{n/2-|S^L|}$ the set $Y \cup S^L$ is an
	$(L,R)$-witness of size $\frac{n}{2}$, there are at least
	$\binom{|S_{\pi(b)}|}{n/2 - |S^L|}$ $(L,R)$-witnesses of cardinality
	$\frac{n}{2}$, which is at least $\binom{2\alpha n}{\alpha n}$ since
	$|S_{\pi(b)}| \ge 2 \alpha n$ and $n/2 - |S^L| \geq \alpha n$. 
	Hence,  there are $\binom{n}{n/2}$ subsets of $\setrng{n}$ of cardinality $n/2$ and at least $\binom{2\alpha n}{\alpha n}$ of those are $(L,R)$-witnesses. 
	Observation~\ref{lem:probabilities} then gives that $\cW$ contains an $(L,R)$-witness with probability $\ge 1/2$ as $\cW$ is a family of $2^{(1-2\alpha)n} = 2^n / 2^{2\alpha n}$ random subsets of
	$\setrng{n}$ of size $\frac{n}{2}$.
Notice that for any witness $W$ it will hold that $\sum_{i=1}^n w(i) - l^*_W - r^*_W \le c_{\pi(b)}$ and so the algorithm will return yes if $W \in \cW$.
	
	Moreover, when the algorithm finds a $W \in \cW$ such that $\sum_{i=1}^n w(i) -
	l^*_W -r^*_W \le c_{\pi(b)}$, it means there exist sets $X^L \subseteq W$ and
	$X^R \subseteq \setrng{n}\setminus W$ that can be distributed to the bins of $L$ and $R$ respectively, such that $
	\setrng{n} \setminus (X^L \cup X^R)$ fits into bin $\pi(b)$. Therefore, $W$ is an $(L,R)$-witness and we proved the existence of a solution to the Bin Packing instance.  
	
	\subsubsection*{Running time Analysis}
	
	We are left to prove the running time of the algorithm. Recall that the algorithm
	will repeat the procedure above for all $m \cdot 2^{m-1}$ combinations of $L$ and $R$.
	The running time per one guess of $L$ and $R$ is dominated by
	$\Oh((|\dc\cW|+|\uc\cW|)n)$, hence we are left to prove that $|\dc\cW|+|\uc\cW| = \Oh(2^{(1-f_B(\alpha))n})$. For this we use Lemma~\ref{lem:boundeddownset} with $z = 2\alpha$ and $c=0$. This implies that 
	\[|\dc\cW|+|\uc\cW| \le \Oh\left(2^{(1-\rho(2\alpha,0))n}\right),\] where 
	\[\rho(2\alpha,0) = \frac{2}{\ln(2)} \cdot \left(\frac{2\alpha}{4\log(12/2\alpha)}\right)^2  = \Omega_{\alpha \to 0} \left (\frac{\alpha^2}{\log^2
		(1/\alpha)} \right) = \Omega_{\alpha \to 0} \left (\frac{\alpha^2}{\log^2
		(\alpha)} \right)\]
\end{proof}

\subsection{Pruned Item Weights and Slack} \label{sec:PrunedSlack}
The results from the previous subsection enable us to assume that both
$|w(2^{\setrng{n}})| \ge 2^{\delta n}$ for some small constant $\delta > 0$ (that we will
fix later) and that there is an $\alpha$-balanced solution for some $\alpha >0$.
To solve these instances of Bin Packing, we first need to define different
parameters of an instance that determine our proof strategy.

\begin{definition} [$s$-pruned item weights]
Let $l= 1+ \lceil \log(\max_i \{w(i)\}) \rceil $. For $s \in \{0,\dots,l\}$, define the $s$\emph{-pruned weight} of an item $i$ as
\[ w_s(i) := \lfloor w(i) / 2^{l-s} \rfloor.\]
\end{definition}
The $s$-pruned weight of item $i$ comes down to pruning the $l$-bit representation of $w(i)$ to the $s$ most significant bits. Indeed, $w_s(i) \le 2^s$, and $w_0(i) = 0$ for all items $i$ and $w_l =w$.
We will need the fact that the sequence
\[
    1= |w_0(2^{\setrng{n}})|, \;\; |w_1(2^{\setrng{n}})|, \;\; \ldots, \;\;
    |w_l(2^{\setrng{n}})| = |w(2^{\setrng{n}})|,
\]
is almost non-decreasing and relatively smooth. Observe that the sequence may
not be non-decreasing. For example when $w=(3,7,10)$ the number of bits is $l=5$, and
\begin{align*}
    &w_0=(0,0,0), &|w_0(2^{\setrng{n}})|&=1 \qquad\quad &&w_1=(0,0,0),
    &&|w_1(2^{\setrng{n}})|=1\\
    &w_2=(0,0,1), &|w_2(2^{\setrng{n}})|&=2 \qquad\quad &&w_3=(0,1,2),
    &&|w_3(2^{\setrng{n}})|=1\\
    &w_4=(1,3,5), &|w_4(2^{\setrng{n}})|&=8 \qquad\quad &&w_5=(3,7,10),
    &&|w_5(2^{\setrng{n}})|=7.
\end{align*}
Nevertheless, this is only an artifact of smaller-order rounding errors
and that the sequence in fact is \emph{smooth} in the following precise sense (see Lemma~\ref{lem:relationps}): For all $s \in \{1,\dots,l\}$:
\[\left(\frac{1}{3n}\right)|w_s(2^{\setrng{n}})| \le
|w_{s-1}(2^{\setrng{n}})| \le \left (\frac{3n}{2} \right
)|w_s(2^{\setrng{n}})|.
\]

A part of our strategy is to use techniques from Lemma~\ref{lem:CaseA} to deal
with mostly empty bins.  The analogous dynamic programming table needs
to be indexed by $w_s$ for some $s \in \{1,\dots,l\}$. To achieve this, we will need a
notion of \emph{precision}. The precision parameter we will use is the
following:

\begin{definition}[Critical pruner]
    Let $\delta \in (0,1)$ be a fixed parameter\footnote{Which we will be set later in Subsection~\ref{subsec:mainproof}}
    such that $|w(2^{\setrng{n}})|\ge 2^{\delta n}$. We define the \emph{critical
    pruner} $\prun$ as
    \[  \prun:= \prun(\delta) \coloneqq \min \left\{s \in \N :
            |w_s(2^{\setrng{n}}) |\ge
2^{\delta n}  \right\}. \]
\end{definition}

Observe, that $|w_\prun(2^{\setrng{n}})|=\Theta(n2^{\delta n})$ by
Lemma~\ref{lem:relationps} and the fact that $w_0(2^{\setrng{n}}) = \{0\}$.
Furthermore, by Corollary~\ref{cor:critPruner} the critical pruner $\prun$ can be computed in $\Os(2^{\delta n})$ time.

\begin{definition}[Slack]

    The \emph{slack} of a bin $j$ is $c_j - \sum_{i \in S_j} w(i)$. A bin has
    \emph{$\delta$-large slack} if it has slack at least $n \cdot 2^{l -
    \prun}$ and \emph{$\delta$-small slack} otherwise. 
    An item is a \emph{large} \emph{slack item} if it is in a bin of large slack and a
    \emph{small slack item} otherwise.
\end{definition} 

We often omit $\delta$ in the above notation, because $\delta$ will be fixed later in Subsection~\ref{subsec:mainproof}.
    
\subsection{Solving Case C: Balanced Solution with Many Small Slack Items} \label{sec:manysmallslack}
In the next lemma we will solve Bin Packing instances with at least $(1/2-\alpha)n$ small slack items. We will use this algorithm in Case C.

\begin{case}
	\textbf{In Case C:}\begin{minipage}[t]{0.8\linewidth}
		\begin{itemize}
			\item[(\ding{51})] $|w(2^{\setrng{n}})| > 2^{\delta n}$.
			\item[(\ding{51})] The instance has only $\alpha$-balanced solutions.
			\item[(\ding{51})] There are at least $(1/2-\alpha)n$ small slack items  with respect to $\prun$.
		\end{itemize}
	\end{minipage}
\end{case}

\begin{lemma} \label{lem:CaseC} 
	Suppose $|w(2^{\setrng{n}})| \ge 2^{\delta n}$ and $0 < \alpha \le
	2^{-2/\delta^3}$. If a Bin Packing instance has a solution that is
	$\alpha$-balanced and has at least $(1/2 - \alpha)n$ items with
	$\delta$-small slack, then such a solution can be found in time $\Os(2^{(1-f_C(\delta))n})$ for $f_C(\delta) = \Omega_{\delta \to 0} \left(2^{-3/\delta^3} \right)$.
	
\end{lemma}

\begin{proof} 
	Use Corollary~\ref{cor:critPruner} to compute the critical pruner $\prun$ in
	time $\Os(2^{\delta n})$. Then iterate over all combinations of sets $L,R
	\subseteq \setrng{m}$ that form a partition of $\setrng{m}$. For each such a partition, the algorithm searches for $(L,R)$-witnesses of size $[\frac{n}{2}\pm \alpha n]$ as follows:
	First, enumerate $\cW$, which is defined as
    \[\cW := \left \{ W \subseteq \setrng{n} : ||W| - \tfrac{n}{2}| \le \alpha n,\;
    \left(\sum_{j \in L} c_j/2^{l-\prun} - w_\prun(W) \right) \in \left [0, n \cdot (|L|+1) \right] \right\}. \]
	
	We can enumerate $\cW$ in time $\Oh(2^{n/2}+|\cW|)$ with a standard Meet-in-the-Middle approach (see e.g. \cite[Section 3.2]{generic-knapsack2} or Lemma~\cite[Lemma 3.8]{MuchaNPW19}).
	Next, for every $W\in \cW$ we determine whether $W$ is an $(L,R)$-witness. This is done by computing the boolean $l_X$ for every $X \in \dc\cW$ and $r_X$ for every $X \in \uc\cW$, where 
	\begin{align*}
		l_X &:= \begin{cases} \mathtt{true} & \text{ if }X \text{ can be divided
				over the bins in } L, \\ \mathtt{false} & \text{ otherwise,}\end{cases} \\
		r_X &:= \begin{cases} \mathtt{true} & \text{ if } \setrng{n}\setminus X
			\text{ can be divided over the bins in } R, \\ \mathtt{false} & \text{ otherwise.} \end{cases}
	\end{align*}
	
	Using Theorem~\ref{thm:ZetaMobBins} we can do this in time
	$\Oh((|\dc\cW|+|\uc\cW|)n)$. Next, the algorithm checks for all $W\in \cW$,
	whether $l_W = r_W = \mathtt{true}$, and if so the algorithm returns \emph{yes}. If for
	no partition $L,R$ of $\setrng{m}$ the algorithm finds a witness, the algorithm
	returns \emph{no}. 
	
	\subsubsection*{Correctness of Algorithm} 
	Assume that there is an $\alpha$-balanced solution $S_1,\ldots,S_m$. Let
	$\pi: \setrng{m} \to \setrng{m}$ be a permutation of the bins such that all bins with
	small slack have smaller index than the large slack bins, i.e. $ \pi(j)
	\le \pi(j')$ for all small slack bins $j$ and large slack bins $j'$.
	Since we assumed the solution to be $\alpha$-balanced, there exists a
	bin $b \in \setrng{m}$ such that $\sum_{j = 1}^{b} |S_{\pi(j)}| \in
	[\frac{n}{2} \pm \alpha n]$. Take $L = \{\pi(1),\dots,\pi(b)\}$ and $R =
	\{\pi(b+1),\dots,\pi(m)\}$. Notice that $S^L = \bigcup_{j=1}^b S_{\pi(j)}$
	is an $(L,R)$-witness. We will prove that in the iteration of the
	algorithm where the correct partition $L,R$ is chosen, it holds that $S^L \in \cW$. Since there are at least $(1/2 - \alpha)n$ small slack items, all bins in $L$ have small slack. Hence,
	\[\sum_{j=1}^b (c_{\pi(j)} - n\cdot2^{l-\prun}) \le w(S^L) \le \sum_{j=1}^b c_{\pi(j)}. \]
	We also have the bound on the pruned weights of the items
	\[ w_\prun (S^L) \le w(S^L) / 2^{l-\prun} \le \sum_{j=1}^b c_{\pi(j)}/2^{l-\prun}. \]
    On the other hand:
	\begin{align*}
	w_\prun(S^L) &\ge w(S^L)/2^{l-\prun} -n \\
	&\ge \sum_{j=1}^b ( c_{\pi(j)}/2^{l-\prun} - n\cdot 2^{l-\prun}) - n\\ &\ge \sum_{j=1}^b  c_{\pi(j)}/2^{l-\prun} - (b+1) \cdot n	.
	\end{align*}
 Combining this with the fact
	that $|S^L| \in [\frac{n}{2} \pm \alpha n]$, we conclude that the set $S^L$ 
	is present in
	\[\left \{ W \subseteq \setrng{n} : |W| \in [(\tfrac{1}{2} \pm
	\alpha)n],  w_\prun(W) \in \left [\sum_{j=1}^b c_j/2^{l-\prun} - (b+1)
	\cdot n, \sum_{j=1}^b c_j/2^{l-\prun} \right ] \right\} \]
    which matches the definition of $\cW$.
	
    Notice that $l_W = r_W = \mathtt{true}$ if and only if $W$ is an
    $(L,R)$-witness, since we chose $L$ and $R$ to partition $\setrng{m}$.
    Because $S^L\in \cW$, the algorithm always returns \emph{yes} in a
    yes-instance. Furthermore, when we find a $W$ s.t., $l_W =  r_W
    =\mathtt{true}$, we can conclude that there is a solution to the Bin Packing
    instance since all items are divided over all bins. 
	
	\subsubsection*{Running time Analysis}
	It remains to analyze the running time of the algorithm.
	Recall that the algorithm iterates over all $m \cdot 2^{m-1}$ combinations of $L$ and $R$. Each iteration takes $\Oh((|\dc \cW|+|\uc \cW|)n + 2^{n/2})$ time.
	
	Before we can prove that $|\dc \cW|+|\uc \cW| \le 2^{(1-f_C(\delta))n}$, we need to bound the size of $\cW$. 		
	Recall that $\prun$ is the critical pruner, and therefore by definition
	$|w_\prun(2^{\setrng{n}})|\ge 2^{\delta n}$. 
	Theorem~\ref{lothm} states that if $\beta(w_\prun) \ge 2^{(1-\eps')n}$,
	then $|w_\prun(2^{\setrng{n}})| \le 2^{\delta' n}$ where $\delta' = \Oh_{\eps' \to 0} \left(\frac{\log\log(1/\eps')}{\sqrt{\log(1/\eps')}}\right)$. 
	
	Because we assume that $\eps' \to 0$ we use a crude bound $\frac{\log\log(1/\eps')}{\sqrt{\log(1/\eps')}} \le
	\frac{1}{\sqrt[3]{\log(1/\eps')}}$ to guarantee that 
	\begin{displaymath}
		\delta' \le \Oh_{\eps' \to 0}\left(\frac{1}{ \sqrt[3]{\log(1/\eps')}}\right) 
		.
	\end{displaymath}
	Next, we manipulate this inequality to get:
	\begin{equation}\label{eq:deltaeps}
		\eps' \le \Oh_{\delta' \to 0}
		\left(2^{-\left(\frac{1}{\delta'}\right)^3}\right)
		.
	\end{equation}
	Now, we denote $\eps(\delta) \coloneqq
	2^{-\left(\frac{1}{\delta}\right)^3}$. By~\eqref{eq:deltaeps} there exists
	a constant $\delta_{0}> 0$ such that for all $\delta < \delta_0$ it
	holds that if $|w_\prun(2^{\setrng{n}})|\ge 2^{\delta n}$, then $\beta(w_\prun) \le 2^{(1-\eps(\delta))n}$.
	Note that because we only claim an asymptotic time bound in the lemma, we may assume that $\delta \leq \delta_0$.
	As a consequence, for fixed weight value $v$, there are at most
	$2^{(1-\eps(\delta))n}$ sets $W\subseteq \setrng{n}$ that have a weight
	$w_\prun(W)=v$. Because for all $W \in \cW$ it holds that $w_\prun(W) \in \left [\sum_{j=1}^b c_j/2^{l-\prun} - (b+1)
	\cdot n, \sum_{j=1}^b c_j/2^{l-\prun} \right ]$ there at at most $m n$ such weights $v$ and so $|\cW|\le mn\cdot2^{(1-\eps(\delta))n}$.
	
	Therefore, we have a set $\cW$ of size $\le mn\cdot2^{(1-\eps(\delta))n}$ and each set $W \in \cW$ has a size $|W|\in[\frac{n}{2}\pm \alpha n]$. Lemma~\ref{lem:boundeddownset} then bounds the sizes of $\dc \cW$ and $\uc \cW$:
	\[|\dc \cW| + |\uc \cW| \le \Os\left(2^{(1-\rho(\varepsilon(\delta),  \alpha))n}\right),\] 
	where $\rho(\varepsilon(\delta), \alpha) = \frac{2}{\ln(2)} \left( \frac{\varepsilon(\delta)}{4\log(12/\varepsilon(\delta))} - \alpha \right )^2$. 
	Recall that we ensured that $0 < \alpha \le 2^{-2/\delta^3}$. Hence, there is a small enough constant $\delta_0$ such that for any $\delta \in (0,\delta_0)$ we have 
	\[\alpha \le 2^{-2\delta^3} \le \frac{2^{-1/\delta^3}}{8\log(12/2^{-\delta^3})} = \frac{\varepsilon(\delta)}{8\log(12/\varepsilon(\delta))}.\]

	 Hence, $|\dc\cW|+|\uc\cW| = \Oh(2^{(1-f_C(\delta))n})$ where
	\begin{align*}
	f_C(\delta) &= \frac{2}{\ln(2)} \left( \frac{\varepsilon(\delta)}{4\log(12/\varepsilon(\delta))} - \frac{\varepsilon(\delta)}{8\log(12/\varepsilon(\delta))} \right )^2 \\
	&= \frac{\eps(\delta)^2}{32 \ln(2) (\log^2
		(12/\eps(\delta)))}\\
	&= \Omega_{\delta \to 0} \left(
	\frac{\eps(\delta)^2}{\log^2 (1/\eps(\delta))} \right) \\
	&= \Omega_{\delta \to 0} \left(
	\frac{2^{-2/\delta^3}}{\log^2(2^{(1/\delta)^3})} \right) \\
	&= \Omega_{\delta \to 0}
	\left( \delta^6 \cdot 2^{-2/\delta^3} \right).
	\end{align*}
	Since $f_C(\delta) \ll 1/2$ for $\delta \in [0,1]$, the
	$\Oh(2^{n/2})$ time bound is subsumed by the $\Oh(2^{(1-f_C(\delta))n})$
	term. Multiplying this term by $m 2^{m-1}$ different choices for $L$ and $R$, gives us the requested running time.  
\end{proof}

\subsection{Solving Case D: Detecting a Balanced Solution with Few Small Slack Items} \label{sec:fewsmallslack}

We are left to prove the remaining case, namely Case D, for which we assume that the following list of conditions holds.

\begin{case}
	\textbf{In Case D:} \begin{minipage}[t]{0.8\linewidth}
		\begin{itemize}
			\item[(\ding{51})] $|w(2^{\setrng{n}})| > 2^{\delta n}$.
			\item[(\ding{51})] Instance has only $\alpha$-balanced solutions.
			\item[(\ding{51})] There are at most $(1/2-\alpha)n$ small slack items with respect to $\prun$.
		\end{itemize}
	\end{minipage}
\end{case}

First, we observe the following property of an $\alpha$-balanced solution.

\begin{observation}
	\label{lem:largebins}
	Let $S_1,\dots,S_m$ be an $\alpha$-balanced solution for some
	$0<\alpha<\frac{1}{4m}$. Assume $k, k' \in \setrng{m}$ to be two different bins with the most items. Then either:
	\begin{enumerate}
		\item $|S_{k}|, |S_{k'}| \in [(1/2 \pm \alpha)n]$, or
        \item $|S_j| \le (\frac{1}{2}- \frac{1}{4m} )n$ for all bins $j \in \{1,\dots,m\}$.
	\end{enumerate}
\end{observation}
\begin{proof}
	Let $S_k$ be the bucket with the highest number of items and $S_k'$ be the
	second highest number of items.
	If condition \emph{(2)} does not hold, then we know that $|S_k| >
	(\frac{1}{2} - \frac{1}{4m})n$. The number of the items in the remaining
	buckets is $n - |S_k|$ and the number of remaining buckets is $m-1$.
	Therefore the second bucket with the highest number of items has a size
	$|S_{k'}| \ge \frac{n - |S_k|}{m-1}$.
	
	Hence, 
	\begin{align*}
		|S_k|+|S_{k'}| &\ge \left (\frac{1}{2} - \frac{1}{4m} + \frac{1 - \frac{1}{2} + \frac{1}{4m}} {m-1} \right)\cdot n > \left( \frac{1}{2} +
    \frac{1}{4m}\right)\cdot n.
	\end{align*}
	Since the solution is $\alpha$-balanced (with $\alpha \le \frac{1}{4m}$), it
	means that for all permutations $\pi : \setrng{m} \rightarrow \setrng{m}$, in particular those with $\pi^{-1}(k)=1$
	and $\pi^{-1}(k')=2$, there exists $b \in \setrng{m}$ such that $\sum_{j = 1}^b
	|S_{\pi(j)}| \in [\frac{n}{2} \pm \alpha n]$. Because $|S_k|+|S_{k'}| > (1/2 +
	\alpha)n$, we know that $b=1$ and $|S_{k}| \in [\frac{n}{2} \pm \alpha n]$. We
	can conclude the same for $k'$ by repeating these last arguments for all
	permutations $\pi$ with $\pi^{-1}(k) =2$ and $\pi^{-1}(k') =1$, and thus
	condition \emph{(1)} must hold, and the observation follows.
\end{proof}

\begin{lemma} \label{lem:CaseD} Assume $\alpha < \frac{1}{4m}$. If a solution of a Bin Packing instance with $m$ bins is $\alpha$-balanced and has at most $(1/2 - \alpha)n$ items that have $\delta$-small slack, then with probability at least $\frac{1}{2}$ a solution can be found in time $\Os(2^{(1-f_D(m)+\delta m)n})$ with 
	$$f_D(m) = \Omega_{m \to \infty}\left(
	\frac{h(\frac{1}{2m})^2}{\log^2(h(\frac{1}{2m}))}\right).$$
\end{lemma}

\begin{proof}
	For an overview of the algorithm, see Algorithm~\ref{alg:CaseD}.
	Compute the critical pruner $\prun$ and the set $w_\prun(2^{\setrng{n}})$ in
	$\Os(2^{\delta n})$ time with Corollary~\ref{cor:critPruner}. The algorithm will
	search for $(L,R)$-witnesses for all $L,R \subseteq \setrng{m}$ such that $|R|=1$ and
	$L\cap R = \emptyset$. For notation purposes, assume without loss of generality
	that $R=\{1\}$, $L = \{2,\dots,k\}$ and let $M = \{k+1,\dots,m\}$. 
	
     Let $\cW$ be obtained by sampling $2^{(1-g(m))n}$ random subsets (with
     replacement, and removing repeating items at the end) of $\setrng{n}$ of size $\frac{n}{2}$, where $g(m) \coloneqq \frac{1}{2}h(1/(2m))$. For a given $L$ and $R$, we
	guess $a_{k+1},\dots,a_m \in w_\prun(2^{\setrng{n}})$. Then, we compute the boolean $l_X$ for every $X \in \dc\cW$ and $r_X$ for every $X \in \uc\cW$, where
	
	\begin{align*}
		l_X &\coloneqq \begin{cases} \mathtt{true} & \begin{aligned} &\text{if there exists a
					partition } X_2,\dots X_{k},Y_{k+1},\dots,Y_m \text{ of } X \text{ such
					that}\\ &\text{for all } j \in L: w(X_j) \le c_j \text{ and
                    for all } i \in M: w_\prun(Y_{i}) \le a_{i}, \end{aligned}\\ \mathtt{false} & \text{otherwise,}\end{cases}\\
		r_X &\coloneqq \begin{cases} \mathtt{true} & \begin{aligned} &\text{if there exists a
					partition } X_1,Y_{k+1},\dots,Y_m \text{ of } \setrng{n}\setminus X \text{ such
					that}\\ & w(X_1) \le c_1 \text{ and for all } i \in M: w_\prun(Y_{i}) \le
				c_{i}/2^{l-\prun} - n -a_{i},\end{aligned}\\ \mathtt{false} & \text{otherwise.}\end{cases}
	\end{align*}
	
	To compute $l_X$ and $r_X$ we can use the fast zeta transformations. For $j \in L\cup R$
    define the functions $\mathtt{fit}_j:2^{\setrng{n}} \to \{0,1\}$ as
    \[\mathtt{fit}_j(X) = \begin{cases} 1, &\text{ if }w(X) \le c_j, \\
		0, &\text{ otherwise.}\end{cases}\] 
        For $i \in M$ define the functions $\mathtt{lfit}_i,\mathtt{rfit}_i:
	2^{\setrng{n}} \to \{0,1\}$ as  
	\begin{align*}
        \mathtt{lfit}_i(X) &= \begin{cases} 1, &\text{ if } w_\prun(X) \le a_{i}, \\ 0, &\text{ otherwise,}\end{cases}\\
        \mathtt{rfit}_i(X) &= \begin{cases} 1, &\text{ if } w_\prun(X) \le
			c_{i}/2^{l-\prun} - n - a_{i}, \\ 0, &\text{ otherwise.}\end{cases}
	\end{align*}
	
	Now, we observe the following:
	\begin{claim}
        $l_X=\mathtt{true}$ if and only if $(\mathtt{fit}_2\coverprod \cdots
            \coverprod \mathtt{fit}_k \coverprod
        \mathtt{lfit}_{k+1} \coverprod \cdots \coverprod \mathtt{lfit}_m) (X) >0$.
	\end{claim}
	\begin{proof}
		Let us assume that $l_X = \mathtt{true}$. Let
        $S'_2,\ldots,S'_m$ be such that $\mathtt{fit}_j(S'_j) =1$ for every $j$.
        Then $S'_2,\ldots,S'_m$ gives a non-zero contribution to $(\mathtt{fit}_2\coverprod \cdots
            \coverprod \mathtt{fit}_k \coverprod \mathtt{lfit}_{k+1} \coverprod \cdots \coverprod
        \mathtt{lfit}_m) (X)$ and hence it must be positive.
		
        For the other direction, if $(\mathtt{fit}_2\coverprod \cdots \coverprod
            \mathtt{fit}_k \coverprod
            \mathtt{lfit}_{k+1} \coverprod \cdots \coverprod
        \mathtt{lfit}_m) (X) >0$ there exist $S'_2,\ldots,S'_m$ such that $\mathtt{fit}_j(S'_j)
		=1$ for every $j$ and $S'_2 \cup \ldots \cup S'_m = X$. Observe that we can
		transform this into a partition $S^{''}_2,\ldots,S^{''}_m$ of $X$ by choosing $S^{''}_j \subseteq S'_j$. 
        Because $\mathtt{fit}_j$ does not decrease when taking subsets we know that $1
        =\mathtt{fit}_j(S^{'}_j) \le \mathtt{fit}_j(S^{''}_j)$, and thus $l_X=\mathtt{true}$. 
	\end{proof}
	
	Similarly, we can argue that $r_X=\mathtt{true}$ if and only if
    $(\mathtt{fit}_1
    \coverprod \mathtt{rfit}_{k+1} \coverprod \cdots \coverprod \mathtt{rfit}_m)(\setrng{n}\setminus X) > 0$.
	
	We can compute booleans $l_X$ and $r_X$ in time $\Oh((|\dc \cW|+|\uc\cW|)n)$
	by combining Theorem~\ref{Thm:zetamob} and Theorem~\ref{thm:coverproductzeta}.
	Finally, if we find $W \in \cW$, such that $l_W=r_W = \mathtt{true}$, we can return \emph{yes}.
	
	\begin{algorithm}
		 \SetAlgoLined
		\DontPrintSemicolon
		\SetKwInOut{Input}{Algorithm}
		\SetKwInOut{Output}{Output}
		\Input{BinPacking($w_1,\ldots,w_n$)}
        \Output{Yes (whp.), if an $\alpha$-balanced solution with $(1/2 - \alpha)n$ small slack items exist}
		Compute the critical pruner $\prun$ with Corollary~\ref{cor:critPruner}. \tcp*{In time $\Os(2^{\delta n})$}
		Compute the set $w_\prun(2^{\setrng{n}})$ with Lemma~\ref{lem:computesums}. \tcp*{In time $\Os(2^{\delta n})$}
		Let $\cW$ be a set of $2^{(1-g(m))n}$ random subsets of $n$ of size $\frac{n}{2}$.\\
		\For(\tcp*[f]{$m2^{m-1}$ repetitions}){$L,R \subseteq \setrng{m}$ \textnormal{such that} $|R| =1$ \textnormal{and} $L\cap R = \emptyset$}{
			Assume without loss of generality that $R=\{1\}$, $L = \{2,\dots,k\}$.\\
			\For(\tcp*[f]{$|w_\prun(2^{\setrng{n}})|^{m-k}$
				repetitions}){$a_{k+1},\dots,a_m \in w_\prun(2^{\setrng{n}})$}{
				Compute $l_X$ for all $X \in \dc\cW$. \tcp*{In time $\Oh((|\dc\cW| + |\uc\cW|)n)$} 
				Compute $r_X$ for all $X \in \uc\cW$. \tcp*{In time $\Oh((|\dc\cW| + |\uc\cW|)n)$}
                \If{$l_W = r_W =\mathtt{true}$, for some $W \in \cW$}{\Return \emph{yes}}
			}	
		}
        \Return \emph{no}\\
		\caption{Overview of the algorithm for Lemma~\ref{lem:CaseD}}
		\label{alg:CaseD}
	\end{algorithm}
	
	\subsubsection*{Constant probability of a witness in $\cW$}
	Recall that the set $\cW$ is a random subset of $\binom{\setrng{n}}{n/2}$ of size
	$2^{(1-g(m))n}$ with $g(m) = h(1/(2m))/2$. We first analyze the number of
	$(L,R)$-witnesses that are in $\binom{\setrng{n}}{n/2}$. 
	Assume that there is an $\alpha$-balanced solution $S_1,\dots,S_m$ for some $0<\alpha<\frac{1}{4m}$. 
	We use Observation~\ref{lem:largebins} to conclude that either $|S_j|<(\frac{1}{2} -\frac{1}{4m})$ for all bins, or that $|S_k|,|S_{k'}| \in [\frac{n}{2} \pm \alpha n]$ for largest bins $k$ and $k'$. Since we assumed that there are at most $(1/2-\alpha)n$ small slack items, we know that in the latter case bins $k$ and $k'$ are therefore large slack bins. In either case, we conclude that $|S_j|\le (\frac{1}{2} - \frac{1}{4m})n$ for all small slack bins. 
	We will assume without loss of generality that bin $1$ is the largest small slack bin and that bins $2,\dots,k$ are the other small slack bins. Then let $L=\{2,\dots,k\}$, $R = \{1\}$ and thus $M=\{k+1,\dots,m\}$ are all large slack bins. We will lower bound the number of $(L,R)$-witnesses of size $\frac{n}{2}$. 
	\begin{figure}[h]
    \centering
     \begin{tikzpicture} [scale=0.6, every node/.style={scale=0.99}] 
		\draw [ fill = white!80!black] (0,0) rectangle (2,1.5) node[pos=.5] {$S_2$};
		\draw [ fill = white!80!black] (2,0) rectangle (4,1.5) node[pos=.5] {$\cdots$};
		\draw [ fill = white!80!black] (4,0) rectangle (5,1.5) node[pos=.5] {$S_k$};
		\draw [ fill = white!100!black] (5,0) rectangle (15,1.5) node[pos=.5] {$M$};
		\draw [ fill = white!80!black] (15,0) rectangle (18,1.5) node[pos=.5] {$S_1$};
		\draw [dashed] (9,-1) -- (9,3);
		\node at (9,3.5) {$\frac{n}{2}$ items};
		
		\draw[decorate,decoration={brace,mirror,raise=4pt,amplitude=8pt}]  (0,0)--(5,0) ;
		\node at (2.5,-1.3) {$L$};
		\draw[decorate,decoration={brace,mirror,raise=4pt,amplitude=8pt}]  (15,0)--(18,0) ;
		\node at (16.5,-1.3) {$R$};
		\draw[decorate,decoration={brace,raise=4pt,amplitude=8pt}]  (0,1.5)--(5,1.5) ;
		\node at (2.5,2.8) {$x - |S_1|$ items};
		\draw[decorate,decoration={brace,mirror,raise=4pt,amplitude=8pt}]  (0,-1.5)--(9,-1.5) ;
		\node at (4.5,-2.8) {$W$};
		\draw[decorate,decoration={brace,raise=4pt,amplitude=8pt,aspect=0.8}]  (5.2,1.5)--(14.8,1.5) ;
		\node at (13,2.8) {$n -x$ items};

		
		

        

    \end{tikzpicture}
    
    \caption{Overview of $(L,R)$-witnesses of size $\frac{n}{2}$ for explanation of equation~\ref{eq:nrwitnesses}. Let $x$ be the number of small slack items. Then any such $(L,R)$-witness, $W$, must include all items of $S_2,\dots,S_k$ and exclude any items of $S_1$. The other items in $W$ can then be any combination of large slack items, which are exactly the items in the bins of $M$.}
    \label{tikz:Witnesses}
\end{figure}

	Let $x$ be the number of small slack items in the solution. Note that the number of $(L,R)$-witnesses of size $\frac{n}{2}$ is equal to \begin{equation} \label{eq:nrwitnesses}
		\binom{n-x}{n/2-(x-|S_1|)},
	\end{equation} since the sets $S_2,\dots,S_k$ together with any subset of $n/2-(x-|S_1|)$ large slack items form a witness. See Figure~\ref{tikz:Witnesses} for an illustration of this. Now, we analyse two cases. If $x \leq n/4$, then 
	because $\binom{n-x}{n/2-(x-|S_1|)}\ge  \binom{n-x}{n/2-x}$ there are at least $\binom{3n/4}{n/4}$  $(L,R)$-witnesses of size $\frac{n}{2}$. 

	In the case when $x \ge \frac{n}{4}$, then notice that $\frac{x}{m} \le |S_1| \le (\frac{1}{2} - \frac{1}{m})n$ (because $S_1$ is the largest among the small slack bins). Therefore, the number of witnesses is at least the number of ways to choose $n/2-|S_1|$ items from $M$ to exclude in the witness. Thus we have that the number of witnesses of size $\frac{n}{2}$ is at least
	\[ \binom{n-x}{\frac{n}{2} - |S_1|} \ge \min  \left \{ \binom{\frac{n}{2}}{\frac{n}{2} - \frac{x}{m}}, \binom{\frac{n}{2}}{\frac{n}{2} - (\frac{n}{2} - \frac{n}{m})} \right \} \ge \binom{\frac{n}{2}}{\frac{n}{4m}}. \]
	So in both cases for $x$, we can conclude that the number of $(L,R)$-witnesses of size $\frac{n}{2}$ is at least $2^{g(m)n}$.
	
	Because $\cW$ is a random subset of $\binom{\setrng{n}}{n/2}$ of size $2^{(1-g(m))n}$ and at least $2^{g(m)n}$ of those are $(L,R)$-witnesses, Observation~\ref{lem:probabilities} tells us that with probability at least $\frac{1}{2}$, $\cW$ contains an $(L,R)$-witness. 
	
	\subsubsection*{Correctness of Algorithm}
	The algorithm returns \emph{yes} if and only if $l_W =r_W = \mathtt{true}$ for some $W\in
	\cW$. So if it returns \emph{yes}, there exists a partition
	$X_2,\dots,X_k,Y_{k+1},\dots,Y_m$ of $W$ and a partition
	$X_1,Y'_{k+1},\dots,Y'_m$ of $\setrng{n}\setminus W$ by definition. Together they
	partition all items. Notice that by definition we know that $X_j$ can be put
	into bin $j$ for all $j\in L\cup R$. Hence we are left to prove that for all $j
	\in M$: $X_j = Y_{j}\cup Y'_{j}$ can be put into bin $j$. Notice that since
    $\mathtt{lfit}_j(Y_j)=\mathtt{rfit}_j(Y'_j)=1$ we have that
	\begin{align*}
		\sum_{i \in X_j} w_\prun (i) &\le a_j + c_j / 2^{l-\prun} - n - a_j &\implies \\
		\sum_{i \in X_j} \lfloor w(i)/2^{l-\prun} \rfloor  &\le c_j / 2^{l-\prun}  -n &\implies \\
		\sum_{i \in X_j} (w(i) - 2^{l - \prun}) &\le c_j -n2^{l - \prun} &\implies\\
		\sum_{i \in X_j} w(i) &\le c_j, &
	\end{align*}
	and so, indeed the items of $X_j$ fit into bin $j$ and we have a yes-instance.
	For the implication in the other direction, we prove that if there exists a
	solution, the algorithm finds it with constant probability. We already showed
	that with constant probability there is an $(L,R)$-witness $W \in \cW$. Next, we
	will prove that there exist $a_{k+1},\dots,a_m \in w_\prun(2^{\setrng{n}})$ such that
	$l_W = r_W=1$ for all witnesses $W$. 
	Let $S_2,\dots,S_k,S_{k+1}\cap W,\dots S_m \cap W$ be the partition of $W$ 
	from the definition of $l_W$, and let $S_1,\dots,S_{k+1}\setminus W,\dots,
	S_m\setminus W$ be the partition of $\setrng{n}\setminus W$ from the definition of $r_W$. 
	
	Note that, for all $j \in L\cup R$ it holds that $w(S_j) \le c_j$ because
	$S_1,\dots,S_m$ is a solution. So we are left to prove that for all $j \in M$
	there exists an $a_j \in w_\prun(2^{\setrng{n}})$ such that
	\[  w_\prun(S_j \cap W) \le a_{j} \qquad \text{and} \qquad  w_\prun(S_j
	\setminus W) \le c_j/2^{l-\prun} -n - a_{j}. \]
	Recall that we assumed that the bins of $M$ are large slack bins. Hence we know that for $j \in M$:
	\begin{align*}
		\sum_{i \in S_j} w(i) &\le c_j - n2^{l-\prun} &\implies\\
		\sum_{i \in S_j} \lfloor w(i)/ 2^{l-\prun} \rfloor &\le c_j/2^{l-\prun} -n &\implies \\
		\sum_{i \in S_j} w_\prun(i) &\le c_j/2^{l-\prun} -n - a_j + a_j.
	\end{align*} 
	So, take $a_j = w_\prun(S_j\cap W) \in w_\prun(2^{\setrng{n}})$ and indeed the correctness of the algorithm follows. 
	\subsubsection*{Running Time Analysis}
	The algorithm will go through the procedure of computing the booleans $l_X$ and
	$r_X$ for all different sets $L,R\subseteq \setrng{m}$ such that $|R|=1$ and for all
	different values of $a_{k+1},\dots,a_{m}\in w_\prun(2^{\setrng{n}})$. This gives a
	total of at most $m\cdot2^{m-1} \cdot |w_\prun(2^{\setrng{n}})|^m$ repetitions. By
	Lemma~\ref{lem:relationps}, we have $|w_\prun(2^{\setrng{n}})| \le
	3n|w_{\prun-1}(2^{\setrng{n}})|$. Because $\prun$ is the critical pruner, and since
	$w_0(2^{\setrng{n}})=\{0\}$, we know that $|w_{\prun-1}(2^{\setrng{n}})|\le 2^{\delta n}$. Hence, the number of repetitions is at most $\Oh( (2n)^m\cdot 2^{\delta m n})$. 
	
	Now we analyze the time complexity per the choice of $(L,R)$ and
	$a_{k+1},\dots,a_m$. Recall that we chose $\cW \subseteq \binom{\setrng{n}}{n/2}$ as a
	random set of size $2^{(1-g(m))n}$. Computing all the booleans $l_X$ and $r_X$
	can be done in $\Oh((|\dc\cW|+|\uc\cW|)n)$ time. We can use Lemma~\ref{lem:boundeddownset} to find that 
	\[|\dc\cW|+|\uc\cW| =\Oh(2^{(1-\rho(g(m),0))n}), \] where 
		\begin{align*}\rho(g(m),0) &= \frac{2}{\ln(2)} \cdot \frac{g(m)^2}{16 (\log
			(12/g(m)))^2} \\
		&= \frac{h(\frac{1}{2m})^2}{32\ln (2) \log(24/h(\frac{1}{2m}))^2} \\ 
		&= \Omega_{m \to \infty} \left( \frac{h(\frac{1}{2m})^2}{\log^2(h(\frac{1}{2m}))} \right).
	\end{align*}
	Combining this with the number of repetitions we get a running time of
	$\Os(2^{(1-f_D(m)+\delta m)n})$. This gives us the claimed running time.
\end{proof}

\subsection{Proof of Theorem~\ref{mainthm}} 
\label{subsec:mainproof}
We are now ready to prove Theorem~\ref{mainthm} by combining all work of the previous sections and setting the parameters $\alpha$ and $\delta$: 

\begin{proof}
	
	We will now combine all previous lemmas.
	An overview of the algorithm can be found in Figure~\ref{tikz:MainThmProof}.
	To facilitate the asymptotic analysis, note we can assume the number of bins $m$ is at least $m_0$ for some constant $m_0$. If this is not the case we can add $m_0-m$ artificial bins with unique small capacities and matching items. Since $m_0$ is constant this does not influence the asymptotic running time of the algorithm.
	Define $f_D(m)$ as in Lemma~\ref{lem:CaseD} as:
	\[f_D(m) =
	\frac{h(\frac{1}{2m})^2}{32\ln (2) \log(24/h(\frac{1}{2m}))^2} =
	\Omega_{m \to \infty} \left( \frac{h(\frac{1}{2m})^2}{\log^2(h(\frac{1}{2m}))}\right).\]
	Then, set $\delta :=  f_D(m)/(2m)$ and $\alpha:= 2^{-2/\delta^3}$. Then $f_D(m)>0$, $\delta >0$ and $ \alpha >0$. 
	\begin{enumerate}
		\item If $|w(2^{\setrng{n}})| \leq 2^{\delta n}$ (Case A), the algorithm from Lemma~\ref{lem:CaseA} solves the instance in time 
		\[\Os(|w(2^{\setrng{n}})|^m) = \Os\left(2^{(\delta n) m}\right) = \Os\left(2^{\frac{f_D(m)}{2} n}\right).\]
		\item If the instance has an $\alpha$-unbalanced solution (Case B), the algorithm from Lemma~\ref{lem:CaseB} can detect with probability at least $\frac{1}{2}$ in time
		\[\Os\left(2^{\left(1-\Omega\left(\frac{\alpha^2}{\log^2(\alpha )}\right)\right)n}\right) = 
		\Os\left(2^{\left(1-\Omega\left(\delta^6 2^{-4/\delta^3}\right)\right)n}\right).
		\]
		\item If the instance has an $\alpha$-balanced solution,
		$|w(2^{\setrng{n}})| \leq 2^{\delta n}$, and a solution with at least $(1/2-\alpha)n$ small slack items (Case C), the upper bound $\alpha \leq 2^{-2/\delta^3}$ ensures that the solution can be detected by the algorithm from Lemma~\ref{lem:CaseC} in time
		\[
		\Os\left(2^{\left( 1-\Omega(2^{-3/\delta^3})\right)n}\right).
		\]
		\item Otherwise, if the instance has an $\alpha$-balanced solution,
		$|w(2^{\setrng{n}})| \leq 2^{\delta n}$, and a solution with at most $(1/2-\alpha)n$ small slack items (Case D), the algorithm from Lemma~\ref{lem:CaseD} detect the solution with probability at least $\frac{1}{2}$ in time
		\[
		\Os\left(2^{\left( 1- f_D(m) + \delta m\right)n}\right) = 
		\Os\left(2^{\left( 1-\frac{f_D(m)}{2}\right)n}\right).
		\]
	\end{enumerate}
	
	Thus, we obtain a probabilistic algorithm for Bin Packing that runs in time $\Os(2^{(1-\sigma_m)n})$, where $\sigma_m$ is a strictly positive number. Notice that any polynomial factor hidden in the $\Os$ notation can be subsumed by $\Oh(2^{(1-\sigma_m)n})$.
	
\end{proof}

\section{The Littlewood--Offord Theorem}\label{sec:LO}

In this section, we will prove our Additive Combinatorics result which we first restate for convenience:

\begin{theorem}
        Let $\eps >0$. If $\beta(w) \geq 2^{(1-\eps)n}$, then $|w(2^{\setrng{n}})|\leq
2^{\delta n}$, where 
\[\delta(\eps) = \Oh_{\eps \rightarrow
	0}\left(\frac{\log(\log(1/\eps))}{\sqrt{\log(1/\eps)}}\right).
\]
\end{theorem}

This theorem was recently improved by \cite{jain2021anticoncentration}, showing that $\delta(\eps) = \Oh(\sqrt{\eps})$.

For our proof, it will be convenient to use a reformulation of
Theorem~\ref{lothm} to a version with two set families that attain the
parameters and use vector notation (so $w$ is a vector and $w(X)$ is the inner
product $\langle w,x \rangle$ of $w$ with the characteristic vector $x$ of set
$X$):

\begin{theorem}[Theorem~\ref{lothm} reformulated]
    \label{thm:udcp}
    Let $w=(w_1,\ldots,w_n) \in \mathbb{Z}^n$ be a vector with integer weights, and let $A,B \subseteq \{0,1\}^n$ be such that $|a^{-1}(1)|=\alpha n$ for each $a \in A$ and
    \begin{itemize}
        \item $\langle w,b \rangle =\tau \textnormal{ for every } b \in B$, and
        \item $\textnormal{if } a, a'\in A \textnormal{ and } \langle w,a \rangle = \langle w,a' \rangle \text{, then } a=a'$.
    \end{itemize}
     If $|B| \geq 2^{(1-\eps)n}$, then $|A|\leq 2^{\delta(\eps) n}$, where
     $$\delta(\eps)= \Oh_{\eps \rightarrow
     0}\left(\frac{\log\log(1/\eps)}{\sqrt{\log(1/\eps)}}\right).$$
\end{theorem}

Recall, that for $a \in \{0,1\}^n$ we use $a^{-1}(1)$ to denote $\{ i \in \{1,\ldots,n\}
\; : \; a(i) = 1 \}$ (see~Section~\ref{sec:prel}).
We first show that this implies Theorem~\ref{thm:udcp}: 

\begin{proof}[Proof of Theorem~\ref{lothm} assuming Theorem~\ref{thm:udcp}]
    Suppose $w_1,\ldots,w_n$ and $\tau$ are such that $|\{X \subseteq \setrng{n}| \geq 2^{(1-\eps)n}$.
	Then $B \coloneqq \{b \in \{0,1\}^n : \langle w,b \rangle = \tau \}$ satisfies the conditions of Theorem~\ref{thm:udcp} and $|B| \ge 2^{(1-\eps) n}$.
    For every $i \in w(2^{\setrng{n}})$ arbitrarily choose a vector $a(i) \in \{0,1\}^n$ such that $\langle w,a(i) \rangle =i$.
    Define $A' = \{a(i) : i \in w(2^{\setrng{n}})\}$.
	Since $|a^{-1}(1)|$ can only take $n$ different values, there exists an $\alpha$ such that $|a^{-1}(1)|=\alpha n$ for at least $1/n$ fraction of the elements of $A'$. This gives a set $A$ that satisfies the condition of Theorem~\ref{thm:udcp},
	and thus
	\[
        |w(2^{\setrng{n}})| = |A'| \leq |A|\cdot n \leq 2^{\delta(\eps)n + o(n)}.
	\]
	Note that we use the $\Oh(\cdot)$ notation in the term $\delta(\eps)$ to hide
    the $2^{o(n)}$ factors. 

\end{proof}
The rest of this section is dedicated to the proof of Theorem~\ref{thm:udcp}.
We use the following standard definitions
from Additive Combinatorics: For sets $X,Y$ we define $X+Y$ as the sumset $\{x+y : x\in X, y \in Y\}$. For an integer $k$, we define $k\cdot X$ as the $k$-fold sum 
\begin{displaymath}
    k \cdot X \coloneqq \underbrace{X+X+\cdots+X}_{k \textnormal{ times}}.
\end{displaymath}

The starting point of the proof of Theorem~\ref{thm:udcp} is the following
simple lemma that proves that $|A||k \cdot B| = |A + k \cdot B|$. It is heavily inspired by the UDCP connection from~\cite[Proposition 4.2]{stacs2016}. 

\begin{lemma}\label{lem:UDCP}
  If $a,a' \in A$ and $b, b' \in k \cdot B$ are such that $a+b = a'+b'$, then $(a,b)=(a',b')$.
\end{lemma}
\begin{proof}
	Note that
\[
        \langle w,a \rangle  + \langle w, b \rangle = \langle w,a+b\rangle = \langle w, a'+b' \rangle = \langle w,a' \rangle + \langle w, b'\rangle.
\]
    By definition of $k\cdot B$, we know that $\langle w, b \rangle  = \langle w, b'\rangle = k
    \cdot \tau$, hence $\langle w, a\rangle = \langle w,a'\rangle$. Therefore by definition
    of set $A$ it has to be that $a' = a$. This implies that $b = b'$,
    since $a+b = a'+b'$.
\end{proof}

Thus $|A|$ is equal to $|A+k\cdot B|/ |k\cdot B|$, and we may restrict our attention to upper bounding the latter quantity for any integer $k \in \mathbb{N}$.
Since this is in general not easy, we instead define a set $P \subseteq A \times
k\cdot B$ of pairs such that for each $(a,b) \in P$ the distribution of the
values in the vector $a+b$ is close to what one would expect for random vectors.
This is useful since the control on pairs $(a,b) \in P$ gives us control on the vectors $a+b$ which allows us to upper bound $P$. Moreover, we also provide a lower bound that shows that $P$ is not much smaller than $|A| \cdot |k\cdot B|$. Combining the two bounds results in the upper bound for $A$.
We will make this more formal in the next subsections, but first, we give a warm-up result that sets up the notation for the main proof.

\subsection{A Warm-up with $B =\{0,1\}^n$}\label{subsec:warmup}
Let us first investigate what happens in a case when $B$ is equal to the whole Boolean hypercube $\{0,1\}^n$. While $|A|$ can be easily upper bounded by direct methods, it is instructive to see what our approach will be in this special case.
In this setting, we can think about vectors from $B$ as sampled uniformly at random. Fix
a parameter $0< \alpha < 1$, let $a \in \{0,1\}^n$ be a fixed, adversarially chosen
vector with $|a^{-1}(1)|=\alpha n$, and let $b_1,\ldots,b_k \in \{0,1\}^n$ be
independently sampled random vectors. Let $b=b_1+b_2+\ldots+b_k$ and $c=a+b$.
Observe that for every $i \in \{0,\ldots,k\}$ and $i' \in \{0,\ldots,k+1\}$
\[
    \Ex{\frac{|b^{-1}(i)|}{n}} =\binom{k}{i}2^{-k}, \textnormal{ and } \qquad \Ex{\frac{|c^{-1}(i')|}{n}} =\left((1-\alpha) \binom{k}{i'} + \alpha \binom{k}{i'-1}\right) 2^{-k},
\]
For further reference, we now define the found distributions
explicitly~\footnote{For convenience we assume that $\binom{n}{i} = 0$ when $i < 0$}:
\begin{definition}[Altered Binomial Distribution]
	For every $k \in \nat$, we let $\Bin(k)$ to denote the \emph{binomial
    distribution} $(\{0,\ldots,k\},p)$ where $p(i)=\binom{k}{i}2^{-k}$. 
    
    For an additional parameter $\alpha \in (0,1)$, we define the \emph{altered
    binomial distribution} $\Bin(k,\alpha)$ as $(\{0,\ldots,k+1\},p')$ where $p'(i) = (1-\alpha)\binom{k}{i}2^{-k}+\alpha\binom{k}{i-1}2^{-k}$.
\end{definition}
Note, that $\Bin(k+1) = \Bin(k,1/2)$ by Pascal's Formula.  
Now, we present the intuition for the random case. 
We have that:
\[
    n\cdot h(\Bin(k,\alpha)) =  h(c) = h(a,b) = h(a) + h(b) = h(a) + n\cdot h(\Bin(k)),
\]
where the second equality follows by Lemma~\ref{lem:UDCP} and the third equality follows because
$a$ and $b$ are independent. Thus $h(a) =n(h(\Bin(k,\alpha))- h(\Bin(k)))$, and
the proof in the random case can be concluded by using
Lemma~\ref{lem:consecbinomentro} in which we show that for any constant $\alpha
\in (0,1)$ it holds that $h(\Bin(k,\alpha))- h(\Bin(k)) = \Oh_{k\rightarrow \infty}((\log k)/
\sqrt{k})$, and the standard fact that the support of any uniform random
variable of entropy $h$ is at most $2^{h}$.

This concludes the analysis for the special case of $B \subset \{0,1\}^n$.

\subsection{Balanced Pairs}

Now, we consider the general setting where $B \subset \{0,1\}^n$. We need to
obtain sufficiently large sample of vectors $b_1,\ldots,b_k$. More precisely, we
will consider the setting with $\eps \le 1/2^{\Theta(k)}$, that will enforce $k
\coloneqq \Theta(\log(1/\eps))$. The following chain of (in-)equalities summarizes the
strategy of our proof.

\begin{equation*}
  |A|
  \underset{\mathllap{
    \begin{tikzpicture}
      \draw[->] (-0.3, 0) to[bend right=20] ++(0.5,5ex);
      \node[below left] at (0,0) {Lemma~\ref{lem:UDCP}};
    \end{tikzpicture}
  }}{=}
  \frac{|A + k\cdot B|}{|k \cdot B|}
  \underset{\mathllap{
    \begin{tikzpicture}
      \draw[->] (-0.3, 0) to[bend right=20] ++(0.5,5ex);
      \node[below left] at (0,0) {Section~\ref{sec:balanced-pairs}};
    \end{tikzpicture}
  }}{\le}
  \frac{|P|}{|k \cdot B|} \cdot 2^{f(\eps,k)n}
  \underset{\mathllap{
    \begin{tikzpicture}
      \draw[->] (-0.3, 0) to[bend right=20] ++(0.5,5ex);
      \node[below left] at (0,0) {Section~\ref{sec:proof-udcp}};
    \end{tikzpicture}
  }}{\le}
	2^{n (h(\Bin(k+1)) - h(\Bin(k)))} \cdot 2^{f(\eps,k)n}
  \underset{\mathllap{
    \begin{tikzpicture}
      \draw[->] (-0.3, 0) to[bend right=20] ++(0.5,5ex);
      \node[below left] at (0,0) {Lemma~\ref{lem:consecbinomentro}};
    \end{tikzpicture}
  }}{\le}
  2^{\delta(\eps)n}
\end{equation*}

Now, we will make the above idea more precise.

\label{sec:balanced-pairs}

The following definition quantifies the `sufficiently random' terms from the previous subsection by measuring how far the distribution of the values of a vector is from a given (expected) distribution.

\begin{definition}[Balanced vectors]\label{def:balvec}
	Let $\mathcal{D}=(\Omega,p)$ be a discrete probability space. Fix $\gamma \in (0,1)$.
    Let $U$ be the finite universe set and let $X \subseteq U$.
	A mapping (or a vector) $v\in\Omega^U$ is \emph{$\gamma$-$\mathcal{D}$ balanced for $X$} if for all $\omega\in\Omega$ it holds that
	\begin{displaymath}
	\frac{|v^{-1}(\omega) \cap X|}{|X|}\in[p(\omega)\pm\gamma].
	\end{displaymath}
	As a shorthand, we say that a mapping (or a vector) $v\in\Omega^U$ is \emph{$\gamma$-$\mathcal{D}$ balanced} if it is \emph{$\gamma$-$\mathcal{D}$ balanced for $U$}. We denote the set of all $\gamma$-$\mathcal{D}$ balanced vectors $v \in \Omega^U$ with $(\mathcal{D}\pm \gamma)^U$.
\end{definition}

As an illustration of Definition~\ref{def:balvec}, suppose
\[
	U = \{1,\ldots,6\}, \quad  X= \{1,\ldots,4\},  \quad \Omega=\{0,1\}, \quad  p(0)=p(1)=\tfrac{1}{2}, \quad \mathcal{D}=(\Omega,p).
\]
Then $(0,1,1,1,1,1)$ is $\tfrac{1}{4}$-$\mathcal{D}$ balanced for $X$ but not $\tfrac{1}{4}$-$\mathcal{D}$ balanced. The vector $(0,0,0,0,1,1)$ is not $\tfrac{1}{4}$-$\mathcal{D}$ balanced for $X$ but it is $\tfrac{1}{6}$-$\mathcal{D}$ balanced.

We will use Definition~\ref{def:balvec} with $\mathcal{D}$ being the distribution we would
get in the random case as outlined in Subsection~\ref{subsec:warmup} (hence
$\mathcal{D}$ will usually be $\Bin(k)$ or $\Bin(k,\alpha)$).
Now, we prove a general upper bound on the number of $\gamma$-$\mathcal{D}$ balanced
vectors.

\begin{lemma}\label{lem:boundbalancedvectors}
	Let $\mathcal{D} = (\Omega,p)$ be a discrete probability space.
	The number of $\gamma$-$\mathcal{D}$ balanced vectors is at most
    \begin{displaymath}
          2^{\left(h(\mathcal{D})+f(\Omega,\gamma) \right)|U|}, 
    \end{displaymath}
    where $f(\Omega,\gamma) \coloneqq \Oh(|\Omega| \cdot \gamma \log(1/\gamma))$.
\end{lemma}
\begin{proof}
    Let $d$ be the dimension of $\mathcal{D}$, i.e., $p = (p_1,\ldots,p_d) : \Omega \rightarrow (0,1)^d$.  The number of
    $\gamma$-$\mathcal{D}$ balanced vectors is at most 
    \begin{displaymath}
        \sum_{q_1,\ldots,q_d \in (0,1)} \binom{|U|}{q_1 |U|,\ldots,q_n |U|},
    \end{displaymath}
   where the sum is over all the probability distributions $q_i$ such that
   $q_i\cdot|U|$ is integer and $q_i(\omega) \in [p_i(\omega) \pm \gamma]$ for
   every $\omega \in \Omega$ and $i \in \{1,\ldots,d\}$.

	Since the number of possibilities for such a $q_i$ is at most
$|U|^{|\Omega|}$ and $\binom{|U|}{q_1|U|,\ldots,q_d |U|} \leq
2^{h(q_1,\ldots,q_d)|U|}$ by Lemma~\ref{lem:multvsentropy}, we obtain that the
number of $\gamma$-$\mathcal{D}$ balanced vector is at most
	\[
        |U|^{|\Omega|} 2^{h(q_1,\ldots,q_d)|U|} 
        \leq |U|^{|\Omega|} 2^{\left(h(\mathcal{D})+\frac{1}{\ln(2)}|\Omega|\gamma\log \tfrac{1}{\gamma} \right)|U|}
        \leq 2^{h(\mathcal{D}) |U|} \cdot 2^{\Oh\left(|U||\Omega| \gamma
        \log(\tfrac{1}{\gamma})\right)},
	\]
	where the second inequality follows from Lemma~\ref{lem:entclose} and the third comes from $|U|^{|\Omega|} \ll 2^{|U||\Omega|}$.
\end{proof}

For example, Lemma~\ref{lem:boundbalancedvectors} bounds the number of
$\gamma$-$\Bin(k)$ balanced vectors by $2^{nh(\Bin(k)) + n f(\gamma,k)}$ for
some positive function $f(\gamma,k)$ that goes to $0$ when $\gamma \rightarrow 0$.

With Definition~\ref{def:balvec} in hand, we are ready to define the set of
pairs mentioned at the beginning of this section:
\[
	P \coloneqq \{ (a,b) \in A \times k\cdot B: b \in B \textnormal{ is } \eps^{0.01}\textnormal{-}\Bin(k) \textnormal{ balanced for }a \}.
\]
In Section~\ref{sec:technical-section} we prove the following somewhat technical lemma: 

\begin{lemma}
\label{lem:technical-contrib}
    Let $k < 0.01 \cdot \log(1/\eps)$.
    Then, for every $a \in \{0,1\}^n$ with $|a^{-1}(1)| > \eps^{0.01}n$, there exists
    $E_a \subseteq k \cdot B$, such that $|E_a|\geq 2^{(h(\Bin(k))-\eps^{0.1})n}$ and $(a,b) \in P$ for every $b \in E_a$.
\end{lemma}

Note, that we can assume that $\alpha > \eps^{0.01}$ because otherwise $|A|
\le \binom{n}{\eps^{0.01}n} \le 2^{\eps^{0.01}\log(4/\eps)n} \le 2^{\delta n}$
and Theorem~\ref{thm:udcp} follows automatically.

Thus, we may apply Lemma~\ref{lem:technical-contrib} for each $a \in A$ and obtain that
\begin{equation}\label{ineq:lowerboundP}
    |P| \geq |A|\cdot 2^{(h(\Bin(k))-\eps^{0.1})n}.
\end{equation}

On the other hand, the balancedness property can be used to give an upper bound on $P$ via Lemma~\ref{lem:UDCP}. To do so, the following will be useful:

\begin{lemma}\label{lem:Pbalanced}
    If $(a,b) \in P$, then $a+b$ is $(2\eps^{0.01})$-$\Bin(k,\alpha)$ balanced.
\end{lemma}
\begin{proof}
    From the definition of $P$, vector $b$ is $\eps^{0.01}$-$\Bin(k)$ balanced
    for $a$. So, for every $i \in \{0,\ldots,k\}$:
    \begin{displaymath}
        |a^{-1}(1) \cap b^{-1}(i)| \in \left[
            \binom{k}{i}\frac{\alpha n}{2^k} \pm \eps^{0.01} n
        \right]
    \end{displaymath}
    And similarly, 
    \begin{displaymath}
        |a^{-1}(0) \cap b^{-1}(i)| \in \left[
            \binom{k}{i}\frac{(1-\alpha) n}{2^k} \pm \eps^{0.01} n
        \right]
    \end{displaymath}
    It follows that for every $i \in \{0,\ldots,k+1\}$ it holds that:
    \begin{displaymath}
        |(a+b)^{-1}(i)| \in \left[
            \binom{k}{i}\frac{(1-\alpha) n}{2^k} +
            \binom{k}{i-1}\frac{\alpha n}{2^k} \pm 2 \eps^{0.01} n
        \right].
    \end{displaymath}
\end{proof}

Next we define function $\eta(a,b) \coloneqq a+b$. Observe that $\eta$ is an injective function on $A \times k\cdot B$ by Lemma~\ref{lem:UDCP}, and since $P \subseteq A \times k\cdot B$ we have $|\eta(P)| = |P|$.
By Lemma~\ref{lem:Pbalanced}, every vector in $\eta(P)$ is $(2\eps^{0.01})$-$\Bin(k,\alpha)$ balanced, and thus Lemma~\ref{lem:boundbalancedvectors} implies
\begin{equation}\label{ineq:upper-bound-P}
    |P| \le  2^{n \cdot h(\Bin(k,\alpha))} \cdot 2^{\Oh(n k \eps \log(1/\eps))}.
\end{equation}

\subsection{Proof of Theorem~\ref{thm:udcp}}
\label{sec:proof-udcp}
By combining~\eqref{ineq:lowerboundP} and~\eqref{ineq:upper-bound-P} we obtain the following bound:
\begin{equation}\label{ineq:Abound}
    |A| \le 2^{n (h(\Bin(k,\alpha))-h(\Bin(k)))} \cdot 2^{\Oh((\eps^{0.01}+k\eps \log \tfrac{1}{\eps})n)}.
\end{equation}

By Lemma~\ref{lem:hatd-hatc} we have that $h(\Bin(k,\alpha))\leq h(\Bin(k+1))$, and thus it remains to bound the difference in entropy of two consecutive binomial distributions as follows:
\begin{lemma}\label{lem:consecbinomentro}
	For large enough $k$, we have that $h(\Bin(k))-h(\Bin(k-1)) \leq \frac{\log k}{\sqrt{k}}$.
\end{lemma}

Before we present the proof of Lemma~\ref{lem:consecbinomentro}, let us see how to use it.
We choose $k \coloneqq \Theta(\log(1/\eps))$. Thus Lemma~\ref{lem:consecbinomentro} implies that
\begin{displaymath}
    |A| \le 2^{n \left( \log{k}/\sqrt{k} + \eps^{0.01}\log(1/\eps)\right)} = 2^{\Oh(n \cdot\delta(\eps))},
\end{displaymath}

where $\delta(\eps)= \Oh_{\eps \rightarrow 0}\left(\frac{\log(\log(1/\eps))}{\sqrt{\log(1/\eps)}}\right)$, 
because $\eps^{0.01} \log(1/\eps) \ll \delta(\eps)$ for small enough $\eps$. This finishes the
proof of Theorem~\ref{thm:udcp}.

\begin{proof}[Proof of Lemma~\ref{lem:consecbinomentro}]
	
	For every $i,k \in \nat$ such that $i \leq k$ let us define an auxiliary function:
	\[
	f(k,i) \coloneqq \frac{\binom{k}{i}}{2^{k}} \log \left( \frac{2^k}{\binom{k}{i}} \right)
	.
	\]
	Thus we have $h(\Bin(k))=\sum_{i=0}^k f(k,i)$.
	To relate $h(\Bin(k))$ with $h(\Bin(k-1))$, the following will be useful:
	\begin{claim}\label{clm:binom}
		\begin{equation}
		f(k,i) \leq
		\begin{cases}
		f(k-1,i), & \text{if } i < \lfloor k/2 \rfloor,   \\
		f(k-1,i-1), & \text{if } i \geq k/2.
		\end{cases}
		\end{equation}
	\end{claim}
	\begin{proof}
		Define $g(x)=x\cdot\log(1/x)$. Since its derivative is $g'(x)=-\frac{\ln(x)+1}{\ln(2)}$ we have that $g(x) \leq g(x')$ whenever $x \leq x'\leq 1/e$. 
		
		Note $f(k,i)=g(\binom{k}{i}/2^{k})$, and since $\binom{k}{i}/2^{k} \leq 1/\sqrt{k}$ by the standard bound $\binom{k}{i} \leq 2^{k}/\sqrt{k}$, we have $\binom{k}{i}/2^{k}\leq 1/e$ for $k \geq 9$. Thus to prove the claim it remains to show that 
		\[
		\binom{k}{i}2^{-k} \leq
		\begin{cases}
		\binom{k-1}{i}2^{-(k-1)}, & \text{if } i < \lfloor k/2 \rfloor,   \\
		\binom{k-1}{i-1}2^{-(k-1)}, & \text{if } i \geq k/2.
		\end{cases}
		\]
		To see this first suppose $i < \lfloor k/2 \rfloor$. Then we have that
		\[
		\binom{k}{i}2^{-k}=\binom{k-1}{i}\frac{k}{k-i-1}2^{-k} \leq \binom{k-1}{i}2^{-(k-1)}.
		\]
		Second, if $i \geq k/2$, then we have that
		\[
		\binom{k}{i}2^{-k} = \binom{k-1}{i-1}\frac{k-1}{i-1}2^{-k} \leq \binom{k-1}{i-1}2^{-(k-1)}.
		\]
	\end{proof}
	Now we can use Claim~\ref{clm:binom} to give the required upper bound:
	\begin{align*}
	h(\Bin(k)) &= \sum_{i=0}^k f(k,i)\\
	&= \left( \sum_{i=0}^{\lfloor k/2\rfloor - 1} f(k,i) \right) + \left(\sum_{i=\lfloor k/2\rfloor+1}^{k} f(k,i)\right) + f(k,\lfloor k/2 \rfloor)\\ 
	&\leq \left( \sum_{i=0}^{\lfloor k/2\rfloor-1} f(k-1,i) \right) + \left(\sum_{i=\lfloor k/2\rfloor+1}^{k} f(k-1,i-1)\right) + f(k,\lfloor k/2 \rfloor)\\
	&= h(\Bin(k-1) + f(k,\lfloor k/2 \rfloor)\\
	&\leq h(\Bin(k-1)) + \log(k)/\sqrt{k},
	\end{align*}
	where we use Claim~\ref{clm:binom} in the first inequality, and $\binom{k}{\lfloor k/2\rfloor} \leq 2^k / \sqrt{k}$ in the second inequality. 
	Hence $h(\Bin(k)) - h(\Bin(k-1)) \le \log(k)/\sqrt{k}$.
\end{proof}

\section{Properties of $k \cdot B$: Proof of Lemma~\ref{lem:technical-contrib}}
\label{sec:technical-section}

\newcommand{\zsum}[0]{\zeta}

In this section, we prove the Lemma~\ref{lem:technical-contrib}.

\begin{lemma}
 Let $k < 0.01 \cdot \log(1/\eps)$.
	Then, for every $a \in \{0,1\}^n$ with $|a^{-1}(1)| > \eps^{0.01}n$, there exists
	$E_a \subseteq k \cdot B$, such that $|E_a|\geq 2^{(h(\Bin(k))-\eps^{0.1})n}$ and $(a,b) \in P$ for every $b \in E_a$.
\   
\end{lemma}
Recall that
\[
	P = \{ (a,b) \in A \times k\cdot B: b \in B \textnormal{ is } \eps^{0.01}\textnormal{-}\Bin(k) \textnormal{ balanced for }a \}.
\]
Intuitively, we prove that for any fixed set $B \subseteq \{0,1\}^n$ there exists a
large set $E_a \subseteq k \cdot B$ with the following property: for every $b \in E_a$ we can perturb $\eps^{0.01}n$ entries in
$b$, such that it is indistinguishable from a vector randomly sampled from
the binomial distribution, even if we focus on a concrete subset of coordinates
$a^{-1}(1) \subseteq \setrng{n}$.


First, observe that we can interpret a tuple $(b_1,\ldots,b_k) \in B^k$ as the $n \times k$ matrix with the $i$'th column equal to $b_i$. Note that we interpret an $x \times y$ matrix as a tuple of $y$ vectors of dimension $x$.
We interchangeably address such a tuple as an $n \times k$ matrix and as an
element of $\{0,1\}^{n \times k}$.
To emphasize the type of such variables, we denote such matrices with boldface.
For example, $(b_1,\ldots,b_k)$ is denoted with $\textbf{b} \in \{0,1\}^{n
    \times k}$.

The notation $\mathbf{b}^T$ denotes the transpose of a matrix. Next, $C \coloneqq \{
\mathbf{b}^T : \mathbf{b} \in B^k \} \subseteq \{0,1\}^{k \times n}$ denotes the set of matrices $B^k$ interpreted in the
transposed way.

In Section~\ref{sec:uniform-k-tuples} we show how to select a subset $D \subseteq C$
of matrices in $C$, in such a way that for all $\mathbf{b} \in D \subseteq
\{0,1\}^{k \times n}$,
any column $z \in \{0,1\}^k$ occurs $\frac{n}{2^k} \pm f(\eps)n$ times in $\mathbf{b}$, that is $|\mathbf{b}^{-1}(z)| \in [\frac{n}{2^k} \pm f(\eps)n]$.

Next, in Section~\ref{sec:many-distinct-sums} we define the operation
$\zsum((a_1,\ldots,a_x))\coloneqq \sum_{i=1}^x a_i$, that sums the columns of a matrix
$\mathbf{a} \in \mathbb{Z}^{y \times x}$ to a single column $\zsum(\mathbf{a}) \in \mathbb{Z}^y$ (see Figure~\ref{fig:zeta}).
We consider the set $E \coloneqq  \{ \zsum(\mathbf{b}^T) : \mathbf{b} \in D \} \subseteq \{0,\ldots,k\}^{n}$, and argue that each vector in $E$ is $\gamma$-$\Bin(k)$ balanced for some small $\gamma >0$.

Finally, in Section~\ref{sec:selecting-E_a} we take care of $a \in \{0,1\}^n$
and select the set $E_a \subseteq E$ to be all vectors in $E$ that are $\eps^{0.01}$-$\Bin(k)$ balanced for $a^{-1}(1)$.

\begin{figure}[t]
	\centering
	\includegraphics[scale=0.75]{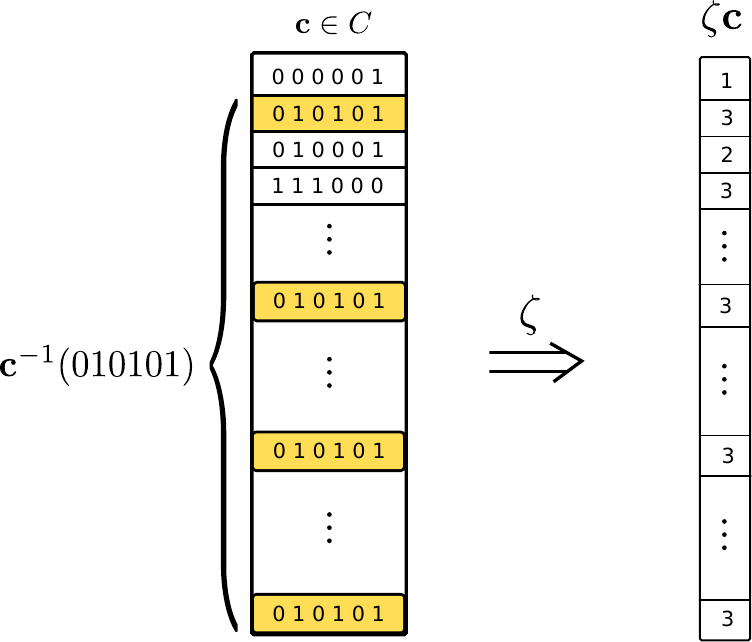}
	\caption{The $\zsum$ operation takes an $x \times y$ matrix $\mathbf{c} \in \mathbb{Z}^{x \times y}$ as input and outputs a vector $\zsum(\mathbf{c}) \in \mathbb{Z}^x$ by adding all columns.}
	\label{fig:zeta}	
\end{figure}

\paragraph{Uniform distribution.} 

We define the \emph{uniform distribution} to be
$\Uni(\Omega)=(\Omega,p)$ if $p(\omega)=\frac{1}{|\Omega|}$ for each
$\omega\in\Omega$. We will focus on the special
cases when $\Omega = \{0,1\}$ and $\Omega = \{0,1\}^k$. Thus, $v
\in \big(\Uni(\{0,1\}) \pm \gamma\big)^{\setrng{n}}$ means that $v^{-1}(i) \in [n/2 \pm
\gamma n]$ for all $i \in \{0,1\}$. Similarly, $\mathbf{v} \in
\big(\Uni(\{0,1\}^k) \pm \gamma\big)^{\setrng{n}}$ means that $\mathbf{v}^{-1}(z) \in
[n/2^k \pm \gamma n]$ for all $z \in \{0,1\}^k$. If $\Omega$ is clear from the context, $v \in \Omega^{U}$ and $X \subseteq U$, we also say that a vector is $\gamma$-uniform for $X$ to refer to the statement that it is $\gamma$-$\Uni(\Omega)$ balanced for $X$.

\paragraph{Inequalities.}

Through the section we assume that $\eps \le 1/2^{4 k}$, $\gamma = 4\sqrt\eps$
and $k > 100$ is an integer. This means that the following inequalities hold:

\begin{equation}
    \label{ineq:2kgamma}
    \eps^{1/2} \le 2^k \gamma \le \eps^{1/4},
\end{equation}

\begin{equation}
    \label{ineq:2kgammalog}
    100 \cdot k 2^k \gamma \log(1/(2^k\gamma)) \le \eps^{1/5}.
\end{equation}

\subsection{Constructing a set $D$ of uniform $k$-tuples}
\label{sec:uniform-k-tuples}
We first prove the following result that will be helpful to obtain the aforementioned set $C$.

\begin{lemma}[Most vectors in $B$ are uniform]\label{cor:balancedLBr}
    Let $U_1 \uplus \ldots \uplus U_\ell = \setrng{n}$ be a partition such that
	$|U_i| \ge \mu n$ for all $i\in[\ell]$. Let $\lambda \in
	(0,1/2]$.
	For every $B \subseteq \{0,1\}^n$ with $|B| \geq
	2^{(1 - \mu \lambda^2)n - o(n)}$ it holds that:
	\begin{displaymath}
	\left| \left\{
	b \in B \; : \; b \text{ is $\lambda$-uniform in }U_i \text{ for every $i$} 
	\right\}\right|  \geq |B|/2.
	\end{displaymath}
\end{lemma}
\begin{proof}
	For a fixed $i$ we argue that the number of vectors that are not
	$\lambda$-uniform for $U_i$ is bounded by $2^{(1-\lambda^2 \mu)n + o(n)}$. This
	will finish the proof since we can sum this bound over all partitions.
	
	Let $s = |U_i|/n$ and note that $s \ge \mu$.
	Observe that the number of vectors $v \in \{0,1\}^n$ such that $|v^{-1}(1) \cap
	U_i| \notin [sn/2 \pm \lambda s n]$ is at most:
	\begin{displaymath}
	\sum_{\lambda' \notin [-\lambda,\lambda]} \binom{s n}{s n/2 - \lambda' s n}
	2^{n - sn}
	\end{displaymath}
	because a vector $v$ that is not $\lambda$-uniform on $U_i$ can
    be arbitrary in $\setrng{n}\setminus U_i$. We upper bound this with the binary
	entropy function
	
	\begin{displaymath}
	\sum_{0 \le \lambda' \le \lambda} 2^{s n \cdot  h(1/2 - \lambda') +o(n)}
	\cdot 2^{n - sn}
	.
	\end{displaymath}
	The expression is maximized when $\lambda' = \lambda$ because the $h(p)$
	entropy function is increasing in $[0,1/2]$. Hence we can upper-bound the expression with
	\begin{displaymath}
	n 2^{(1+ s (h(1/2 - \lambda) - 1))n + o(n)}
	.
	\end{displaymath}
	Now, we use a bound $h(1/2
	-x) \le 1 - x^2$ when $0 \le x \le 1/2$
	(recall that $\lambda \in [0,1/2]$) and obtain that the
	number of vectors that are not $\lambda$-uniform for $U_i$ is at most
	\begin{displaymath}
	2^{\left(1 - s \lambda^2 \right)n +
		o(n)} \le
	2^{\left(1 - \mu\lambda^2 \right) n + o(n)}
    .
	\end{displaymath}
	
	Thus, by summing over all $U_i$, the number of vectors that are not $\lambda$-uniform for some $U_i$ is at most $2^{\left(1 - \mu\lambda^2 \right) n + o(n)} $, and the number of vectors in $B$ that are $\lambda$-uniform for all $U_i$ is at least
	\[
	|B| - 2^{\left(1 - \mu\lambda^2 \right) n + o(n)}  \geq |B|/2,
	\]
	and the claim follows.
\end{proof}
Set a balance parameter $\gamma= 4\sqrt{\eps}$, and define
\begin{displaymath}
    D \coloneqq  C \cap \big(\Uni(\{0,1\}^k) \pm \gamma\big)^{\setrng{n}}.
\end{displaymath}
\begin{lemma}[Most $k$-tuples are uniform]\label{balancedBkr}
	Let $k \in \nat$ be such that $\eps < 1/4^{k+2}$. Then it holds that
	\begin{displaymath}
	\left| D \right| \ge \left(\frac{|B|}{2}\right)^k.
	\end{displaymath}
\end{lemma}
\begin{proof}
    We will use set of matrices $C_j \subseteq (\{0,1\}^{j \times n}$ to select all matrices obtained by keeping the first $j$ columns of matrices  of $C$, namely
	\[
		C_j \coloneqq \{ \mathbf{b}^T : \mathbf{b} \in B^j \}.
	\]
	Thus, $C = C_k$.
	For $j \in \{1,\ldots,k\}$, let
	\begin{displaymath}
        D_j  \coloneqq C_j \cap \big(\Uni(\{0,1\}^j) \pm\gamma\big)^{\setrng{n}}
	.
	\end{displaymath}
	We prove that $\left| D_j \right| \ge (\frac{|B|}{2})^j$ by induction on $k$.
	First, we prove the base case $j=1$ of the induction, so $|D_1| \ge |B|/2$. This follows by applying Lemma~\ref{cor:balancedLBr} with $\lambda =
    \sqrt{\eps}$ and partition $U_1 = \setrng{n}$, since it implies that
	\begin{displaymath}
	 |B|/2 \leq \left| B \cap \big(\Uni(\{0,1\}) \pm \sqrt{\eps}\big)^{U_1} \right|  \leq  \left| C_1 \cap \big(\Uni(\{0,1\}) \pm \gamma\big)^{U_1} \right| = \left| D_1\right|.
	\end{displaymath}
	The induction step with $j > 1$ is a direct consequence of the following claim, which therefore is sufficient to finish the proof.
	
	\begin{claim}
	 Let $ \mathbf{b} \in D_{j-1}$.
	 Then there are at least $|B|/2$ vectors $b_j \in B$, such that
     $\mathbf{b_+} \in \big(\Uni(\{0,1\}^j) \pm \gamma\big)^{\setrng{n}}$, where $\mathbf{b_+}$ is obtained from $\mathbf{b}$ by appending $b_j$ as the $j$'th row to it.
	\end{claim}
	 \begin{proof}
         Define a partition $\{U_z\}_{z \in \{0,1\}^{j-1}}$ of $\setrng{n}$ by $U_z={\mathbf{b}}^{-1}(z)$.
	Because $\mathbf{b} \in \big(\Uni(\{0,1\}^{(j-1)}) \pm
    \gamma\big)^{\setrng{n}}$ we know that:
\begin{displaymath}
\mu = \min_{z\in\{0,1\}^{j-1}} |U_z|/n \ge \frac{1}{2^{j-1}} - \gamma .
\end{displaymath}

Note, that $\mu > 1/2^j$ because we assumed that $\eps < 1/4^{k+2}$
(hence $\gamma < 1/2^j$).
Now, we use Lemma~\ref{cor:balancedLBr} with partition
$\{U_z\}_{z \in \{0,1\}^{j-1}}$ and $\lambda = 2^{j-3}\cdot \gamma$.
First let us assert that the condition
$|B| \ge 2^{(1-\mu \lambda^2)n + o(n)}$ holds. Recall that we
assumed $|B| \ge 2^{(1-\eps)n + o(n)}$ and $\mu \lambda^2 \ge
\frac{1}{2^j} (2^{j-3} \cdot 4\sqrt{\eps})^2 \ge 2^{j-2} \eps \ge \eps$ (for $j
\ge 2$). Hence $|B| \ge 2^{(1-\eps) n} \ge 2^{(1-\mu \lambda^2)n +
	f(n)}$ for some function $f \in o(n)$.

Lemma~\ref{cor:balancedLBr} states that 
there are at least $|B|/2$ vectors $b_j \in B$ such that for each $z \in \{0,1\}^{j-1}$
\begin{equation}
\label{eq:Uz1_intvlr}
|U_z\cap b_j^{-1}(1)| \in
\left[\frac{|U_z|}{2} \pm  \lambda |U_z|\right]
.
\end{equation}
We know that $|U_z| \in [\frac{n}{2^{j-1}} \pm \gamma n]$ (because $\mathbf{b} \in (\Uni(\{0,1\}^{j-1}) \pm
\gamma)^{\setrng{n}}$). Thus in fact \eqref{eq:Uz1_intvlr} can be rewritten to
\begin{displaymath}
|U_z\cap b_j^{-1}(1)| \in
\left[\frac{n}{2^j} \pm  \bigl(\lambda |U_z|+(\gamma / 2) n\bigr)\right]
.
\end{displaymath}
We bound $\lambda|U_z|$ by
\begin{multline*}
\lambda|U_z|
=2^{j-3}\gamma|U_z|
\le2^{j-3}\gamma\left(\frac{n}{2^{j-1}}+\gamma n\right)
=\gamma n\left(\frac14+2^{j-3}\gamma\right)
=\gamma n\left(\frac14+2^{j-3}4\sqrt{\eps}\right)\\
<\gamma n\left(\frac14+2^{j-3}4\sqrt{1/4^{k+2}}\right)
=\gamma n\left(\frac14+2^{j-1}/2^{k+2}\right)
<\gamma n\left(\frac14+\frac14\right)=(\gamma/2) n,
\end{multline*}
where we use the assumption $\eps \leq \tfrac{1}{4^{k+2}}$ in the second line of the inequality.
Thus, for every $z\in\{0,1\}^{j-1}$ we have
\begin{displaymath}
|U_z\cap b_j^{-1}(1)| \in
\left[\frac{n}{2^j} \pm  2(\gamma/2) n\right]
.
\end{displaymath}
Now, observe that for all $z\in\{0,1\}^{j-1}$ it holds
that:

\begin{displaymath}
U_z\cap b_j^{-1}(1)
=\mathbf{b}^{-1}(z)\cap b_j^{-1}(1)=\mathbf{b}_\mathbf{+}^{-1}\bigl(z'),
\end{displaymath}
where $z' \in \{0,1\}^{j}$ is the vector obtained from $z$ by adding a $j$-th entry with value $1$.
Thus vector $z'$ fulfills  the condition for $\mathbf{b_+}$ to be
in $\big(\Uni(\{0,1\}^j) \pm \gamma\big)^{\setrng{n}}$.
Similarly, we can prove this condition by concatenating a $0$ to the vector $z$.
Hence, for every $\mathbf{b}\in D_{j-1}$ there are
at least $|B|/2$ vectors $b_j \in B$, such that $\mathbf{b_+} \in
\big(\Uni(\{0,1\}^j \pm \gamma\big)^{\setrng{n}}$. 
	 \end{proof}
     Thus this claim proves our induction hypothesis and hence the lemma.
\end{proof}

\subsection{Summing tuples from $D$ gives many distinct sums}
\label{sec:many-distinct-sums}

As mentioned at the beginning of this section we define the operation
$\zsum((a_1,\ldots,a_x))\coloneqq \sum_{i=1}^x a_i$, that sums the columns of a matrix
$\mathbf{a} \in \mathbb{Z}^{y \times x}$ to a single column $\zsum(\mathbf{a}) \in \mathbb{Z}^y$ (see Figure~\ref{fig:zeta}).

We define $E$ to be all sums of tuples from $D$:
\begin{displaymath}
E \coloneqq  \{ \zsum(\mathbf{b}^T) : \mathbf{b} \in D \} \subseteq \{0,\ldots,k\}^{n}
\end{displaymath}
In fact, by the assumption on $D$ we have the following control on the distributions of the values in the vectors in $E$:

\begin{lemma}
	\label{obs:bnml-dstr}
	 If $v \in E$, then for $j \in \{0,\ldots,k\}$ it holds that $|v^{-1}(j)| \in \left[\binom{k}{j}\frac{n}{2^k}\pm\binom{k}{j}\gamma n\right]$, i.e., every vector in $E$ is a $(2^k \gamma)$-$\Bin(k)$ balanced vector.
\end{lemma}

\begin{proof}
	Consider an arbitrarily vector $v \in E$ and fix $j \in
	\{0,\ldots k\}$. From the definition of $E$, there exists a vector
    $\mathbf{b} \in D =  C \cap \big(\Uni(\{0,1\}^k) \pm \gamma\big)^{\setrng{n}}$ such that
	$\zsum(\mathbf{b}^T) = v$. Hence for every $z \in \{0,1\}^k$:
	\begin{displaymath}
	|\mathbf{b}^{-1}(z)|
	\in 
	\left[ \frac{n}{2^k} \pm \gamma n\right]
	.
	\end{displaymath}
	Hence, if we sum over all vectors $z\in \{0,1\}^k$ such that $|z^{-1}(1)|=j$ we have:
	\begin{displaymath}
	|v^{-1}(j)| \le \sum_{\substack{z \in \{0,1\}^k \\ |z^{-1}(1)|=j}} \frac{n}{2^k} + \gamma n
	= \binom{k}{j} \frac{n}{2^k} + \binom{k}{j} \gamma n,
	\end{displaymath}
	and analogously $|v^{-1}(j)| \ge \binom{k}{j} \frac{n}{2^k} - \binom{k}{j} \gamma n$. Thus indeed $v$ is a $(2^k\gamma)$-$\Bin(k)$ balanced vector, as desired.
\end{proof}
We now show that $E$ is sufficiently large:
\begin{lemma}\label{lem:sizekBr}
    It holds that $|E| \geq 2^{(h(\Bin(k)-\eps^{0.2})n}$.
\end{lemma}
\begin{proof}
	For a vector $v \in E$ we define
	\[
		D_v \coloneqq \{ \mathbf{b} \in D : \zsum(\mathbf{b}^T)=v\}.
	\]
	By grouping all elements of $D$ on their image with respect to $\zsum$:	
	\[
		|D| = \sum_{v \in E} |D_v| \leq |E| \max_{v \in E} |D_v|,
	\] 
	which can be rewritten into the following lower bound on $|E|$:
	\begin{equation}\label{eq:boundE}
		|E| \geq |D| / \max_{v \in E} |D_v| \geq (|B|/2)^k / \max_{v \in E} |D_v| \geq  2^{k(n-\eps n+1)} / \max_{v \in E} |D_v|.
	\end{equation}
	Thus in the remainder of the proof we can focus on showing that for any
    vector $v \in E$, $|D_v| \leq 2^{n(k - h(\Bin(k))+\eps^{0.2})}$; the Lemma
    would then follow by the bound in~\eqref{eq:boundE}.
		
	Let $\mathbf{b} \in D_v$. This means that for every $j \in \{0,\ldots,k\}$:
	\[
		\bigcup_{\substack{z \in \{0,1\}^k \\ |z^{-1}(1)|=j}} \mathbf{b}^{-1}(z)=v^{-1}(j).
	\]
	Thus the number of possibilities for $\mathbf{b}$ is
	\[
		\prod_{j=0}^k \binom{k}{j}^{|v^{-1}(j)|}.
	\]
    We multiply this quantity with $\binom{\setrng{n}}{|v^{-1}(0)|, \ldots,
    |v^{-1}(k)|}$  and obtain
	\[
        \binom{\setrng{n}}{|v^{-1}(0)|, \cdots, |v^{-1}(k)|} \cdot \prod_{j=0}^k \binom{k}{j}^{|v^{-1}(j)|} \leq 2^{kn},
	\]
	where the inequality follows since the left-hand side counts partitions of
    $\setrng{n}$ into $\sum_{i=0}^k \binom{k}{i}=2^k$ parts.  Now, let $\bold{r} = (|v^{-1}(0)|/n,\ldots,|v^{-1}(k)|/n)$
    and observe that by Lemma~\ref{lem:multvsentropy} we have $|D_v| \leq 2^{n(k-h(\bold{r}))}$.
	Because $v$ is $\gamma$-$\Bin(k)$ balanced (since it is in $E$), we have
	\[
        h(\bold{r}) \geq h(\Bin(k))- \frac{1}{\ln(2)}\cdot k2^k\gamma \log \frac{1}{2^k\gamma} \geq
        h(\Bin(k)) - \eps^{0.2},
	\]
	where the first inequality is by Lemma~\ref{lem:entclose}, and the second
    inequality uses that $\gamma=4\sqrt{\eps}$ (see Inequality~\ref{ineq:2kgammalog}).
\end{proof}
\subsection{Selecting a set $E_a \subseteq E$ for every $a \in A$}
\label{sec:selecting-E_a}
\begin{lemma}\label{lem:manybalancers}
	Let $\mathcal{D}=(\Omega,p)$ be a discrete probability space, and let $X
    \subseteq \setrng{n}$ with $|X|=\alpha n$.
    The number of vectors $v \in (\mathcal{D}+0)^{\setrng{n}}$ that are not
    $\rho$-$\mathcal{D}$ balanced for $X$ and $\setrng{n}\setminus X$ is at most $2^{n(h(\mathcal{D})-\alpha^2\min(\rho^2,\log \alpha))}$.
\end{lemma}
\begin{proof}
    We define a relation $R\subseteq (\mathcal{D}+0)^{\setrng{n}}
    \times \binom{\setrng{n}}{\alpha n}$ as follows:
	\begin{displaymath}
	(v, X)\in R\Leftrightarrow
	v \text{ is \textbf{not} $\rho$-$\mathcal{D}$ balanced for } X.
	\end{displaymath}
	Additionally, let
	\begin{align*}
        R_v &= R\cap \left(\{v\}\times \binom{\setrng{n}}{\alpha n}\right),
        &&\text{ for }v\in (\mathcal{D}+0)^{\setrng{n}},\\
        R_X &= R\cap \left((\mathcal{D}+0)^{\setrng{n}}\times\{X\}\right),
        &&\text{ for }X\in \binom{\setrng{n}}{\alpha n}.
	\end{align*}
	Note that $|R_X|$ is the value we want to bound.
	Note that the mapping $(v,X)\mapsto(v\circ\pi,\pi(X))$ for any permutation
    $\pi: \setrng{n} \leftrightarrow \setrng{n}$ of the index set $\setrng{n}$ is an automorphism of $R$ (i.e., $(v,X) \in R$ if and only if $(\pi(v),\pi(X)) \in R$). Therefore, we have
	\begin{equation}
	\label{obs:equation-mainr}
    |R|=|(\mathcal{D}+0)^{\setrng{n}}|\cdot|R_v|=
	|R_X|\binom{n}{\alpha n}
	\end{equation}
	for a fixed $v$ and $X$. 
	By~\eqref{obs:equation-mainr} we can focus on bounding $|R_v|$ instead of $|R_X|$.
	To do so, note that if $(v,X)\in R$
	for $X\in\binom{\setrng{n}}{\alpha n}$,
	there must exist $\omega \in \Omega$
	such that $|X\cap v^{-1}(\omega)|\notin[p(\omega) \alpha n \pm \rho \alpha n]$.
	We can construct any such $X$ by first selecting a subset of $v^{-1}(\omega)$
	(which has cardinality $p(\omega)n$), and then choosing the remaining elements.
	Hence:
	
	\begin{align*}
	|R_v|&\le
	\sum_{\omega \in \Omega} \sum_{x \notin [-\rho,\rho]}
	\binom{p(\omega)n}{(p(\omega)+ x)\alpha n}
	\binom{(1 - p(\omega))n}{(1-p(\omega)- x)\alpha n}.
	\end{align*}
	Next, we use Lemma~\ref{binom:inequality}. 
	In our case (with parameters
	$\beta = p(\omega)$,
	$\alpha = \alpha$ and
	$\gamma = -\alpha x$) it implies:
	
	\begin{displaymath}
	|R_v| \le \sum_{j \in \{0,\ldots,k\}} \sum_{x \notin [-\rho,\rho]}
	\binom{n}{\alpha n} 2^{-c \cdot n}
	\end{displaymath}
	where
	\begin{displaymath}
	c=
	\begin{cases}
	(\alpha x)^2, & \text{if } |\alpha x| < \alpha(1-\alpha)p(\omega),\\
	\alpha^2\log(1/(2\alpha)), & \text{otherwise}.
	\end{cases}
	\end{displaymath}
	Since $|x|\ge\rho$ we have
	$c\ge\alpha^2\min\{\rho^2, \log(1/(2\alpha))\}$,
	thus
	\begin{displaymath}
	|R_v| \le \binom{n}{\alpha n} 2^{-\alpha^2\min\{\rho^2, \log(1/(2\alpha)) \} n},
	\end{displaymath}
	which plugged into \eqref{obs:equation-mainr} gives the desired inequality.
\end{proof}

\begin{proof}[Proof of Lemma~\ref{lem:technical-contrib}]
    Recall that we assume that $\alpha > \eps^{0.01}$.
	Lemma~\ref{lem:sizekBr} gives us a lower bound on $E$, and by
    Lemma~\ref{obs:bnml-dstr} each vector in $E$ is a $(2^k\gamma)$-$\Bin(k)$
    balanced vector. Hence,
	by the pigeonhole principle there exists a distribution
    $\mathcal{D}=(\{0,1\ldots,k\},p)$ where $p=(p_0,\ldots,p_k)$ such that
    $|p_0-\binom{k}{j}2^{-k}|\leq 2^k\gamma$ for each $j$ and $E$ has a subset
    $E'$ of at least $|E|/n^{k}$ vectors that are in $\mathcal{D}^{\setrng{n}}$. Hence:

    \begin{displaymath}
        |E'| \ge |E|/n^k \ge 2^{(h(\mathcal{D}) - \eps^{0.2}n) - o(n)}
    \end{displaymath}

	Now for each $a \in A$, define $E_a$ to be all vectors in $E'$ that are
    $\eps^{0.05}$-$\mathcal{D}$ balanced for $a^{-1}(1)$. Observe that this means
    that vectors in $E_a$ are $\eps^{0.01}$-$\Bin(k)$ balanced (because
    $\eps^{0.05}+2^k\gamma \ll \eps^{0.01}$).
    
    Applying
    Lemma~\ref{lem:manybalancers} with $E'$ and $a^{-1}(1)$, we get that there
    are at most $2^{h(\mathcal{D})-\alpha^2\min(\eps^{0.1},\log(1/(2\alpha))}
\le 2^{(h(\mathcal{D}) - \eps^{0.12}) n)}$ vectors
    in $E'$ that are not $\eps^{0.05}$-$\mathcal{D}$ balanced. Hence:

    \begin{displaymath}
        |E' \setminus E_a| \le 2^{(h(\mathcal{D}) - \eps^{0.12}n)} \le |E'|/2
    \end{displaymath}

	Now the lemma follows because
    \begin{displaymath}
        |E_a| \ge |E'|/2 \ge 2^{(h(\mathcal{D}) - \eps^{0.2})n - o(n)} \ge
        2^{(h(\Bin(k) - \frac{1}{\ln(2)}\cdot 2^k\gamma\log(1/(2^k\gamma)) - \eps^{0.2})n} \ge 2^{(h(\Bin(k) - \eps^{0.1})n}
    \end{displaymath}
    
    where the last inequality follows from Lemma~\ref{lem:entclose} and
    Inequality~\ref{ineq:2kgammalog} (since $\eps$ is small enough).

\end{proof}

\section{Conclusion and Open Problems}

In this paper, we present a randomized $\Oh(2^{(1-\sigma_m)n})$ time algorithm
for the Bin Packing problem, where $\sigma_m >0$ and $m$ denotes the number of
bins. This is an improvement over the state-of-the-art algorithm of Bj\"orklund
et al.~\cite{BjorklundHK09} that runs in $\Os(2^n)$ time for small $m$.
Nevertheless, it remains to give an algorithm for Bin Packing that works
in $\Os((2-\eps)^n)$ time for an unbounded number of bins for some fixed
constant $\eps > 0$.
We believe our algorithm made significant progress on this question. One open
end for further research is how the number of bins influences the complexity of an instance. By the methods of~\cite{Nederlof16}, instances of Bin Packing with
a linear number of bins (with equal capacity) can also be solved in time
$\Oh(2^{(1-\eps)n})$ based on a witness sampling technique similar to what we
used in some of our cases.  It is thus natural to wonder whether (an extension)
of the methods presented in this paper are enough to give improved algorithms
for all numbers of bins.

We believe our Additive Combinatorics result is natural and may have
applications beyond the scope of this paper. As mentioned in the introduction,
Littlewood--Offord's theory has a wide variety of applications, and it is natural
to expect that the setting that we address may be of interest in any of these
settings.

In the introduction we mentioned Question~\ref{q1} as one motivation for
studying improved exact exponential time algorithms for the Bin Packing problem.
While it is not clear if we made direct progress on this question, we do believe
that some of our ideas such as the approach to narrow down the number of
witnesses may inspire future work on improved algorithms for Set Cover.

\subsection*{Acknowledgement}

The research leading to the results presented in this paper was partially
carried out during the Parameterized Algorithms Retreat of the University of
Warsaw, PARUW 2020, held in Krynica-Zdr\'{o}j in February 2020. This workshop was
supported by a project that has received funding from the European Research
Council (ERC) under the European Union's Horizon 2020 research and innovation
programme under grant agreement No 714704 (PI: Marcin Pilipczuk).

\bibliographystyle{abbrv}
\bibliography{bib}

\begin{thebibliography}{10}

\bibitem{Abboud19}
A.~Abboud.
\newblock Fine-grained reductions and quantum speedups for dynamic programming.
\newblock In C.~Baier, I.~Chatzigiannakis, P.~Flocchini, and S.~Leonardi,
  editors, {\em 46th International Colloquium on Automata, Languages, and
  Programming, {ICALP} 2019, July 9-12, 2019, Patras, Greece}, volume 132 of
  {\em LIPIcs}, pages 8:1--8:13. Schloss Dagstuhl - Leibniz-Zentrum f{\"{u}}r
  Informatik, 2019.

\bibitem{stacs2015}
P.~Austrin, P.~Kaski, M.~Koivisto, and J.~Nederlof.
\newblock {Subset Sum in the Absence of Concentration}.
\newblock In E.~W. Mayr and N.~Ollinger, editors, {\em 32nd International
  Symposium on Theoretical Aspects of Computer Science, {STACS} 2015, March
  4-7, 2015, Garching, Germany}, volume~30 of {\em LIPIcs}, pages 48--61.
  Schloss Dagstuhl - Leibniz-Zentrum fuer Informatik, 2015.

\bibitem{stacs2016}
P.~Austrin, P.~Kaski, M.~Koivisto, and J.~Nederlof.
\newblock Dense subset sum may be the hardest.
\newblock In N.~Ollinger and H.~Vollmer, editors, {\em 33rd Symposium on
  Theoretical Aspects of Computer Science, {STACS} 2016, February 17-20, 2016,
  Orl{\'{e}}ans, France}, volume~47 of {\em LIPIcs}, pages 13:1--13:14. Schloss
  Dagstuhl - Leibniz-Zentrum fuer Informatik, 2016.

\bibitem{AustrinKKN18}
P.~Austrin, P.~Kaski, M.~Koivisto, and J.~Nederlof.
\newblock Sharper upper bounds for unbalanced uniquely decodable code pairs.
\newblock {\em {IEEE} Trans. Inf. Theory}, 64(2):1368--1373, 2018.

\bibitem{BansalGN018}
N.~Bansal, S.~Garg, J.~Nederlof, and N.~Vyas.
\newblock {Faster Space-Efficient Algorithms for Subset Sum, k-Sum, and Related
  Problems}.
\newblock {\em {SIAM} J. Comput.}, 47(5):1755--1777, 2018.

\bibitem{generic-knapsack2}
A.~Becker, J.~Coron, and A.~Joux.
\newblock {Improved Generic Algorithms for Hard Knapsacks}.
\newblock In K.~G. Paterson, editor, {\em Advances in Cryptology - {EUROCRYPT}
  2011 - 30th Annual International Conference on the Theory and Applications of
  Cryptographic Techniques, Tallinn, Estonia, May 15-19, 2011. Proceedings},
  volume 6632 of {\em Lecture Notes in Computer Science}, pages 364--385.
  Springer, 2011.

\bibitem{BjorklundHKK07}
A.~Bj{\"{o}}rklund, T.~Husfeldt, P.~Kaski, and M.~Koivisto.
\newblock Fourier meets {M}{\"{o}}bius: fast subset convolution.
\newblock In D.~S. Johnson and U.~Feige, editors, {\em Proceedings of the 39th
  Annual {ACM} Symposium on Theory of Computing, San Diego, California, USA,
  June 11-13, 2007}, pages 67--74. {ACM}, 2007.

\bibitem{BjorklundHKK09}
A.~Bj{\"{o}}rklund, T.~Husfeldt, P.~Kaski, and M.~Koivisto.
\newblock {Counting Paths and Packings in Halves}.
\newblock In A.~Fiat and P.~Sanders, editors, {\em Algorithms - {ESA} 2009,
  17th Annual European Symposium, Copenhagen, Denmark, September 7-9, 2009.
  Proceedings}, volume 5757 of {\em Lecture Notes in Computer Science}, pages
  578--586. Springer, 2009.

\bibitem{BjorklundHK09}
A.~Bj{\"{o}}rklund, T.~Husfeldt, and M.~Koivisto.
\newblock {Set Partitioning via Inclusion-Exclusion}.
\newblock {\em {SIAM} J. Comput.}, 39(2):546--563, 2009.

\bibitem{BjorklundKK17}
A.~Bj{\"{o}}rklund, P.~Kaski, and I.~Koutis.
\newblock {Directed Hamiltonicity and Out-Branchings via Generalized
  Laplacians}.
\newblock In I.~Chatzigiannakis, P.~Indyk, F.~Kuhn, and A.~Muscholl, editors,
  {\em 44th International Colloquium on Automata, Languages, and Programming,
  {ICALP} 2017, July 10-14, 2017, Warsaw, Poland}, volume~80 of {\em LIPIcs},
  pages 91:1--91:14. Schloss Dagstuhl - Leibniz-Zentrum f{\"{u}}r Informatik,
  2017.

\bibitem{calabro2009exponential}
C.~Calabro.
\newblock {\em The exponential complexity of satisfiability problems}.
\newblock PhD thesis, UC San Diego, 2009.

\bibitem{CoffmanJr.2013}
E.~G. Coffman~Jr., J.~Csirik, G.~Galambos, S.~Martello, and D.~Vigo.
\newblock {\em Bin Packing Approximation Algorithms: Survey and
  Classification}, pages 455--531.
\newblock Springer New York, New York, NY, 2013.

\bibitem{csiszar2004information}
I.~Csisz{\'a}r and P.~C. Shields.
\newblock {\em Information theory and statistics: A tutorial}.
\newblock Now Publishers Inc, 2004.

\bibitem{subset-sum-lower}
M.~Cygan, H.~Dell, D.~Lokshtanov, D.~Marx, J.~Nederlof, Y.~Okamoto, R.~Paturi,
  S.~Saurabh, and M.~Wahlstr{\"{o}}m.
\newblock On problems as hard as {CNF-SAT}.
\newblock In {\em Proceedings of the 27th Conference on Computational
  Complexity, {CCC} 2012, Porto, Portugal, June 26-29, 2012}, pages 74--84.
  {IEEE} Computer Society, 2012.

\bibitem{CyganDLMNOPSW16}
M.~Cygan, H.~Dell, D.~Lokshtanov, D.~Marx, J.~Nederlof, Y.~Okamoto, R.~Paturi,
  S.~Saurabh, and M.~Wahlstr{\"{o}}m.
\newblock On problems as hard as {CNF-SAT}.
\newblock {\em {ACM} Trans. Algorithms}, 12(3):41:1--41:24, 2016.

\bibitem{CyganFKLMPPS15}
M.~Cygan, F.~V. Fomin, L.~Kowalik, D.~Lokshtanov, D.~Marx, M.~Pilipczuk,
  M.~Pilipczuk, and S.~Saurabh.
\newblock {\em Parameterized Algorithms}.
\newblock Springer, 2015.

\bibitem{delorme2016bin}
M.~Delorme, M.~Iori, and S.~Martello.
\newblock Bin packing and cutting stock problems: Mathematical models and exact
  algorithms.
\newblock {\em European Journal of Operational Research}, 255(1):1--20, 2016.

\bibitem{DiakonikolasS13}
I.~Diakonikolas and R.~A. Servedio.
\newblock Improved approximation of linear threshold functions.
\newblock {\em Comput. Complex.}, 22(3):623--677, 2013.

\bibitem{eilon1971loading}
S.~Eilon and N.~Christofides.
\newblock The loading problem.
\newblock {\em Management Science}, 17(5):259--268, 1971.

\bibitem{eisemann57}
K.~Eisemann.
\newblock The trim problem.
\newblock {\em Management Science}, 3(3):279--284, 1957.

\bibitem{etscheid2017polynomial}
M.~Etscheid, S.~Kratsch, M.~Mnich, and H.~R{\"{o}}glin.
\newblock Polynomial kernels for weighted problems.
\newblock {\em J. Comput. Syst. Sci.}, 84:1--10, 2017.

\bibitem{FominK13}
F.~V. Fomin and P.~Kaski.
\newblock Exact exponential algorithms.
\newblock {\em Commun. {ACM}}, 56(3):80--88, 2013.

\bibitem{FominK10}
F.~V. Fomin and D.~Kratsch.
\newblock {\em Exact Exponential Algorithms}.
\newblock Texts in Theoretical Computer Science. An {EATCS} Series. Springer,
  2010.

\bibitem{DBLP:journals/combinatorica/FrankT87}
A.~Frank and {\'{E}}.~Tardos.
\newblock {An application of simultaneous Diophantine approximation in
  combinatorial optimization}.
\newblock {\em Combinatorica}, 7(1):49--65, 1987.

\bibitem{GoemansR14}
M.~X. Goemans and T.~Rothvo{\ss}.
\newblock Polynomiality for bin packing with a constant number of item types.
\newblock In C.~Chekuri, editor, {\em Proceedings of the Twenty-Fifth Annual
  {ACM-SIAM} Symposium on Discrete Algorithms, {SODA} 2014, Portland, Oregon,
  USA, January 5-7, 2014}, pages 830--839. {SIAM}, 2014.

\bibitem{GolovnevKM14}
A.~Golovnev, A.~S. Kulikov, and I.~Mihajlin.
\newblock Families with infants: {A} general approach to solve hard partition
  problems.
\newblock In J.~Esparza, P.~Fraigniaud, T.~Husfeldt, and E.~Koutsoupias,
  editors, {\em Automata, Languages, and Programming - 41st International
  Colloquium, {ICALP} 2014, Copenhagen, Denmark, July 8-11, 2014, Proceedings,
  Part {I}}, volume 8572 of {\em Lecture Notes in Computer Science}, pages
  551--562. Springer, 2014.

\bibitem{GolovnevKM16}
A.~Golovnev, A.~S. Kulikov, and I.~Mihajlin.
\newblock {Families with Infants: Speeding Up Algorithms for NP-Hard Problems
  Using {FFT}}.
\newblock {\em {ACM} Trans. Algorithms}, 12(3):35:1--35:17, 2016.

\bibitem{Griggs1998DatabaseSA}
J.~R. Griggs.
\newblock Database security and the distribution of subset sums in
  $\mathbb{R}^m$.
\newblock In {\em Graph Theory and Combinatorial Biology}, 1998.

\bibitem{halasz1977estimates}
G.~Hal{\'a}sz.
\newblock Estimates for the concentration function of combinatorial number
  theory and probability.
\newblock {\em Periodica Mathematica Hungarica}, 8(3-4):197--211, 1977.

\bibitem{hoberg17}
R.~Hoberg and T.~Rothvoss.
\newblock A logarithmic additive integrality gap for bin packing.
\newblock In P.~N. Klein, editor, {\em Proceedings of the Twenty-Eighth Annual
  {ACM-SIAM} Symposium on Discrete Algorithms, {SODA} 2017, Barcelona, Spain,
  Hotel Porta Fira, January 16-19}, pages 2616--2625. {SIAM}, 2017.

\bibitem{HorowitzS74}
E.~Horowitz and S.~Sahni.
\newblock {Computing Partitions with Applications to the Knapsack Problem}.
\newblock {\em J. {ACM}}, 21(2):277--292, 1974.

\bibitem{jain2021anticoncentration}
V.~Jain, A.~Sah, and M.~Sawhney.
\newblock Anticoncentration versus the number of subset sums.
\newblock {\em Advances in Combinatorics}, page 24872, 2021.

\bibitem{JansenKMS13}
K.~Jansen, S.~Kratsch, D.~Marx, and I.~Schlotter.
\newblock Bin packing with fixed number of bins revisited.
\newblock {\em J. Comput. Syst. Sci.}, 79(1):39--49, 2013.

\bibitem{first-fit}
D.~S. Johnson.
\newblock {\em Near-optimal bin packing algorithms}.
\newblock PhD thesis, Massachusetts Institute of Technology, 1973.

\bibitem{KaneW16}
D.~M. Kane and R.~Williams.
\newblock Super-linear gate and super-quadratic wire lower bounds for depth-two
  and depth-three threshold circuits.
\newblock In D.~Wichs and Y.~Mansour, editors, {\em Proceedings of the 48th
  Annual {ACM} {SIGACT} Symposium on Theory of Computing, {STOC} 2016,
  Cambridge, MA, USA, June 18-21, 2016}, pages 633--643. {ACM}, 2016.

\bibitem{kantorovich}
L.~V. Kantorovich.
\newblock Mathematical methods of organizing and planning production.
\newblock {\em Management science, English Translation of a 1939 paper written
  in Russian}, 6(4):366--422, 1960.

\bibitem{karmarkar}
N.~Karmarkar and R.~M. Karp.
\newblock An efficient approximation scheme for the one-dimensional bin-packing
  problem.
\newblock In {\em 23rd Annual Symposium on Foundations of Computer Science
  (sfcs 1982)}, pages 312--320. IEEE, 1982.

\bibitem{0015106}
J.~M. Kleinberg and {\'{E}}.~Tardos.
\newblock {\em Algorithm design}.
\newblock Addison-Wesley, 2006.

\bibitem{Koivisto09}
M.~Koivisto.
\newblock Partitioning into sets of bounded cardinality.
\newblock In J.~Chen and F.~V. Fomin, editors, {\em Parameterized and Exact
  Computation, 4th International Workshop, {IWPEC} 2009, Copenhagen, Denmark,
  September 10-11, 2009, Revised Selected Papers}, volume 5917 of {\em Lecture
  Notes in Computer Science}, pages 258--263. Springer, 2009.

\bibitem{krauthgamer_et_al:LIPIcs:2019:10284}
R.~Krauthgamer and O.~Trabelsi.
\newblock {The Set Cover Conjecture and Subgraph Isomorphism with a Tree
  Pattern}.
\newblock In R.~Niedermeier and C.~Paul, editors, {\em 36th International
  Symposium on Theoretical Aspects of Computer Science (STACS 2019)}, volume
  126 of {\em Leibniz International Proceedings in Informatics (LIPIcs)}, pages
  45:1--45:15, Dagstuhl, Germany, 2019. Schloss Dagstuhl--Leibniz-Zentrum fuer
  Informatik.

\bibitem{LenteLST13}
C.~Lent{\'{e}}, M.~Liedloff, A.~Soukhal, and V.~T'Kindt.
\newblock {O}n an extension of the {S}ort {\&} {S}earch method with application
  to scheduling theory.
\newblock {\em Theor. Comput. Sci.}, 511:13--22, 2013.

\bibitem{littlewood1943number}
J.~E. Littlewood and A.~C. Offord.
\newblock On the number of real roots of a random algebraic equation.
\newblock {\em Journal of the London Mathematical Society}, s1-13(4):288--295,
  1938.

\bibitem{knapsack-book-martello}
S.~Martello and P.~Toth.
\newblock {\em Knapsack Problems: Algorithms and Computer Implementations}.
\newblock Wiley Series in Discrete Mathematics and Optimization. Wiley, 1990.

\bibitem{MekaNV16}
R.~Meka, O.~Nguyen, and V.~Vu.
\newblock Anti-concentration for polynomials of independent random variables.
\newblock {\em Theory Comput.}, 12(1):1--17, 2016.

\bibitem{MuchaNPW19}
M.~Mucha, J.~Nederlof, J.~Pawlewicz, and K.~W\k{e}grzycki.
\newblock {E}qual-{S}ubset-{S}um {F}aster {T}han the {M}eet-in-the-{M}iddle.
\newblock In M.~A. Bender, O.~Svensson, and G.~Herman, editors, {\em 27th
  Annual European Symposium on Algorithms, {ESA} 2019, September 9-11, 2019,
  Munich/Garching, Germany}, volume 144 of {\em LIPIcs}, pages 73:1--73:16.
  Schloss Dagstuhl - Leibniz-Zentrum f{\"{u}}r Informatik, 2019.

\bibitem{Nederlof16}
J.~Nederlof.
\newblock {F}inding {L}arge {S}et {C}overs {F}aster via the {R}epresentation
  {M}ethod.
\newblock In P.~Sankowski and C.~D. Zaroliagis, editors, {\em 24th Annual
  European Symposium on Algorithms, {ESA} 2016, August 22-24, 2016, Aarhus,
  Denmark}, volume~57 of {\em LIPIcs}, pages 69:1--69:15. Schloss Dagstuhl -
  Leibniz-Zentrum f{\"{u}}r Informatik, 2016.

\bibitem{NederlofLZ12}
J.~Nederlof, E.~J. van Leeuwen, and R.~van~der Zwaan.
\newblock {Reducing a Target Interval to a Few Exact Queries}.
\newblock In B.~Rovan, V.~Sassone, and P.~Widmayer, editors, {\em Mathematical
  Foundations of Computer Science 2012 - 37th International Symposium, {MFCS}
  2012, Bratislava, Slovakia, August 27-31, 2012. Proceedings}, volume 7464 of
  {\em Lecture Notes in Computer Science}, pages 718--727. Springer, 2012.

\bibitem{rothvoss13}
T.~{Rothvoß}.
\newblock Approximating bin packing within {$O(\log OPT \cdot \log \log OPT)$}
  bins.
\newblock In {\em 2013 IEEE 54th Annual Symposium on Foundations of Computer
  Science}, pages 20--29, 2013.

\bibitem{RUDELSON2008600}
M.~Rudelson and R.~Vershynin.
\newblock The {L}ittlewood–{O}fford problem and invertibility of random
  matrices.
\newblock {\em Advances in Mathematics}, 218(2):600 -- 633, 2008.

\bibitem{schlegelgrant}
C.~Schlegel and A.~Grant.
\newblock {\em Coordinated multiuser communications}.
\newblock Springer, 2006.

\bibitem{0020358}
T.~Tao and V.~H. Vu.
\newblock {\em Additive combinatorics}, volume 105 of {\em Cambridge studies in
  advanced mathematics}.
\newblock Cambridge University Press, 2007.

\bibitem{annals}
K.~Tikhomirov.
\newblock Singularity of random {B}ernoulli matrices.
\newblock {\em Annals of Mathematics}, 191(2):593--634, 2020.

\bibitem{Tilborg78}
H.~C.~A. van Tilborg.
\newblock An upper bound for codes in a two-access binary erasure channel
  (corresp.).
\newblock {\em {IEEE} Trans. Inf. Theory}, 24(1):112--116, 1978.

\bibitem{VVWW10}
V.~V. Williams and R.~R. Williams.
\newblock {Subcubic Equivalences Between Path, Matrix, and Triangle Problems}.
\newblock {\em J. {ACM}}, 65(5):27:1--27:38, 2018.

\bibitem{wiman2017improved}
M.~Wiman.
\newblock {I}mproved {C}onstructions of {U}nbalanced {U}niquely {D}ecodable
  {C}ode {P}airs, 2017.
\newblock Bachelor Thesis KTH.

\bibitem{zamir:LIPIcs.ICALP.2021.113}
O.~Zamir.
\newblock {{B}reaking the {$2^n$} {B}arrier for 5-{C}oloring and 6-{C}oloring}.
\newblock In N.~Bansal, E.~Merelli, and J.~Worrell, editors, {\em 48th
  International Colloquium on Automata, Languages, and Programming (ICALP
  2021)}, volume 198 of {\em Leibniz International Proceedings in Informatics
  (LIPIcs)}, pages 113:1--113:20, Dagstuhl, Germany, 2021. Schloss Dagstuhl --
  Leibniz-Zentrum f{\"u}r Informatik.

\end{thebibliography}

\begin{appendices}
\section{Computing the Number of Distinct Sums and Critical Pruner}
\label{sec:pruner}

\begin{lemma} \label{lem:computesums}
    Let $w: \setrng{n} \to \N $ be an item weight function. Then the set
    $w(2^{\setrng{n}})$ can be computed in time $\Oh(n \cdot |w(2^{\setrng{n}})|)$. 
\end{lemma}
\begin{proof}
    Algorithm is a simple dynamic programming procedure. For all $i \in \{0,\dots,n\}$ define the set $\DP_i$ as:
	\[ \DP_i =  \{w(X): X\subseteq \{0,\dots,i\}\}. \]
    Notice that $\DP_n = w(2^{\setrng{n}})$. We iterate over $i \in \{0,\ldots,n\}$ to compute these sets. In the base case we set $\DP_0 = \{0\}$. Then, for $i \in \{1,\ldots,n\}$ given $\DP_{i-1}$ we compute $\DP_{i}$ as follows:
	\[ \DP_i = \DP_{i-1}\cup \{ x + w(i) \mid x \in \DP_{i-1}\}.  \]
    Note that in a single iteration each item in $\DP_{i-1}$ is touched at most twice.
	Hence, the total number of arithmetic operations can be upper bounded by $\sum_{i=0}^n 2 \cdot |\DP_i| = \Oh(n \cdot |\DP_n|)$.
\end{proof}

\begin{corollary} \label{cor:critPruner}
    Let $\delta \in (0,1)$ be a fixed parameter. If $|w(2^{\setrng{n}})| \ge 2^{\delta n}$, then the critical pruner $\prun(\delta)$ can be computed in time $\Os(2^{\delta n})$. 
\end{corollary}
\begin{proof}
    Recall the following definitions. Let $l = 1+ \lceil \log(\max_i \{w(i)\})
    \rceil $. For $s \in \{0,\dots,l\}$, the $s$-pruned weight of item $i$ is
    $w_s(i) \coloneqq \lfloor w(i) / 2^{l-s} \rfloor.$ The critical pruner,
    $\prun$, is $  \prun(\delta) = \min\{s \in \N \mid | w_s(2^{\setrng{n}})
    |\ge 2^{\delta n}  \} $. Notice that we can assume $l = \poly(n)$ by
    \cite{DBLP:journals/combinatorica/FrankT87}.
	
    The algorithm finds $\prun$ by computing $|w_s(2^{\setrng{n}})|$ using
    Lemma~\ref{lem:computesums} for consecutive $s=1,2,\dots$ until
    $|w_s(2^{\setrng{n}})|\ge 2^{\delta n}$. At the end it returns this last $s$
    as $\prun$. Because $w_0(2^{\setrng{n}}) = \{0\}$ and
    Lemma~\ref{lem:relationps} tells us that
    $\left(\frac{1}{3n}\right)|w_s(2^{\setrng{n}})| \le
    |w_{s-1}(2^{\setrng{n}})|$ for any $s$, we know that
    $|w_\prun(2^{\setrng{n}})| = \Oh(2^{\delta n})$. Analogously, the algorithm
    takes $\Oh(n^2 \cdot 2^{\delta n})$ time per iteration and the number of iterations is
    at most $l$, which gives the claimed running time.
\end{proof}

\begin{lemma} \label{lem:relationps} 

Let $w: \setrng{n} \to \N $ be an item weight function and let $l = 1+\lceil \log(\max_i\{w(i)\})\rceil$ Then for all $s \in \{1,\dots,l\}$:
	\[\left(\frac{1}{3n}\right)|w_s(2^{\setrng{n}})| \le
	|w_{s-1}(2^{\setrng{n}})| \le \left (\frac{3n}{2} \right
	)|w_s(2^{\setrng{n}})|.
	\]
\end{lemma}
\begin{proof}
	Let $l = 1+\lceil \log(\max_i\{w(i)\})\rceil$. We are given $w_s(A)$ for all
	$A\in 2^{\setrng{n}}$. Observe, that we can bound the value of $w_{s-1}(A)$ by the following:
	
	\begin{align*}
		&&&\sum_{i \in A} (w(i)/2^{l-s+1} - 1 ) && \le  w_{s-1}(A) \le && \sum_{i\in A}w(i)/2^{l-s+1}\\
		\Rightarrow   &&& \frac{1}{2}\sum_{i \in A} ( w(i)/2^{l-s} - 2) && \le  w_{s-1}(A) \le&& \frac{1}{2}\sum_{i\in A} w(i)/2^{l-s}\\
		\Rightarrow   &&& \frac{1}{2}\sum_{i \in A} ( \lfloor w(i)/2^{l-s} \rfloor - 2) && \le  w_{s-1}(A) \le&&\frac{1}{2}\sum_{i\in A} (\lfloor w(i)/2^{l-s} \rfloor +1)\\
		\Rightarrow   &&& \frac{1}{2}(w_s(A) - 2n) && \le  w_{s-1}(A) \le&& \frac{1}{2}(w_s(A) +n)
	\end{align*}
	
	Hence, for each value in $w_s(2^{\setrng{n}})$, there are at most
	$\frac{3n}{2}$ values in  $w_{s-1}(2^{\setrng{n}})$, i.e.
	\[|w_{s-1}(2^{\setrng{n}})| \le \left (\frac{3n}{2} \right)
	|w_s(2^{\setrng{n}})|.\]
	
	Analogously, for a given $w_{s-1}(A)$ and any subset $A\in 2^{\setrng{n}}$ we can bound the value of $w_{s}(A)$ by the following:

	\begin{align*}
		&&&\sum_{i \in A} (w(i)/2^{l-s} - 1) && \le  w_{s}(A) \le&& \sum_{i\in A}w(i)/2^{l-s}\\
		\Rightarrow   &&& 2 \sum_{i \in A} \left ( w(i)/2^{l-s+1} - \frac{1}{2} \right ) && \le  w_{s}(A) \le&& 2\sum_{i\in A}  w(i)/2^{l-s+1} \\
		\Rightarrow   &&& 2 \sum_{i \in A} \left ( \lfloor w(i)/2^{l-s+1} \rfloor - \frac{1}{2} \right ) && \le  w_{s}(A) \le&& 2\sum_{i\in A} (\lfloor w(i)/2^{l-s+1} \rfloor +1)\\
		\Rightarrow   &&& 2 \left (w_{s-1}(A) - \frac{n}{2} \right ) && \le  w_{s}(A) \le&& 2(w_{s-1}(A) +n)
	\end{align*}
	
	Hence, for each value in $w_{s-1}(2^{\setrng{n}})$, there are at most $3n$
	values in  $w_{s}(2^{\setrng{n}})$, i.e.
	\[|w_{s-1}(2^{\setrng{n}})| \ge \left (\frac{1}{3n}\right)
	|w_s(2^{\setrng{n}})|.\]
\end{proof}

\section{Inequalities with Binomials and Entropy}
\label{sec:inequalities}

Let us start with the useful facts about binary entropy function.
\begin{displaymath}
    h(x) \coloneqq -x \log(x) - (1-x)\log(1-x).
\end{displaymath}
The first derivative of binary entropy is:
\begin{displaymath}
    h'(x) \coloneqq \log(1-x) - \log(x)
\end{displaymath}
The second derivative:
\begin{displaymath}
    h''(x) \coloneqq -\frac{1}{(\ln{2}) x(1-x)}
\end{displaymath}
and we will also need third derivative
\begin{displaymath}
    h'''(x) \coloneqq \frac{1-2x}{(\ln{2}) x^2(1-x)^2}
\end{displaymath}
Observe, that for $x \in [0,\frac{1}{2}]$ we have that $h'(x) \ge 0$, $h''(x) \le 0$
and $h'''(x) \ge 0$. From 4th derivative we will only need that $h^{(4)}(x) \le
0$ when $x \in [0,\frac{1}{2}]$.
Hence from Taylor expansion for $x \in [0,\frac{1}{2}]$ it holds that:
\begin{displaymath}
    h(x + \eps) \coloneqq h(x) + h'(x) \eps + \frac{h''(x)}{2} \eps^2 + \frac{h'''(x)}{6} \eps^3 + \Oh(\eps^4).
\end{displaymath}
If we assume, that $x \in [0,\frac{1}{2}]$ and $\eps \le \frac{3 h''(x)}{2 h'''(x))}$ then:
\begin{equation}
    \label{eq:taylor}
    h(x + \eps) \le h(x) + h'(x) \eps + \frac{h''(x)}{4} \eps^2.
\end{equation}
because $h^{(4)}(x) \le 0$ when $x \in [0,\frac{1}{2}]$.
\begin{lemma}[Theorem 2.2 from \cite{calabro2009exponential}] \label{lem:inequalEntropy}
	\begin{align*}
	\forall x \in [0,1]: & \qquad 1-4 \left (x-\frac{1}{2} \right)^2 \le h(x) \le 1 - \frac{2}{\ln (2)} \left(x - \frac{1}{2}\right)^2,\\
	\forall x \in [0,1]: & \qquad \frac{x}{2\log(\frac{6}{x})} \le h^{-1}(x) \le \frac{x}{\log\frac{1}{x}},
	\end{align*}
	where the inverse entropy function $h^{-1}: [0,1] \rightarrow [0,1]$ is the inverse of $h$ restricted to the interval $[0,\frac{1}{2}]$.
\end{lemma}

\begin{lemma} \label{lem:boundeddownset}
	Given $\cW \subseteq \{W \subseteq \setrng{n} : |W| \in [\frac{n}{2} \pm c\cdot n]\}$ for some $c \in [0,\frac{1}{2})$ with $|\cW|  \le 2^{(1- z)n}\poly(n)$ and $z \in (0,1)$ such that $\frac{z}{4\log(12/z)} > c$. 
	Then 
	\[|\dc\cW| + |\uc\cW| \le \Oh\left(2^{(1-\rho(z, c))n}\right)\]
	where $\rho(z,c) = \frac{2}{\ln(2)}\left(\frac{z}{4\log(12/z)}-c\right)^2$.
\end{lemma}

\begin{proof}
	We will bound $|\dc\cW|$. Take $\lambda = \frac{z}{4\log(12/z)}-c$ and note that by assumption $\lambda >0$. We can describe any $X \in \dc\cW$ either as a set in $\binom{\setrng{n}}{|X|}$ (if $|X| \le (\frac{1}{2} - \lambda)n$) or as an element $W \in \cW$ together with the items on which $W$ and $X$ differ (if $|X|\ge (\frac{1}{2} - \lambda)n$). In the latter case, $|W \setminus X|\le (\lambda + c)n$. This, together with the fact that $|X|$ can only take $n$ distinct values implies:
	\begin{align*}
		|\dc\cW| &\le n\cdot \binom{n}{(\frac{1}{2} - \lambda)n} + |\cW|\cdot \binom{n}{(\lambda + c)n}\\
		&\le n \cdot 2^{h(\frac{1}{2}- \lambda)n} +  \poly(n) \cdot 2^{(1 - z + h(\lambda + c))n}
	\end{align*}
Which means that:
\[\frac{\log{|\dc\cW|}}{n} \le \max\left\{h\left(\frac{1}{2}- \lambda\right), 1 - z + h(\lambda + c)\right\} + o(1) \]
First, we will show that $h\left(\frac{1}{2}- \lambda\right) \ge 1 - z + h(\lambda + c)$. 
Note that 
\[\lambda+c = \frac{z}{4\log(12/z)} = \frac{z/2}{2\log(6/(z/2))} \le h^{-1}(z/2)\] 
by Lemma~\ref{lem:inequalEntropy}. Hence, $h(\lambda + c) \le z/2$ as $h$ is monotonic on $[0,\frac{1}{2}]$ and $\lambda + c \le \frac{1}{2}$. Therefore:
\begin{align*}
	h\left(\frac{1}{2}- \lambda\right) &\ge 1 - 4 \lambda^2 \\
	& \ge 1 - 4 \left(\frac{z}{4 \log(12/z)} - c \right)^2 \\
	& \ge 1 - 4 \left(\frac{z^2}{16 \log^2(12/z)} \right) \\
	& \ge 1 - 4 \left(\frac{z}{16 \log^2(12)} \right) \\
	& \ge 1 - z/2  \\
	& \ge 1 - z + h(\lambda + c).
\end{align*}
Where the first inequality follows from Lemma~\ref{lem:inequalEntropy}, the second is because $z/\log^2(12/z)$ is an increasing function for $z <1 $ and the last because $h(\lambda + c) \le z/2$.
Finally, we know, using Lemma~\ref{lem:inequalEntropy} that $h\left(\frac{1}{2}- \lambda\right) < 1 - \frac{2}{\ln(2)} \lambda^2$. Here the inequality is strict since $\lambda >0$. Therefore, 
\[|\dc\cW| \le \Oh\left(2^{(1-\rho(z,c))n}\right),\] where $\rho(z, c) = \frac{2}{\ln(2)}\left(\frac{z}{4\log(12/z)}-c\right)^2$ and we use that the inequality on $h\left(\frac{1}{2}-\lambda\right)$ is strict to omit $\poly(n)$ factors.
The same argument establishes $|\uc\cW| \le \Oh\left(2^{(1-\rho(z,c))n}\right)$
\end{proof}

\begin{lemma}
    \label{binom:inequality}
    For every $\beta,\alpha,\gamma \in [0,\frac{1}{2}]$ it holds that :

    \begin{displaymath}
        \binom{\beta n}{\alpha \beta n - \gamma n}
        \binom{(1-\beta)n}{\alpha(1-\beta)n + \gamma n} \le \binom{n}{\alpha n}
        \cdot 2^{-f(\gamma,\alpha,\beta)n}.
    \end{displaymath}

    where

    \begin{displaymath}
        f(\gamma,\alpha,\beta) \coloneqq \begin{cases}
            \gamma^2 & \text{if } |\gamma| < \alpha(1-\alpha)\min\{\beta,(1-\beta)\} \\
            -\alpha^2 \log{2\alpha} & \text{otherwise}
        \end{cases}
    \end{displaymath}

\end{lemma}

\begin{proof}

    First, observe that when $|\gamma| > \alpha(1-\alpha)\min\{\beta,(1-\beta)\}$
    our expression is upper bounded by:

    \begin{displaymath}
        \binom{\beta n}{\alpha \beta n(1-\alpha)}
            \binom{(1-\beta)n}{\alpha(1-\alpha)(1-\beta)n} \le
            \binom{n}{\alpha(1-\alpha) n}
            .
    \end{displaymath}

    This however is bounded by $2^{h(\alpha(1-\alpha))n}$. Observe that $h(\alpha
        - \alpha^2) \le h(\alpha) - h'(\alpha)\alpha^2 = h(\alpha) - \alpha^2
        (\log(\alpha) - \log(1-\alpha)) \le h(\alpha) - \alpha^2 \log(2\alpha)$. Hence when $|\gamma|$ is large we upper bound our expression with:

    \begin{displaymath}
        \binom{n}{\alpha n} 2^{\alpha^2 \log(2\alpha) n}.
    \end{displaymath}

    Now, we consider the case of small $|\gamma|$. We upper bound the expression with binary entropy.

    \begin{displaymath}
        \binom{\beta n}{\alpha \beta n - \gamma n} \binom{(1-\beta)n}{\alpha(1-\beta)n + \gamma n} = 2^{n\left( 
                    \beta h(\alpha - \frac{\gamma}{\beta}) + (1-\beta)h(\alpha + \frac{\gamma}{1-\beta})
                \right)}
    \end{displaymath}
    
    Let us consider an exponent:
    
    \begin{displaymath}
        \beta h\left(\alpha - \frac{\gamma}{\beta}\right) + (1-\beta)h\left(\alpha + \frac{\gamma}{1-\beta}\right)
    \end{displaymath}

    We use Inequality~\ref{eq:taylor} with $x = \alpha$ and $\eps \coloneqq
    - \frac{\gamma}{\beta}$ for $h(\alpha - \gamma/\beta)$ and with $\eps \coloneqq
    \frac{\gamma}{1-\beta}$ for $h(\alpha + \gamma/(1-\beta))$.

    Observe that at the beginning we assumed that $|\gamma| \le \alpha(1-\alpha)
    \min\{\beta,(1-\beta)\}$ hence
    $|\frac{\gamma}{\beta}|$ and $|\frac{\gamma}{1-\beta}|$ are upper bounded
    by $|\frac{3 h''(\alpha)}{2 h'''(\alpha)}|$. So, by
    Inequality~\ref{eq:taylor}:

    \begin{displaymath}
        \beta h\left(\alpha - \frac{\gamma}{\beta}\right) + (1-\beta)h\left(\alpha
        + \frac{\gamma}{1-\beta}\right) \le h(\alpha) + \frac{h''(\alpha)}{4}
        \frac{\gamma^2}{\beta(1-\beta)}.
    \end{displaymath}

    Observe that the first order factors cancel out. Hence 

    \begin{displaymath}
        \binom{\beta n}{\alpha \beta n - \gamma n}
        \binom{(1-\beta)n}{\alpha(1-\beta)n + \gamma n} \le \binom{n}{\alpha n}
        2^{\frac{h''(\alpha) \gamma^2}{4\beta(1-\beta)}}
    \end{displaymath}

    Finally, observe that $h''(\alpha) < -1$ for all $\alpha \in [0,\frac{1}{2}]$ and
    $\frac{1}{\beta(1-\beta)} \ge 4$ for all $\beta \in [0,\frac{1}{2}]$ hence:

    \begin{displaymath}
        \binom{\beta n}{\alpha \beta n - \gamma n}
        \binom{(1-\beta)n}{\alpha(1-\beta)n + \gamma n} \le \binom{n}{\alpha n}
        2^{-\gamma^2 n}
    \end{displaymath}
\end{proof}

\begin{lemma}
    \label{lem:hatd-hatc}
    For all $k \in \nat $ and $\alpha\in [0,1]$ we have:

    \begin{displaymath}
        h(\Bin(k,\alpha)) \le h(\Bin(k+1))
        .
    \end{displaymath}
\end{lemma}

\begin{proof}
    Let us fix $k \in \nat$. Recall that $\Bin(k+1) = (\{0,\ldots,k+1\}, p(i))$
    and $\Bin(k,\alpha) = (\{0,\ldots,k+1\}, p_\alpha(i))$ where $p(i)
    = \binom{k+1}{i}\frac1{2^{k+1}}$ and $p_\alpha(i)
    =\binom{k}{i}\frac{(1-\alpha)}{2^{k}}+\binom{k}{i-1}\frac{\alpha}{2^{k}}$.
    Hence, we need to prove that for all $\alpha \in [0,1]$:

    \begin{displaymath}
        h(p_\alpha(0),\ldots p_\alpha(k+1)) \le h(p(0),\ldots,p(k+1))
        .
    \end{displaymath}
        
    Let us denote $\phi(\alpha) = h(p_\alpha(0),\ldots,p_\alpha(k+1))$. First,
    observe that $\phi(\frac{1}{2}) = h(p(0),\ldots,p(k+1))$ because $\binom{k+1}{i} =
    \binom{k}{i} + \binom{k}{i-1}$. Therefore we need to prove that for all
    $\alpha \in [0,1]$ it holds that:

    \begin{displaymath}
        \phi(\alpha) \le \phi\left(\frac{1}{2}\right)
        .
    \end{displaymath}

    Recall that the binary entropy of a multinomial is $h(a_0,\ldots,a_{k+1})
    = -a_0\log(a_0) - \ldots - a_{k+1} \log(a_{k+1})$ and $(x \ln(x))' = \ln(x) + 1$
    Observe that function $\phi(\alpha)$ is well defined
    for $\alpha=0$ and $\alpha=1$ as limits.
    Moreover $\phi(\alpha) \ge 0$ for all $\alpha\in [0,1]$.

    Now, we compute the first derivative.
    \begin{equation*}
        \phi'(\alpha) = -\frac1{2^k\ln2}
        \sum_i\left(\binom{k}{i-1}-\binom{k}{i}\right)\bigl(1+\ln\left(p_\alpha(i)\right) \bigr).
    \end{equation*}
    Because $\sum_i\binom{k}{i-1}=\sum_i\binom{k}{i}$ the first derivative simplifies to:
    \begin{equation*}
        \phi'(\alpha)
         = -\frac1{2^k\ln2}
         \sum_i\left(\binom{k}{i-1}-\binom{k}{i}\right)\ln\left(p_\alpha(i))\right)
    \end{equation*}

    Now the second derivative is
    \begin{equation*}
        \phi''(\alpha) =
        -\frac1{4^k\ln2}
        \sum_i\left(\binom{k}{i-1}-\binom{k}{i}\right)^2\cdot\frac1{p_\alpha(i)}\le 0,
    \end{equation*}

    thus $\phi(\alpha)$ is concave for all $\alpha \in [0,1]$.
    So in order to show that the $\phi(\alpha)$ function has exactly one maximum
    in $\alpha = \frac{1}{2}$ it is sufficient to show that $\phi'(\frac{1}{2})=0$.

    Let us rearrange the sum:
    \begin{align*}
        \phi'\left(\frac{1}{2}\right) 
        =\sum_i\left[\binom{k}{i-1}-\binom{k}{i}\right]\ln (p_{\frac{1}{2}}(i))&=
        \sum_i\binom{k}{i}\ln (p_{\frac{1}{2}}(i+1))-\sum_i\binom{k}{i}\ln (p_{\frac{1}{2}}(i))\\
        &=\sum_i\binom{k}{i}\ln\frac{p_{\frac{1}{2}}(i+1)}{p_{\frac{1}{2}}(i)}.
    \end{align*}
    Because
    \begin{equation*}
        p_{\frac{1}{2}}(i)=\frac1{2^{k+1}}\binom{k+1}{i}
    \end{equation*}
    we can simplify the fraction:
    \begin{equation*}
        \frac{p_{\frac{1}{2}}(i+1)}{p_{\frac{1}{2}}(i)}=
        \frac{\binom{k+1}{i+1}}{\binom{k+1}{i}}=
        \frac{k+1-i}{i+1},
    \end{equation*}
    thus
    \begin{align*}
        \phi'\left(\frac{1}{2}\right) &=\sum_i\binom{k}{i}\ln\left(\frac{k+1-i}{i+1}\right)=
        \sum_i\binom{k}{i}\ln(k+1-i)-\sum_i\binom{k}{i}\ln(i+1)\\
        &=\sum_i\binom{k}{i}\ln(k+1-i)-\sum_i\binom{k}{k-i}\ln(k-i+1)=0,
    \end{align*}
    which finishes the proof.
\end{proof}

\end{appendices}

\end{document}